\documentclass[a4paper,11pt]{article}
\usepackage[utf8]{inputenc}
\usepackage{geometry}
\usepackage{xcolor}
\usepackage{amsmath, array, amssymb, amsfonts,amsthm}
\usepackage{upgreek}
\usepackage{xfrac}
\usepackage[inline]{enumitem}
\usepackage[french, english]{babel}
\usepackage[all]{xy}
\usepackage{txfonts}  % good1
\usepackage{sectsty}
\sectionfont{\centering}
\subsectionfont{\centering}
\subsubsectionfont{\centering}
\usepackage{booktabs}
\usepackage{caption}
\usepackage{dsfont}
\usepackage{mathtools}
\usepackage{slashed}
\usepackage[hidelinks]{hyperref}
\usepackage{textcomp}
\usepackage{multicol}
\usepackage[title]{appendix}
\usepackage[square,numbers,sort,compress,semicolon,merge]{natbib}
\let\cite\citep % \cite is then \citep of natbib
\usepackage{hyperref}
\hypersetup{colorlinks=true, urlcolor=blue, citecolor=blue, linktoc=page}
\allowdisplaybreaks

\usepackage{tikz}
\usepackage{tikz-cd}
\usetikzlibrary{cd}
\tikzcdset{
arrow style=tikz,
diagrams={>={Straight Barb[scale=0.8]}}
}

\DeclareMathAlphabet{\mathpzc}{OT1}{pzc}{m}{it}

\makeatletter
\renewcommand*\env@matrix[1][\arraystretch]{%
  \edef\arraystretch{#1}%
  \hskip -\arraycolsep
  \let\@ifnextchar\new@ifnextchar
  \array{*\c@MaxMatrixCols c}}
\makeatother

\newcommand{\defeq}{\vcentcolon=}
\newcommand{\rdefeq}{=\vcentcolon}
\renewcommand\P{\mathcal{P}}
\newcommand\M{\mathcal{M}}

\newcommand\RR{\mathbb{R}}
\newcommand\CC{\mathbb{C}}
\newcommand\C{\mathcal{C}}
\renewcommand\1{\textbf{1}}
\newcommand\id{\textit{id}}
\newcommand\T{\mathcal{T}}

\renewcommand\H{\mathcal{H}}
\newcommand\A{\mathcal{A}}
\newcommand\E{\mathcal{E}}

\renewcommand\L{\mathcal{L}}

\newcommand\U{\mathcal{U}}
\newcommand\SO{\mathcal{SO}}
\newcommand\SL{\mathcal{SL}}
\newcommand\Q{\mathcal{Q}}
\newcommand\K{\mathcal{K}}

\newcommand\W{\mathcal{W}}
\newcommand\vphi{\varphi}

\newcommand\sC{\mathsf{C}}
\newcommand\sT{\mathsf{T}}
\newcommand\sW{\mathsf{W}}
\newcommand\sS{\mathsf{S}}
\newcommand\sM{\mathsf{M}}

\renewcommand\epsilon{\varepsilon}

\newcommand\rarrow{\rightarrow}

\newcommand\aut{\mathfrak{aut}}

\newcommand\LieG{\mathfrak{g}}

\newcommand\so{\mathfrak{so}}
\renewcommand\sl{\mathfrak{sl}}

\renewcommand\t{\tilde}

\renewcommand\b{\bar }
\newcommand\w{\wedge}
\renewcommand\d{\partial}
\newcommand\s{\sigma}

\renewcommand\-{^{-1}}
\newcommand\Ad{\text{Ad}}

\renewcommand\id{\text{id}}
\renewcommand\1{\mathds{1}}

\makeatletter

\newcommand{\Rmnum}[1]{\expandafter\@slowromancap\romannumeral #1@}
\makeatother

\makeatletter
\newcommand{\leqnomode}{\tagsleft@true\let\veqno\@@leqno}
\newcommand{\reqnomode}{\tagsleft@false\let\veqno\@@eqno}
\makeatother

\DeclareMathOperator{\Diff}{Diff}
\DeclareMathOperator{\Aut}{Aut}
\DeclareMathOperator{\Tr}{Tr}

\DeclareMathOperator{\im}{Im}
\newtheorem{thm}{Theorem}

\newtheorem{prop}[thm]{Proposition}
\theoremstyle{definition}

\begin{document}

%%%%%%%%%%%%%%%%%%%%%%%%%%
% standard LaTeX
\title{Twisted gauge fields}
\author{J. François${\,}^{a}$}
\date{}

\maketitle
\begin{center}
\vskip -0.8cm
\noindent
${}^a$ Service de Physique de l'Univers, Champs et Gravitation, Universit\'e de Mons -- UMONS\\
20 Place du Parc, B-7000 Mons, Belgique.
\end{center}
%%%%%%%%%%%%%%%%%%%%%%%%%%

\begin{abstract}
We propose a generalisation of the notion of associated bundles to a principal bundle constructed via group action cocycles rather than via mere representations of the structure group. We devise a notion of connection generalising Ehresmann connection on principal bundles, giving rise to the appropriate covariant derivative on sections of these twisted associated bundles (and on twisted tensorial forms). We study the action of the group of vertical automorphisms on the objects introduced (active gauge transformations). We also provide  the gluing properties of the local representatives (passive gauge transformations). The latter are generalised gauge fields: They satisfy the gauge principle of physics, but are of a different geometric nature than standard Yang-Mills fields. We also examine the conditions under which this new geometry coexists and mixes with the standard one. We show that (standard) conformal tractors and Penrose's twistors can be seen as  simple instances of this general picture. We also indicate that the twisted geometry arises naturally in the definition and study of anomalies in quantum gauge field theory. 
\end{abstract}


\vspace{1.5mm}

%%%%%%%%%%% MAIN TEXT %%%%%%%%%%%%%%%%%

\tableofcontents

\section{Introduction}  %%%%%%%%%%%%%%%%%%%%%%%%%%%%%%%%%%%%%%%%%%%%%%%%%%

Classical gauge field theory is founded on the differential geometry of connections on fibered spaces: Ehresmann (principal) connections  underlies Yang-Mills type theories, relevant for particle physics, while Cartan connections are the foundation for gauge theories of gravitation. 
%Let us very briefly remind how is it so. % for the reader who is not already a consummate expert.
Since Wigner, it is admitted that given a symmetry (either spatio-temporal of internal) identified as  a Lie group $H$, different kinds of (fundamental) particules correspond to different (irreducible) representations $(\rho, V)$ of $H$. But particle are actually manifestations of (quantized) fields. 
So,  one considers a principal bundle $\P(\M, H)$ over spacetime $\M$, to which are associated - for each representation - bundles $E$ whose sections $s: \M \rarrow E \in \Gamma(E)$ describe matter fields of different kinds (to be further quantized). 
%Connections on $\P$, of either Ehresmann or Cartan type, give Yang-Mills or gravitational gauge potentials on $\M$ and induce covariant differentiation on $\Gamma(E)$ that represents the minimal coupling of matter fields with gauge interactions. 
Ehresmann connections on $\P$ give Yang-Mills potentials on $\M$ and induce covariant differentiation on $\Gamma(E)$ that represents the minimal coupling of matter fields to gauge interactions. 
Similarly, Cartan connections on $\P$ give gravitational potentials on $\M$, and in many interesting cases also induce covariant differentiation representing the coupling of matter fields to gravity. 

A gauge field theory is specified by choosing a Lagrangian. This choice is constrained by a list of desiderata that might be long and in part guided by empirical data. But at the top of the list is that the Lagrangian should satisfy the gauge principle: it must be invariant (or quasi-invariant) under local (point dependent) transformations of the field variables. These gauge transformations are induced by the group of vertical automorphisms of the underlying principal bundle $\Aut_v(\P)$, which is a normal subgroup of its group of automorphisms $\Aut(\P)$, itself the subgroup of $\Diff(\P)$ that preserve the fibration  structure and projects as diffeomorphisms of spacetime.  We have the short exact sequence, $\Aut_v(\P) \xrightarrow{\iota} \Aut(\P) \xrightarrow{\pi} \Diff(\M)$. The gauge principle should then be understood as a direct, though abstract, extension of the principle of general covariance at the heart of General Relativity (GR), and the requirement of gauge invariance as a rather natural generalisation of the requirement of diffeomorphism invariance. 
 Our most successful theories of fundamental physics, from GR to the Standard Model (SM), are gauge field theories of this kind. 
%\medskip

Various generalisations, some far reaching, of the above differential geometric framework with their associated notions of connection have been proposed,  sometimes with physical motivations, often with relevant physical applications. The rise to popularity of supersymmetry and supergravity e.g. inspired the study of differential super-geometry, super-bundles and  super-connections \cite{Rogers2007}. The development of derivation-based noncommutative geometry (NCG) opened the possibility to give a geometrico-algebraic interpretation to the potential of the  the scalar field in the electroweak model when the latter is unified with the gauge potential in a noncommutative connection \cite{Dubois-Violette_kerner_Madore1990a, Dubois-Violette-Kerner-Madore1990b}. The same feat is made possible by defining a special class of connections on Lie algebroids \cite{Masson-Lazz, Masson-Lazz2013}, a framework  presenting the advantage of being closer to standard differential geometry now familiar to most. Famously, NCG à la Connes - using spectral triplets/actions - has ambitions matching its abstraction since it was advocated as having the ressources to  explain in a unified way various features of the SM as well as naturally incorporating GR with it \cite{Connes-Lott1991, Connes-Marcolli, Chamseddine-et-al2007}. See \cite{FLMasson} for a short review on  formulations of gauge theories, and \cite{vanSuijlekom} for application of NGC  in physics. 
\medskip

Here we put forward what we believe is an original  generalisation of  connections on principal bundles.
The main new ingredient is a cocycle for the action of the structure group $H$ on $\P$, from which a new notion of \emph{twisted} associated bundles is defined. A corresponding  notion of connection on $\P$ is needed, that generalises Ehresmann's and induces a good notion of covariant differentiation on sections of these twisted  bundles, and more generally on the space of twisted tensorial forms - whose subspace of degree $0$ is isomorphic with the space of twisted sections. 
We develop this picture  in extensive computational details. 
The expert differential geometer might find this repetitive, but we believe it benefits  the broader potential readership. 
  
  We also take care to verify that the construction is well-behaved under bundle morphisms (the construction is functorial), which means in particular  that there is a well-behaved right action of $\Aut_v(\P)$ on the new objects, defining their active gauge transformations. 
The local picture on $\M$ is also detailed, where the local representatives of the above global objects are seen as generalised - or twisted -  gauge fields. Indeed, they provide the means to naturally  implement the gauge principle, while being of a different geometric nature than Yang-Mills fields.

The conditions necessary for the above twisted geometry to coexist with the standard one are examined, and %at the risk of some repetitiveness 
we provide the complete explicit treatment of this \emph{mixed} geometry. In particular, mixed gauge fields appear in the local picture. 
The building of twisted/mixed gauge theories is then briefly sketched. 
We also attempt to identify a subclass of twisted/mixed connections that seems to be a reasonable generalisation of Cartan connections. 

Our proposal is   conservative, but it has relevant contacts with the literature both in mathematics and physics. We indeed indicate how conformal tractors and Penrose' local twistors can be seen as simple - and in a precise sense, degenerate - examples of the general framework to be developed here. As an aside, conformal gravity is shown to be an instance of mixed gauge theory hiding in plain sight. Finally, we point out that the twisted geometry we advocate appears naturally in the study of anomalies in quantum gauge field theory.

\section{Twisted functions and twisted associated vector bundles}%%%%%%%%%%%%%%%%%%%%%%%%%%%%%
\label{Twisted functions and twisted associated vector bundles}%%%%%%%%%%%%%%%%%%%%%%%%%%%%%%%

Given $\P$ a $H$-principal fiber bundle and $G$ a Lie group, consider the smooth map:
\begin{align}
\label{cocycle}
C : \P \times H &\rarrow G \notag \\
       (p, h) & \mapsto C_p(h), \qquad \text{s.t} \quad C_p(hh')=C_p(h) \,C_{ph}(h').  
\end{align}
This is known as a \emph{group action cocycle} \cite{Cartan-Eilenberg1956} (see also the mathematical literature concerned with abstract ergodic theory \cite{Zimmer1982, Zimmer1984}). %\footnote{A fact the author discovered only after the present paper was essentially completed.}
From the very definition follows:
\begin{align}
\begin{split}
C_p(h')=C_p(e)\,C_p(h') \quad &\rarrow \quad C_p(e)=e'. \\
C_p(h)=C_p(h)\,C_{ph}(e) \quad &\rarrow \quad C_{ph}(e)=e'.\\
C_p\big(hh\-\big)=C_p(e)=e'=C_p(h)\, C_{ph}\big(h\-\big) \quad &\rarrow \quad C_p(h)\-=C_{ph}\big(h\-\big).  \\
C_p\big(h\-h\big)=e'=C_p(h\-)\, C_{ph\-}(h) \quad &\rarrow \quad C_p\big(h\-\big) = C_{ph\-}(h)\-.
\end{split}
\end{align}
Notice that the defining relation \ref{cocycle} can be seen as an equivariance relation on $\P$,
\begin{align}
\label{HeqC}
%\begin{split}
C_{ph}(h')=C_p(h)\- C_p(hh'),   \quad \text{or} \quad
R^*_h C(h')= C(h)\- C(hh').  
%\end{split}
\end{align} 
%{\color{red} Write the infinitesimal version via Lie derivative. Use it to simplify proof top page 5.}

The  differential of this map is,
%\begin{align*}
$dC_{|(p, h)}=dC(h)_{|p} + dC_{p| h}: T_p\P \oplus T_hH \rarrow T_{C_p(h)}G$,
%\end{align*}
where $\ker dC(h)_{|p}=T_hH$ and $\ker dC_{p|h}=T_p\P$, with by definition:
\begin{align*}
dC(h)_{|p}(X_p)&=\tfrac{d}{dt}C_{\phi_t}(h) |_{t=0}, \qquad \phi_t \text{ the flow of } X\in \Gamma(T\P) \text{ and } \phi_{t=0}=p,\\
dC_{p|h}(Y_h)&=\tfrac{d}{dt} C_p(\vphi_t)|_{t=0},\qquad \vphi_t \text{ the flow of } Y\in \Gamma(TH) \text{ and } \vphi_{t=0}=h.
\end{align*}
Notice that $C_p(h)\-dC_{|(p,h)} : T_p\P \oplus T_hH \rarrow  T_{e'}G=\text{Lie}G$.

\subsection{Twisted functions} %%%%%%%%%%%%%%%%%%%%%%%%%%%%%%%%%%%%%%%%%%%%%%%%%%
\label{Twisted functions}%%%%%%%%%%%%%%%%%%%%%%%%%%%%%%%%%%%%%%%%%%%%%%%%%%%%%%%

Given a representation $(\rho, V)$ of $G$, define the space $\Omega_{\text{eq}}^0\big(\P, C(H)\big)$ of $V$-valued $C$-twisted equivariant smooth functions on $\P$ as:
\begin{align}
\label{equiv-function}
\vphi : \P &\rarrow V \qquad \qquad \text{s.t} \quad  R^*_h \vphi = \rho\left[ C(h)\-\right] \vphi.    \\[-2mm]
       p &\mapsto \vphi(p).  \notag
\end{align}
This is a well behaved space under the right action of $H$ on $\P$ since on the one hand $\vphi(phh')=\rho\left[ C_p(hh')\-\right] \vphi(p)$, and on the other hand:
\begin{align*}
\vphi(phh')=\rho\left[  C_{ph}(h')\-\right] \vphi(ph)&=\rho\left[  C_{ph}(h')\-\right] \rho\left[  C_{p}(h)\-\right] \vphi(p)\\
&=\rho\left[  C_{ph}(h')\-  C_p(h)\-\right] \vphi(p)=\rho\left[ \left( C_p(h) C_{ph}(h')\right)\-  \right] \vphi(p)=\rho\left[ C_p(hh')\-\right] \vphi(p).
\end{align*}

\subsection{Twisted associated vector bundles} %%%%%%%%%%%%%%%%%%%%%%%%%%%%%%%%%%%%%%%%%%%%%%%%%%
\label{Twisted associated vector bundles}%%%%%%%%%%%%%%%%%%%%%%%%%%%%%%%%%%%%%%%%%%%%%%%%%%%%%%%

We generalise the usual construction of associated bundles to $\P(\M, H)$ via representations of the structure group $H$.  
Define a right $H$-action on $\P \times V$ by,
\begin{align*}
    (\P \times V) \times H &\rarrow \P \times V \\
      \big( (p, v), h \big) & \mapsto \left( ph, \rho\big[C_p(h)\-\big]v \right)
\end{align*}
This is well defined since we have on the one hand $\left((p, v), hh'\right) \mapsto \left( phh', \rho\left[C_p(hh')\-\right]v\right) $, and on the other hand we find
$\left(\left( ph, \rho\left[C_p(h)\-\right]v \right)\!,\, h'\right) \mapsto  \left( (ph)h', \rho\left[C_{ph}(h')\-\right]\left(\rho\left[C_p(h)\-\right]v\right)\right) = \left( phh', \rho\left[\left(C_p(h)C_{ph}(h')\right)\-\right]v\right) $. Declare equivalent the paires in $\P\times H$ that are related by this right action,  note $\sim$ the equivalence relation, and $[p, v]$ the equivalence class of $(p, v)$. We have then the $C$-twisted associated bundle $E^C = \P \times_{\rho(C)} V \coloneqq \P \times V\,/\!\sim$. 
%{\color{red} Noter $\times_{\rho(C)}$ plutôt, so that rep still seen ? I use $\Ad_{C(H)}$ in Cartan... }
Denoting $\Gamma\big(E^C\big)$ the space of sections of the $C$-twisted associated bundle, as in the usual case we have the isomorphism $\Gamma\big(E^C\big) \simeq \Omega_{\text{eq}}^0\big(\P, C(H)\big)$. 

The question then arise as to find  natural differential operators on $\Gamma\big(E^C\big) \simeq \Omega_{\text{eq}}^0\big(\P, C(H)\big)$, and in particular a good notion of covariant differentiation. 
This requires the appropriate notion of connection.

\section{Twisted connection and covariant differentiation} %%%%%%%%%%%%%%%%%%%%%%%%%%%%%%%%%%%%%%%%%%%%%%%%%%
\label{Twisted connection and covariant differentiation} %%%%%%%%%%%%%%%%%%%%%%%%%%%%%%%%%%%%%%%%%%%%%%%%%%%%

\subsection{Vertical vector fields} %%%%%%%%%%%%%%%%%%%%%%%%%%%%%%%%%%%%%%%%%%%%%%%%%%
\label{Vertical vector fields}  %%%%%%%%%%%%%%%%%%%%%%%%%%%%%%%%%%%%%%%%%%%%%%%%%%

The space $\Gamma(T\P)$  of vector fields on $\P$ is a subalgebra of the derivations of smooth functions on $\P$, $\Omega^0(\P)$. The right action of $H$ on $\P$ gives the canonical space of vertical vector fields $\Gamma(V\P)$, that naturally act on any equivariant function of $\P$. In particular, for $X^v \in \Gamma(V\P)$, $X\in$ Lie$H$, and $\vphi \in \Omega_{\text{eq}}^0\big(\P, C(H)\big)$ we have:
\begin{align*}
\left(X^v \vphi \right)(p)=\tfrac{d}{d\tau} \vphi\left(pe^{\tau X}\right)\big|_{\tau=0}
=\tfrac{d}{d\tau} \, \rho\left[C_p\left(e^{\tau X}\right)\-\right]\vphi(p)\big|_{\tau=0}
=-\tfrac{d}{d\tau} \, \rho\left[C_p\left(e^{\tau X}\right) \right]\vphi(p)\big|_{\tau=0}
= -\rho_*\left[ \tfrac{d}{d\tau} C_p\left(e^{\tau X}\right)\big|_{\tau=0} \right]\vphi(p).
\end{align*}
Where $\rho_*$ is the induced representation of  Lie$G$. The relation $X^v(\vphi) = -\rho_*\left[ \tfrac{d}{d\tau}C \left(e^{\tau X}\right)\big|_{\tau=0} \right]\vphi$, is nothing but the infinitesimal $C$-equivariance of $\vphi$, eq \eqref{equiv-function}.
\bigskip

For later use, we also derive an identity flowing from the action of $[X^v, Y^v]$ and $[X, Y]^v$ on $\vphi$. First, we simply have  $[X, Y]^v \vphi=-\rho_*\left[\tfrac{d}{dt} C_p\left(e^{t [X, Y]} \right)\big|_{t=0}\right]\vphi$. Then, 
\begin{align*}
\left(X^v (Y^v \vphi)\right) (p) &= \left(- X^v \tfrac{d}{d\tau}\, \rho\left[  C\left( e^{\tau Y}\right) \right]\big|_{\tau=0} \vphi \right)(p)
					    =-\tfrac{d}{d\s}\tfrac{d}{d\tau}\, \rho\left[  C_{pe^{\s X}}\left( e^{\tau Y} \right)\right]  \vphi\left(pe^{\s X}\right) \big|_{\tau=0,\, \s=0}\, , \\
					    &= -\tfrac{d}{d\s}\tfrac{d}{d\tau}\, \rho\left[ C_p\left( e^{\s X} \right)\- C_p\left( e^{\s X}e^{\tau Y}\right) \right]  \rho\left [ C_p\left( e^{\s X} \right)\-\right] \big|_{\tau=0,\, \s=0}\,\vphi(p)\, ,\\
					    &= -\tfrac{d}{d\s}\, \rho\left[ C_p\left( e^{\s X} \right)\- \right] \tfrac{d}{d\tau}\, \rho \left[C_p\left( e^{\s X}e^{\tau Y}\right) \right] \big|_{\tau=0} \, \rho\left[ C_p\left( e^{\s X} \right)\-\right] \big|_{\s=0} \, \vphi(p)\, ,\\
					    &= \left(   -\tfrac{d}{d\s}\, \rho\left[C_p\left( e^{\s X} \right)\- \right]\big|_{\s=0} \, \tfrac{d}{d\tau}\, \rho\left[ C_p\left( e^{\tau Y} \right) \right]\big|_{\tau=0}    -   \tfrac{d}{d\s}\tfrac{d}{d\tau}\, \rho \left[C\left( e^{\s X}e^{\tau Y}\right)\right] \big|_{\s=0\, \tau=0}  \right.  \\
					    &\hspace{7cm} \left. -       \tfrac{d}{d\tau}\, \rho\left[ C_p\left( e^{\tau Y} \right)\right]\big|_{\tau=0}\, \tfrac{d}{d\s} \,\rho\left[C_p\left( e^{\s X}\right)\-\right]\big|_{\s=0}       \right) \vphi(p).
\end{align*}

\noindent By the same process
\begin{align*}
\left(Y^v (X^v \vphi)\right) (p) &=\left(   -\tfrac{d}{d\tau}\, \rho\left[C_p\left( e^{\tau Y} \right)\- \right]\big|_{\tau=0} \, \tfrac{d}{d\s}\, \rho\left[ C_p\left( e^{\s X} \right) \right]\big|_{\s=0}    -  \tfrac{d}{d\tau}\tfrac{d}{d\s}\, \rho \left[C\left( e^{\tau Y}e^{\s X}\right)\right] \big|_{\s=0\, \tau=0}  \right.  \\
					 &\hspace{7cm}    \left. -       \tfrac{d}{d\s} \, \rho\left[ C_p\left( e^{\s X} \right)\right]\big|_{\s=0}\, \tfrac{d}{d\tau} \,\rho\left[C_p\left( e^{\tau Y}\right)\-\right]\big|_{\tau=0}       \right) \vphi(p).
\end{align*}
Combining the two results, and since $[X, Y]^v=[X^v, Y^v]$, we obtain the identity,
\begin{align}
\label{Id1}
\left[X, Y\right]^v\vphi(p) = \left[ X^v, Y^v\right]\vphi(p) \quad \Rightarrow \quad
\tfrac{d}{dt}C_p\left(e^{t[X, Y]} \right)\big|_{t=0} =  \tfrac{d}{d\s} \tfrac{d}{d\tau} \left( C_p\left(e^{\s X}e^{\tau Y} \right) - C_p\left(e^{\tau Y} e^{\s X}\right) \right)\big|_{\s=0,\, \tau=0} .
\end{align}

\smallskip

\subsection{Connection} %%%%%%%%%%%%%%%%%%%%%%%%%%%%%%%%%%%%%%%
\label{Connection} %%%%%%%%%%%%%%%%%%%%%%%%%%%%%%%%%%%%%%%%%%

In the usual case, a covariant derivative on sections $\Gamma(E)\simeq \Omega_{\text{eq}}^0(P, V)$ of a bundle $E$ associated to $\P$ is obtained through the definition of an horizontal distribution $H\P$ that complements the  canonical vertical distribution $V\P$ on $\P$, so that $\forall p \in \P$, $T_p\P= V_p\P \oplus H_p\P$ and $R_{h*} H_p\P=H_{ph}\P$. This is Ehresmann's far reaching notion of a connection. Most often, such an horizontal distribution is defined as the kernel of a connection $1$-form $\omega \in \Omega^1(\P, \text{Lie}H)$ ($\forall p \in \P, H_p\P\defeq \ker\omega_p$) satisfying the algebraic axioms - equivalent to the previous geometric axioms: $\omega_p(X_p^v)=X$, for $X \in \text{Lie}H$, and $R^*_h\omega_{ph}=\Ad_{h\-} \omega_p$. The covariant derivative of a section $\vphi \in \Omega_{\text{eq}}^0(P, V)$ along an horizontal vector field $X^h \in \Gamma(H\P)$ is then $X^h \vphi$, and satisfies $R^*_h X^h\vphi= \rho(h\-) X^h \vphi$ , so that $X^h\vphi \in \Omega_{\text{eq}}^0(P, V)$.% as expected. 

The curvature of the connection form is defined as $\Omega(X,Y) \defeq d\omega(X^h, Y^h)$, for $X,Y \in \Gamma(T\P)$. It is clearly a tensorial $2$-form ($\Ad$-equivariant and horizontal): $\Omega \in \Omega_{\text{tens}}^2(\P, \Ad)$. It also happens to satisfy Cartan's structure equation: $\Omega=d\omega+\tfrac{1}{2}[\omega, \omega]$. The curvature is the motivating example for defining the exterior covariant derivative on equivariant forms, $D : \Omega^\bullet_{\text{eq}}(\P, \rho) \rarrow \Omega^\bullet_{\text{tens}}(\P, \rho)$.
Denoting the horizontal projection by $|^h : \Gamma(T\P) \rarrow \Gamma(H\P)$, one indeed defines  $D\defeq d\circ |^h$, so that on the one hand $\Omega = D \omega$, and on the other hand $X^h\vphi=D\vphi(X)$. The exterior covariant derivative thus generalises the covariant derivative of sections. It happens that on tensorial forms, it can be expressed in terms of the connection $1$-form : $D\alpha =d\alpha + \rho_*(\omega)\alpha$, for $\alpha \in \Omega^\bullet_{\text{tens}}(P, \rho)$. 
\bigskip

It turns out that to define a notion of covariant differentiation on sections of $C$-twisted associated bundles - or $C$-twisted equivariant functions -  it is unnecessary to define an horizontal distribution. %\footnote{Problematic even. See appendix.}  
All that is required is the appropriate notion of connection $1$-form.%\footnote{{\color{gray}Notice that Cartan connections, which Ehresmann connections generalise, have no non-trivial kernel. So the notion of horizontal distribution is also bypassed, although it can be recovered in the important special class of reductive Cartan geometries. See \cite{Sharpe}.}}

Let us propose a Lie$G$-valued $1$-form $\omega \in \Omega^1(\P, LieG)$ satisfying:
\begin{align}
\label{1st axiom}
\omega_p(X_p^v)=\tfrac{d}{d\tau} C_p\left( e^{\tau X} \right)\big|_{\tau=0}=dC_{p|e}(X),       \qquad \text{for } X_p^v \in V_p\P \text{ generated by } X\in \text{Lie}H.                       \tag{\Rmnum{1}}
\end{align}
Given that $R_{*h} X_p^v =\tfrac{d}{d\tau} pe^{\tau X}h\big|_{\tau=0}=\tfrac{d}{d\tau} phh\-e^{\tau X}h\big|_{\tau=0}=\tfrac{d}{d\tau} phe^{\tau h\-Xh}\big|_{\tau=0}\rdefeq \left(\Ad_{h\-}X\right)^v_{ph}$, we deduce the equivariance of $\omega$ on $V\P$ :
\begin{align*}
\left( R^*_h \omega_{ph}\right) \left(X_p^v\right)&=\omega_{ph}\left( R_{h*}X_p^v \right)= \omega_{ph}\left(\left(\Ad_{h\-}X \right)_{ph}^v\right) 
									  = \tfrac{d}{d\tau} C_{ph}\left( e^{\tau \Ad_{h\-}X}\right)\big|_{\tau=0}
									  =\tfrac{d}{d\tau} C_{ph}\left( h\-e^{\tau X}h\right)\big|_{\tau=0}\, , \\
									&= \tfrac{d}{d\tau} C_{ph}\left(h\- \right) C_p\left( e^{\tau X}h \right)\big|_{\tau=0}
									 =C_p\left(h \right)\- \tfrac{d}{d\tau} C_p\left(e^{\tau X} \right) C_{pe^{\tau X}}\left( h\right)\big|_{\tau=0} \, ,  \\
									&= C_p\left(h \right)\-    \bigg(  \underbrace{\tfrac{d}{d\tau} C_p\left( e^{\tau X}\right)\big|_{\tau=0}}_{\omega_p\left(X_p^v \right)} C_p\left(h \right)  +   \underbrace{\tfrac{d}{d\tau} C_{pe^{\tau X}}(h)}_{\left(X^v C(h)\right)(p)}  \bigg)\, ,\\
									 &=\left(C_p\left(h \right)\-  \omega_p\, C_p(h) + C_p\left(h \right)\- dC(h)_{|p} \right) \left( X_p^v \right),
\end{align*}
%where the last equality stems from rewritting $\left(X^v C(h)\right)(p)$ as $dC(h)_p \left(X_p^v\right)$. 
We extend this result by requiring that it holds on the full tangent bundle of $\P$, i.e. $\omega$ has prescribed equivariance:
\begin{align}
\label{2nd axiom}
R^*_h \omega_{ph}\ = C_p\left(h \right)\-  \omega_p\, C_p(h) + C_p\left(h \right)\- dC(h)_{|p} .	 		\tag{\Rmnum{2}}
\end{align}
This prescription is well-behaved  with respect to (w.r.t) the right action of $H$ on $\P$. Indeed,  for $h,h' \in H$ we have the composition of pullbacks $R^*_{h'} \circ R^*_h= (R_h \circ R_{h'})^*= R_{h'h}^*$. So that on the one hand $R^*_{h'} \left( R^*_h\omega\right)= R^*_{h'h}\omega= C(h'h)\-\omega\, C(h'h) + C(h'h)\-d C(h'h)$. On the other hand, a direct computation using \eqref{HeqC} gives:
\begin{align*}
R^*_{h'} \left( R^*_h\omega\right) &= R^*_{h'}\left( C(h)\- \omega\, C(h) + C(h)\-dC(h) \right), \\
						&= C(h'h)\- C(h') \left( C(h')\-\omega\, C(h') + C(h')\-d C(h')\right) C(h')\- C(h'h) + C(h'h)\- C(h')\,d \left(C(h')\-C(h'h) \right), \\
						&= C(h'h)\- \omega\, C(h'h) + C(h'h)\-dC(h')\cdot C(h')\- C(h'h) + C(h'h)\-C(h')\, dC(h')\- \cdot C(h'h) + C(h'h)\-dC(h'h), \\
						&=C(h'h)\- \omega\, C(h'h)  + C(h'h)\-dC(h'h)= R^*_{h'h}\omega.
\end{align*}
%By the way, \eqref{2nd axiom} is also consistent with the first axiom since, first:
%\begin{align*}
%R^*_h \omega_{ph} (X^v_p)&= \omega_{ph}\left(R_{*h}X^v_p \right)= \omega_{ph}\left( \Ad_{h\-}X\big|_{ph}^v \right)= \tfrac{d}{d\tau}C_{ph}\left( e^{\tau h\-Xh} \right)\big|_{\tau=0}, \\
%					  &= \tfrac{d}{d\tau} C_{ph}\left( h\- e^{\tau X}h \right)\big|_{\tau=0}= \tfrac{d}{d\tau} C_{ph}\left( h\-\right) C_p\left( e^{\tau X}h \right)\big|_{\tau=0}
%					  = C_p(h)\- \tfrac{d}{d\tau} C_p\left( e^{\tau X}h \right)\big|_{\tau=0}.
%\end{align*}
%And second:
%\begin{align*}
%\left( C_p\left(h \right)\-  \omega_p C_p(h) + C_p\left(h \right)\- dC(h)_p \right) (X_p^v) 
%   			&= C_p\left(h \right)\-  \left(  \tfrac{d}{d\tau} C_p\left( e^{\tau X}\right)\big|_{\tau=0} C_p(h) +\tfrac{d}{d\tau} C_{pe^{\tau X}}(h) \big|_{\tau=0}  \right), \\
%   			&= C_p\left(h \right)\-  \left(  \tfrac{d}{d\tau} C_p(e^{\tau X}) C_{pe^{\tau X}}(h) \big|_{\tau=0}  \right)= C_p(h)\- \tfrac{d}{d\tau} C_p\left( e^{\tau X}h \right)\big|_{\tau=0}.
%\end{align*}
The axioms \eqref{1st axiom} and \eqref{2nd axiom} define our notion of \emph{twisted connection}. 
Let us denote $\C(\P)^T$ the space of these connections.

\medskip
 
For later use, and because it is a result in its own right, we here give the infinitesimal version of the  equivariance law \eqref{2nd axiom} of $\omega$. Given $X^v \in \Gamma(V\P)$, we have:
\begin{align}
\label{Id2}
\left(L_{X^v}\omega \right)_p&\defeq \tfrac{d}{d\tau} R^*_{e^{\tau X}} \omega_{pe^{\tau X}} \big|_{\tau=0}
                          =\tfrac{d}{d\tau} \left( C_p\left(e^{\tau X}\right)\- \omega_p \, C_p\left(e^{\tau X}\right) +  C_p\left(e^{\tau X}\right)\- d C\left( e^{\tau X}\right)_p  \right)\big|_{\tau=0}\, ,  \notag\\
			&= -\underbrace{\tfrac{d}{d\tau} C_p\left( e^{\tau X}\right)\big|_{\tau=0}}_{dC_{p|e}(X))} \omega_p + \omega_p\, \tfrac{d}{d\tau} C_p\left( e^{\tau X}\right)\big|_{\tau=0}  - \tfrac{d}{d\tau} C_p\left(e^{\tau X}\right)\big|_{\tau=0} \underbrace{dC(e)_p}_{=0} +\, d\left( \tfrac{d}{d\tau} C_p\left( e^{\tau X} \right)\big|_{\tau=0} \right), \notag \\
L_{X^v}\omega &= d\big(dC_{|e}(X)\big)+ \left[\omega, dC_{|e}(X) \right].   \tag{\Rmnum{2}.b}
\end{align}

\reqnomode

\subsection{Covariant differentiation} %%%%%%%%%%%%%%%%%%%%%%%%%%%%%%%%%%%%%%%%%%%%%
\label{Covariant differentiation} %%%%%%%%%%%%%%%%%%%%%%%%%%%%%%%%%%%%%%%%%%%%%%%%

Consider the space of $C$-equivariant differential forms $\Omega^\bullet_{\text{eq}}\big(\P, C(H)\big)\!\defeq\! \left\{ \alpha\in \Omega^\bullet(\P, V) \,|\,R^*_h \alpha_{ph}=\rho \left[C_p\left( h\right)\- \right] \alpha_p\right\}$. The infinitesimal version of this equivariance property is given by the Lie derivative along a vertical vector field:
\begin{align}
\label{LieDerAlpha}
L_{X^v}\alpha = \tfrac{d}{d\tau} R^*_{e^{\tau X}} \alpha\big|_{\tau=0} = \tfrac{d}{d\tau} \rho\left[ C\left(e^{\tau X}\right)\-\right] \alpha\big|_{\tau=0} = -\rho_* \left[ \tfrac{d}{d\tau} C\left( e^{\tau X}\right) \big|_{\tau=0} \right] \alpha.
\end{align}
Elements of the subspace of tensorial $C$-equivariant forms $\Omega^\bullet_{\text{tens}}\big(\P, C(H)\big)$ further satisfies $\alpha(X^v, Y, \ldots)=0$ for $X^v \in \Gamma(V\P)$.
Sections of a $C$-twisted associated bundles are also $C$-tensorial $0$-forms, $\vphi \in \Omega_{\text{tens}}^0\big(\P, C(H)\big)=\Omega_{\text{eq}}^0\big(\P, C(H)\big)$.

\begin{prop} %%%%%%%%%%%%%%%%%%%%%%%%%%%%%%%%%%%%%%%%%%%%%%%
\label{CovDiff}
The exterior covariant derivative  defined as $D\defeq d\, +\! \rho_*(\omega)$ preserves both $ \Omega_{\text{eq}}\big(\P, C(H)\big)$ and $ \Omega_{\text{tens}}\big(\P, C(H)\big)$.
\end{prop}
\begin{proof} %%%%%%%%%%%%%%%%%%%%%%%%%%%%%%%%%%%%%%%%%%%%%%
First, we show that  $D:  \Omega^\bullet_{\text{eq}}\big(\P, C(H)\big) \rarrow \Omega^\bullet_{\text{eq}}\big(\P, C(H)\big)$. For $\alpha \in \Omega^\bullet_{\text{eq}}\big(\P, C(H)\big)$:
\begin{align*}
R^*_h D\alpha &= d R^*_h \alpha + \rho_*\left(R^*_h \omega \right) R^*_h \alpha, \\
		       &= d \rho\left[C_p(h)\-\right]  \cdot \alpha + \rho\left[C_p(h)\-\right]d \alpha + \rho_*\left( C(h)\- \omega\, C(h) + C(h)\-dC(h) \right) \rho\left[C(h)\-\right]\alpha,\\
		       &= \rho\left[C(h)\-\right] \big(d\alpha + \rho_*(\omega)\alpha \big)=\rho\left[C(h)\-\right] D\alpha.
\end{align*}
So indeed $D\alpha \in \Omega^\bullet_{\text{eq}}\big(\P, C(H)\big)$. Here we used the second defining property \eqref{2nd axiom} of $\omega$.

Then we show that $D:  \Omega^\bullet_{\text{tens}}\big(\P, C(H)\big) \rarrow \Omega^\bullet_{\text{tens}}\big(\P, C(H)\big)$. It is enough to prove it for $\alpha \in \Omega^1_{\text{tens}}\big(\P, C(H)\big)$:
\begin{align*}
D\alpha(X^v, Y) &= \left(d\alpha+\rho_*(\omega)\alpha \right)(X^v, Y), \\
		     % &= X^v\cdot \omega(Y) - Y\cdot \omega(X^v) - \omega([X^v, Y]) + \rho_*(\omega(X^v))\alpha(Y) - \rho_*(\omega(Y))\underbrace{\alpha(X^v)}_{=0}.
		         &= \underbrace{d\alpha(X^v, Y)}_{\left( L_{X^v}\alpha \right)(Y)}+ \rho_*(\omega(X^v))\alpha(Y) - \rho_*(\omega(Y))\underbrace{\alpha(X^v)}_{=0}, \\
		         &=-\rho_* \left[ \tfrac{d}{d\tau} C\left( e^{\tau X}\right) \big|_{\tau=0} \right] \alpha(Y) + \rho_* \left[ \tfrac{d}{d\tau} C\left( e^{\tau X}\right) \big|_{\tau=0} \right] \alpha(Y)=0.
\end{align*}
We used Cartan's magic formula $L_X=i_X d+ di_X$ and the first defining property \eqref{1st axiom} of $\omega$.
\end{proof}

This operator provide the adequate notion of covariant differentiation on $\Gamma\big(E^C\big)\simeq \Omega_{\text{tens}}^0\big(\P, C(H)\big)$. Indeed for any section $\vphi$, its covariant derivative is $D\vphi= d\vphi +\rho_*(\omega)\vphi \in \Omega_{\text{tens}}^1\big(\P, C(H)\big)$.

\subsection{Curvature} %%%%%%%%%%%%%%%%%%%%%%%%%%%%%%%%%%%%%%%%%%%%%
\label{Curvature} %%%%%%%%%%%%%%%%%%%%%%%%%%%%%%%%%%%%%%%%%%%%%%%%

The curvature $2$-form of the twisted connection is defined via Cartan's structure equation: $\Omega\defeq d\omega + \tfrac{1}{2}[\omega, \omega]$. It then identically satisfies the Bianchi identity: $d\Omega + [\omega, \Omega]=0$.

\begin{prop}%%%%%%%%%%%%%%%%%%%%%%%%%%%%%%%%%%%%%%%%%%%%%%%
\label{curvature}
The curvature is a $C$-tensorial $2$-form, $\Omega \in \Omega_{\text{tens}}^2\big(\P, C(H)\big)$.
\end{prop}
\begin{proof}%%%%%%%%%%%%%%%%%%%%%%%%%%%%%%%%%%%%%%%%%%%%%%%
We begin by proving, using \eqref{2nd axiom}, that $\Omega \in \Omega^2_{\text{eq}}\big(\P, C(H)\big)$:
\begin{align*}
R^*_h \Omega &= dR^*_h \omega + \tfrac{1}{2}[R^*_h \omega, R^*_h\omega], \\
		        &= d\left( C(h)\-\omega\, C(h) + C(h)\-dC(h) \right) + \tfrac{1}{2}\left[C(h)\-\omega\, C(h) + C(h)\-dC(h) , C(h)\-\omega\, C(h) + C(h)\-dC(h)  \right], \\
		        &= dC(h)\- \omega\, C(h) + C(h)\- d\omega\, C(h) - C(h)\-\omega dC(h) + dC(h)\-dC(h)  \\
		        & \qquad + \tfrac{1}{2} C(h)\- [\omega, \omega] C(h) + [C(h)\-\omega\, C(h), C(h)\- dC(h)] + \tfrac{1}{2}[C(h)\- dC(h),  C(h)\- dC(h)], \\
		        &=  C(h)\-\left( d\omega + \tfrac{1}{2} [\omega, \omega]\right) C(h),\\
		        &= C(h)\- \Omega \, C(h).
\end{align*}
The curvature $\Omega$ of the twisted connection $\omega$ is thus an $\Ad_{C(H)}$-equivariant $2$-form.
Then we prove that $\Omega$ is tensorial. Firstly:
\begin{align*}
\Omega(X^v, Y^v)&= d\omega (X^v, Y^v) + [\omega(X^v), \omega(Y^v)], \\
			    &=X^v\cdot \omega(Y^v) - Y^v\cdot \omega(X^v) - \omega([X^v, Y^v]) + [\omega(X^v), \omega(Y^v)].
\end{align*}
Now, using \eqref{1st axiom}:
\begin{align*}
    X^v\cdot \omega(Y^v) &= X^v \left( \tfrac{d}{d\tau} C_p\left(e^{\tau Y}\right)\big|_{\tau=0}\right) = \tfrac{d}{d\s} \tfrac{d}{d\tau} C_{pe^{\s X}}\left( e^{\tau Y}\right)\big|_{\tau=0,\, \s=0}
            			         = \tfrac{d}{d\s} \tfrac{d}{d\tau}  C_p\left( e^{\s X}\right)\- C_p\left( e^{\s X}e^{\tau Y}\right) \big|_{\tau=0,\, \s=0}, \\
			               &= \tfrac{d}{d\s}  \left( C_p\left( e^{\s X}\right)\- \tfrac{d}{d\tau} C_p\left(  e^{\s X}e^{\tau Y}\right)\big|_{\tau =0}  \right)\big|_{\s=0}, \\
			               &= -\tfrac{d}{d\s} C_p\left( e^{\s X}\right)\big|_{\s=0} \, \tfrac{d}{d\tau} C_p\left( e^{\tau Y}\right)\big|_{\tau =0}  +  \tfrac{d}{d\s} \tfrac{d}{d\tau} C_p\left( e^{\s X}e^{\tau Y}\right)\big|_{\tau=0,\, \s=0}
\end{align*}
Idem: $ -Y^v\cdot \omega(X^v) = \tfrac{d}{d\tau} C_p\left( e^{\tau Y}\right)\big|_{\tau=0} \, \tfrac{d}{d\s} C_p\left( e^{\s X}\right)\big|_{\s =0}  -  \tfrac{d}{d\tau} \tfrac{d}{d\s} C_p\left( e^{\tau Y}e^{\s X}\right)\big|_{\tau=0,\, \s=0}$. But then we have,
\begin{align*}
-\omega([X^v, Y^v])&=-\omega([X, Y]^v)= -\tfrac{d}{dt}C_p\left( e^{t [X, Y]} \right)\big|_{t=0} \\[1mm]
\text{and }\quad [\omega(X^v), \omega(Y^v)]&= \left[\tfrac{d}{d\s} C_p\left( e^{\s X}\right)\big|_{\s=0}, \tfrac{d}{d\tau} C_p\left( e^{\tau Y}\right)\big|_{\tau =0} \right].
\end{align*}
Finally, by using the identity \eqref{Id1} we see by inspection that $\Omega(X^v, Y^v)=0$. Secondly:
\begin{align*}
\Omega(X^v, Y)&= X^v\cdot \omega(Y) - Y\cdot \omega(X^v) - \omega([X^v, Y]) + [\omega(X^v), \omega(Y)], \\
			&= X^v\cdot \omega(Y)- \omega([X^v, Y])  -\bigg(d\omega(X^v)+ [\omega, \omega(X^v)] \bigg)(Y).
\end{align*}
Now, 
\begin{align*}
\left( L_{X^v}\omega \right) (Y) &= \big( (i_{X^v}d +di_{X^v})\omega \big)(Y)=d\omega(X^v, Y) + d(\omega(X^v))(Y), \\
					       &= X^v\cdot \omega(Y) - Y\cdot \omega(X^v) - \omega([X^v, Y]) + Y\cdot \omega(X^v)= X^v\cdot \omega(Y) - \omega([X^v, Y]).
\end{align*}
But also, by \eqref{Id2} $L_{X^v}\omega = d\omega(X^v) +[\omega, \omega(X^v)]$. So, by inspection we have that $\Omega(X^v, Y)=0$. Which finishes to demonstrate that $\Omega \in \Omega^2_{\text{tens}}\big(\P, C(H)\big)$.
\end{proof}
This fact then allows to see that the Bianchi identity can be written as $D\Omega=0$.
It is by the way easy to show that, as usual,  $D^2 \alpha=DD\alpha=\rho_*(\Omega)\alpha$.

%\clearpage

\section{Functoriality} %%%%%%%%%%%%%%%%%%%%%%%%%%%%%%%%%%%%%%%%%%%%%
\label{Functoriality} %%%%%%%%%%%%%%%%%%%%%%%%%%%%%%%%%%%%%%%%%%%%%%%

In this section we check that the twisted objects defined above enjoy the same functoriality as  standard constructions, i.e. that principal bundle morphisms induce morphisms of associated twisted bundles, and of spaces of twisted connections and  twisted equivariant forms. Of particular interest is the special case of vertical automorphisms of a principal bundle, giving rise to (active) gauge transformations.

\subsection{Naturality under bundle morphisms} %%%%%%%%%%%%%%%%%%%%%%%%%%%%%%%%%%%%%%%%%%%%%
\label{Naturality under bundle morphisms} %%%%%%%%%%%%%%%%%%%%%%%%%%%%%%%%%%%%%%%%%%%%%%%

Consider a $H$-bundle (iso)morphism $\phi: \P \rarrow \P'$, with $\phi(ph)=\phi(p)h$, which induces a smooth map (diffeomorphism) $\b\phi:\M \rarrow \M'$. Given  $C'\!: \P' \times H \rarrow G$, a $H$-cocycle on $\P'$, $C\!:=\phi^*C'\! : \P \times H \rarrow G$ is a $H$-cocycle on $\P$. 
Indeed, $C_p(h):=C'_{\phi(p)}(h)$, so $C_p(hh')=C'_{\phi(p)}(hh')=C'_{\phi(p)}(h) C'_{\phi(p)h}(h')=C'_{\phi(p)}(h) C'_{\phi(ph)}(h')=:C_p(h)C_{ph}(h')$. 

So, $H$-bundle morphisms induce morphisms of twisted bundles in the following way (we omit the representation $\rho$ for simplicity):
\begin{align*}
\t\phi : E^C &\rarrow {E'}^{C'}, \\
           [p, v] &\mapsto \t\phi([p, v]):=[\phi(p), v], \\
           [ph, C_p(h)\-v] &\mapsto \t\phi([ph,C_p(h)\-v] ):=[\phi(ph), C_p(h)\-v]=:[\phi(p)h, C'_{\phi(p)}(h)\- v ].
  %         [ph, \rho\left[C_p(h)\-\right]v] &\mapsto \t\phi([[ph, \rho\left[C_p(h)\-\right]v] ]):=[\phi(ph), \rho\left[C_p(h)\-\right]v]=:[\phi(p)h, \rho\left[C'_{\phi(p)}(h)\-\right] v ].
\end{align*}
Naturally this implies the existence of a morphism of spaces of sections. Indeed, given $\vphi' \in \Omega_\text{eq}^0(\P', C'(H))$, a twisted function on $\P'$ such that $R^*_h\vphi'={C'(h)}\-\vphi'$, the map $\vphi:=\phi^*\vphi'$ is such that $\vphi(ph):=\vphi'(\phi(ph))=\vphi'(\phi(p)h)={C'_{\phi(p)}(h)}\- \vphi'(\phi(p))=:C_p(h)\-\vphi(p)$. We thus have the morphism $\phi^*\!: \Omega_\text{eq}^0(\P', C'(H)) \rarrow \Omega_\text{eq}^0(\P, C(H))$, which, by the isomorphism mentioned in section \ref{Twisted associated vector bundles}, induces the morphism ${\t\phi}^*\!: \Gamma({E'}^{C'}) \rarrow \Gamma(E^C)$.

More generally,  $\phi$ induces a morphism of spaces of $C$-equivariant forms, $\phi^*\!: \Omega_\text{eq}^\bullet(\P', C'(H)) \rarrow \Omega_\text{eq}^\bullet(\P, C(H))$. Indeed, given  $\alpha' \in \Omega_\text{eq}^\bullet(\P', C'(H))$ and $R_h\circ \phi=\phi\circ R_h$, the form $\alpha:=\phi^*\alpha'$ is such that (still omitting $\rho$) $R^*_h \alpha:= R^*_h \phi^*\alpha'=\phi^*R^*_h \alpha' = \phi^* \left( C'(h)\- \alpha'\right)=\phi^*C'(h)\- \phi^*\alpha'=: C(h)\- \alpha$. So,  $\alpha \in \Omega_\text{eq}^\bullet(\P, C(H))$. Also, since for $X^v_p \in V_p\P$ we have $\phi_*X^v_p=X^v_{\phi(p)} \in V_{\phi(p)}\P'$,  pullback by $\phi$ preserves horizontality and the above morphism restricts to the spaces of $C$-tensorial forms, $\phi^*\!: \Omega_\text{tens}^\bullet(\P', C'(H)) \rarrow \Omega_\text{tens}^\bullet(\P, C(H))$.
\bigskip

One can further show that there are induced morphisms of  spaces of twisted connections $\phi^*\!: \C(\P')^T \rarrow \C(\P)^T$. If $\P'$ is endowed with a twisted connection $\omega'\in \C(\P')^T$, which therefore satisfies, for $q\in \P'$, $\omega'_q(X^v_q)=dC_{q|e}(X)$, $X \in$ Lie$H$, and 
$R^*_h\omega'_{qh}=C'_q(h)\- \omega'_q C'_q(h)+ C'_q(h)\- dC'(h)_{|q}$, 
%$R^*_h\omega'=C'(h)\- \omega' C'(h)+ C'(h)\- dC'(h)$, 
then $\omega:=\phi^*\omega'$ satisfies on the one hand, 
\begin{align*}
\omega_p(X^v_p):=\phi^*\omega'_{\phi(p)}(X_p^v)=\omega'_{\phi(p)}(\phi_*X^v_p)=\omega'_{\phi(p)}(X^v_{\phi(p)})= dC'_{\phi(p)|e}(X)=:dC_{p|e}(X),
\end{align*}
and on the other hand, for $X_p \in T_p\P$, 
\begin{align*}
%R^*_h \omega:=R^*_h(\phi^*\omega')=\phi^*(R^*_h\omega')=\phi^*\left(C'(h)\-\omega'C'(h)+C'(h)\-dC'(h)\right)=:C(h)\-\omega\, C(h)+C(h)\-dC(h).
R^*_h \omega_{ph}(X_p)&=\omega_{ph}(R_{h*}X_p):=\phi^*\omega'_{\phi(ph)}(R_{h*}X_p)=\omega'_{\phi(p)h}(\phi_*R_{h*}X_p)=\omega'_{\phi(p)h}(R_{h*}\phi_*X_p)=R_h^*\omega'_{\phi(p)h}(\phi_*X_p),\\
                                      &=\left(C'_{\phi(p)}(h)\- \omega'_{\phi(p)}C'_{\phi(p)}(h) + C'_{\phi(p)}(h)\- dC'(h)_{|\phi(p)}\right)(\phi_*X_p), \\
                                      &=: C_p(h)\- \phi^*\omega'_{\phi(p)}(X_p)\, C_p(h) + C_p(h)\- d(C'(h)\circ\phi)_{|p}(X_p), \\
                                      &=\left(  C_p(h)\- \omega_p\, C_p(h) + C_p(h)\- dC(h)_{|p}  \right)(X_p).
\end{align*}
It is then a twisted connection on $\P$,  $\omega \in \C(\P)^T$.

 From the above we obtain readily that a covariant derivative $D'\!=d \, +\rho_*(\omega')$ on $\P'$ pulls-back as a covariant derivative $D:=\phi^*D'$ on $\P$. In particular it means that $H$-bundle morphisms $\phi$ induce morphisms of  twisted associated bundles equiped with covariant derivatives, $\t\phi : (E^C, D) \rarrow ({E'}^{C'}, D')$.
\bigskip

If $\M=\M'$ and $\b\phi=\id_\M$, then $\phi$ is a bundle equivalence, and we have established equivalence of the associated twisted structures. The above functoriality holds also in the special case $\P=\P'$ with $\phi \in \Aut(\P)$ covering $\b\phi\in \Diff(\M)$,  and therefore in the case $\phi \in \Aut_v(\P)$ covering $\b\phi=\id_\M$. The latter is of particular interest for physical application to gauge theories.

%\clearpage

\subsection{Action of vertical automorphisms and gauge transformations} %%%%%%%%%%%%%%%%%%%%%%%%%%%%%%%%%%%%%%%%%%%%%
\label{Action of vertical automorphisms and gauge transformations} %%%%%%%%%%%%%%%%%%%%%%%%%%%%%%%%%%%%%%%%%%%%%%%%

The group of vertical automorphisms of the principal bundle $\P$ is  a subgroup of its group of diffeomorphisms, $\Aut_v(\P)\defeq \bigg\{ \Phi \in \Diff(\P)\, |\, \Phi(ph)=\Phi(p)h \text{ for } h\in H, \text{ and } \pi \circ\Phi=\pi \bigg\}$. It acts on itself by composition of maps. The gauge group is defined as $\H\defeq \left\{ \gamma: \P \rarrow H\, |\, R^*_h\gamma=h\-\gamma h \right\}$. Both group are isomorphic by the identification $\Phi(p)=p\gamma(p)$.
The group of vertical automorphisms  acts by pullback on $\Omega^\bullet(\P)$. A pullback by $\Phi \in \Aut_v(\P)$ will then equivalently be called an active gauge transformation by $\gamma \in \H$. 
Now, for  $\Psi \in \Aut_v(\P)$ associated  to the elements $\eta \in \H$, we have: 
$\Psi^*\gamma(p)=\gamma(\Psi(p))=\gamma(p\eta(p))=\eta(p)\-\gamma(p) \eta(p)$. So the action of the gauge group $\H$ on itself, noted $\gamma^\eta$, is  defined as: $\gamma^\eta\defeq\Psi^*\gamma=\eta\- \gamma \eta$. It reflects the defining equivariance of its elements.

To give the gauge transformations of the objects defined in the previous sections, we need to first consider the smooth map,
\begin{align*}
C(\gamma): \P  &\rarrow  G, \\
		    p	&\mapsto C_p\left(\gamma(p)\right).
\end{align*} 
Its equivariance is, 
\begin{align}
R^*_h C(\gamma)(p)&= C_{ph}\left( \gamma(ph) \right)= C_{ph}\left(  h\-\gamma(p)h  \right) = C_{ph}\big(h\-\big) C_p\left( \gamma(p)h \right)
=C_p(h)\- C_p(\gamma(p)) C_{p\gamma(p)}(h),  \notag\\[1.5mm]
R^*_h C(\gamma) & =C(h)\-C(\gamma h) = C(h)\- C(\gamma)\,\Phi^*C(h).   \label{HequivCgamma}
\end{align}
%{\color{red} Give the Lie derivative version and use it to give an alternative derivation for Ia. }
Correspondingly we have,
\begin{align}
\Psi^*C(\gamma)(p)&=C_{\Psi(p)}\big( \gamma\left(\Psi(p)\right)\big)=C_{p\eta(p)}\left(\eta(p)\- \gamma(p) \eta(p)\right)
					     %=C_{p\eta(p)}\left( \eta(p)\- \right) C_p\bigg( \gamma(p)\eta(p) \bigg)
					        =C_p\left( \eta(p) \right)\- C_p\big( \gamma(p)\eta(p) \big).  \notag\\[1.5mm]
C(\gamma)^\eta\defeq \Psi^*C(\gamma)&= C(\eta)\- C(\gamma\eta).   \label{GTCgamma}
\end{align}
%{\color{red} Give this before, when equivariance of $C(\gamma)$ is given, label it and use it here.}
This map is given by the composition,
\begin{align*}
\P   &\xrightarrow{\ \Delta\, }   \P \times \P    \xrightarrow{\ \id \times \gamma\, }    \P\times H   	\xrightarrow{\quad C\quad }  G,\\
p    & \xmapsto{\phantom{bbb}}    (p, p)	 	       \xmapsto{\phantom {bbbb}}     (p, \gamma(p))  \xmapsto{\phantom {bbbb}} 	C_p(\gamma(p)).     
\end{align*} 
So its differential $dC(\gamma)_{|p}: T_p\P \rarrow T_{C_p\left(\gamma(p)\right)}G$ is,
\begin{align*}
T_p\P   &\xrightarrow{\ d\Delta\, }   T_p\P \times T_p\P    \xrightarrow{\ \id \times d\gamma\, }    T_p\P\times T_{\gamma(p)}G   	\xrightarrow{\quad dC(\gamma(p))_{|p} \,+\, dC_{p|\gamma(p)}\quad }  T_{C_p(\gamma(p))}G,\\
X_p    & \xmapsto{\phantom{bbbbb}}    (X_p, X_p)  	  \xmapsto{\phantom {bbbbbbbb}}       \bigg(X_p, \underbrace{d\gamma_{|p}(X_p)}_{\left[X(\gamma)\right](p)}\bigg)           \xmapsto{\phantom {bbbbbbbbbbbb}}          dC(\gamma(p))_{|p}(X_p)\, +\,  \underbrace{dC_{p|\gamma(p)}\left( d\gamma_p(X_p) \right)}_{dC_p(\gamma)_{|p}(X_p)}.     
\end{align*} 
So, if $\phi_\tau$  is the flow of X with $\phi_{\tau=0}=p$, we have:
\begin{align}
dC(\gamma)_{|p}(X_p)= dC(\gamma(p))_{|p}(X_p) + dC_p(\gamma)_{|p}(X_p)= \tfrac{d}{d\tau}\left\{ C_{\phi_\tau}(\gamma(p)) + C_p\left(\gamma(\phi_\tau)\right)   \right\}\big|_{\tau=0}.
\end{align}
Notice then that $C_p\!\left( \gamma(p) \right)\-\!dC(\gamma)_{|p}: T_p\P \rarrow T_{e'}G=\text{Lie}G$.
We are then ready to state and prove the following two propositions.

\begin{prop} %%%%%%%%%%%%%%%%%%%%%%%%%%%%%%%%%%%%%%%%%%%%%%%%%%%%%%%%%%
\label{ConnectionGT}
The active gauge transformation of a twisted connection, noted $\omega^\gamma$, is  %{\color{red}  [Introduce $\omega^{C(\gamma)}$ ?]  }
\begin{align}
\label{H-GTconnection}
\omega^\gamma  \defeq  \Phi^*\omega  =  C(\gamma)\-\omega\, C(\gamma) + C(\gamma)\-dC(\gamma) \ \in \C(\P)^T.  
\end{align}
\end{prop}
\begin{proof} %%%%%%%%%%%%%%%%%%%%
%\paragraph{Connection} %%%%%%%%%%%%%%%%%%%%%%%%%%%%%%%%%%%%%%%%%%%%%%%%%%%%%
Given the standard result  $\Phi_*X_p = R_{\gamma(p)*} X_p + [\gamma\-d\gamma]_{|p} (X_p)\big|^v_{\Phi(p)} $ for $X \in \Gamma(T\P)$, we have:
\begin{align*}
\left(  \Phi^*\omega \right)_p(X_p) = \omega_{\Phi(p)}\left( \Phi_*X_p \right)=\omega_{\Phi(p)}\left( R_{\gamma(p)*} X_p + \gamma(p)\-d\gamma_{|p} (X_p)\big|^v_{\Phi(p)}\right).
\end{align*}
The first term is easily worked out, by \eqref{2nd axiom}:
\begin{align*}
R^*_{\gamma(p)}\omega_{p\gamma(p)} (X_p) =\left( C_p\left(\gamma(p) \right)\- \omega_p\, C_p\left( \gamma(p)\right) + C_p\left(\gamma(p) \right)\- dC\left( \gamma(p) \right)_{|p} \right) (X_p)
\end{align*}
The second term needs special attention. By \eqref{1st axiom} we have:
\begin{align*}
\omega_{\Phi(p)}\left(  \gamma(p)\-d\gamma_{|p} (X_p)\big|^v_{\Phi(p)} \right)
&= \tfrac{d}{d\tau} C_{p\gamma(p)}\left( e^{\tau \, \gamma(p)\-d\gamma_{|p}(X_p)} \right)\big|_{\tau=0}
= C_p\left( \gamma(p) \right)\- \tfrac{d}{d\tau} C_{p}\left( \gamma(p)e^{\tau \, \gamma(p)\-d\gamma_{|p}(X_p)} \right)\big|_{\tau=0},\\
%\end{align*}
%Then we should notice that \eqref{1st axiom} can be rewritten as 
%\begin{align*}
%\omega_p(X^v_p)=\tfrac{d}{d\tau} C_p\left( e^{\tau X} \right)\big|_{\tau=0}= dC_{p|e}\left( \tfrac{d}{d\tau} e^{\tau X} \big|_{\tau =0} \right)=dC_{p|e}(X) \, \in T_{C_p(e)G}=\LieG, \text{ for } X\in \LieH.
%\end{align*}
%So, 
%\begin{align*}
%\omega_{\Phi(p)}\left(  \gamma(p)\-d\gamma_{|p} (X_p)\big|^v_{\Phi(p)} \right) 
&=  C_p\left( \gamma(p) \right)\- dC_{p|\gamma(p)} \left(  \gamma(p) \tfrac{d}{d\tau} e^{\tau\, \gamma(p)\-d\gamma_{|p}(X_p)} \big|_{\tau =0}  \right) %, \\
= C_p\left( \gamma(p) \right)\- dC_{p|\gamma(p)} \left( d\gamma_{|p} (X_p)  \right), \\
&= C_p\left( \gamma(p) \right)\-  dC_p(\gamma)_{|p} (X_p).
\end{align*}
Finally,  we then obtain,
\begin{align*}
\left(  \Phi^*\omega \right)_p(X_p) &= \left( C_p\left(\gamma(p) \right)\- \omega_p\, C_p\left( \gamma(p)\right) + C_p\left(\gamma(p) \right)\- \left( dC\left( \gamma(p) \right)_{|p} +  dC_p(\gamma)_{|p} \right) \right)(X_p), \\
&=\left( C_p\left(\gamma(p) \right)\- \omega_p\, C_p\left( \gamma(p)\right) + C_p\left(\gamma(p) \right)\-  dC( \gamma )_{|p}  \right)(X_p).
\end{align*}
%The active gauge transformation of the connection is then written as,
%\begin{align}
%\label{GTconnection}
%\omega^\gamma  \defeq  \Phi^*\omega  =  C(\gamma)\-\omega\, C(\gamma) + C(\gamma)\-dC(\gamma).  
%\end{align}
By the way, from section \ref{Naturality under bundle morphisms} above, we have $\Phi^*\!: \C(\P)^T \rarrow \C(\P)^T$. So indeed $\omega^\gamma \in \C(\P)^T$.
\end{proof}

\begin{prop}%%%%%%%%%%%%%%%%%%%%%%%%%%%%%%%%% 
\label{TensorialGT}
The active gauge transformation of a $C$-tensorial form, noted $\alpha^\gamma$, is
\begin{align}
\label{GTtensorial}
\alpha^\gamma \defeq \Phi^*\alpha = \rho\left[ C(\gamma)\- \right] \alpha \ \in\Omega^\bullet_{\text{tens}}\big(\P, C(H)\big).
\end{align}
\end{prop}
\begin{proof}
Proceeding as above we get,
\begin{align*}
\left( \Phi^*\alpha \right)_p(X_p, \ldots ) &= \alpha_{\Phi(p)} \left( R_{\gamma(p)*}X_p + [\gamma\-d\gamma]_{|p}(X_p)\big|^v_{\Phi(p)} , \ldots \right), \\
&= \alpha_{\Phi(p)} \left( R_{\gamma(p)*}X_p, \ldots \right), \\
&= R^*_{\gamma(p)} \alpha_{\Phi(p)} (X_p, \ldots)= \rho\left[C_p\left( \gamma(p) \right)\-\right]\alpha_p (X_p, \ldots).
\end{align*}
Also, from section \ref{Naturality under bundle morphisms}  we have $\Phi^*\!:\Omega^\bullet_{\text{tens}}\big(\P, C(H)\big) \rarrow \Omega^\bullet_{\text{tens}}\big(\P, C(H)\big)$. So indeed $\alpha^\gamma \in \Omega^\bullet_{\text{tens}}\big(\P, C(H)\big)$.
\end{proof}
\noindent
From this follows the gauge transformations of the curvature, of sections and their covariant derivatives:
\begin{align}
\Omega^\gamma &\defeq	\Phi^*\Omega = C(\gamma)\- \Omega\, C(\gamma),    \label{GTcurv}  \\
\vphi^\gamma &\defeq \Phi^*\vphi = \rho\left[C(\gamma)\-\right]\vphi,    \label{GTsection}  \\
(D\vphi)^\gamma &\defeq \Phi^*D\vphi = \rho\left[C(\gamma)\-\right]D\vphi.    \label{GTcovder} 
\end{align}
Using  \eqref{H-GTconnection}, equation \eqref{GTcurv} is alternatively found from the Cartan structure equation by having $\Omega^\gamma=d\omega^\gamma+\tfrac{1}{2}[\omega^\gamma, \omega^\gamma]$, and equation \eqref{GTcovder} by having $(D\vphi)^\gamma=D^\gamma \vphi^\gamma= d\vphi^\gamma + \rho_*(\omega^\gamma)\vphi^\gamma$. 
\bigskip

Finally, we explicitly verify the following.
\begin{prop} %%%%%%%%%%%%%%%%%%%%%
\label{Rightaction}
The action of $\Aut_v(\P)\simeq \H$ on $\C(\P)^T$ and $\Omega^\bullet_{\text{tens}}\big(\P, C(H)\big)$ is  a right action.
\end{prop}
\begin{proof}%%%%%%%%%%%%%%%%%%%%%
Given $\Phi, \Psi \in \Aut_v(\P)$ associated respectively to the elements $\gamma, \eta \in \H$, and using \eqref{GTCgamma}
%\begin{align}
%\Psi^*C(\gamma)(p)&=C_{\Psi(p)}\left( \gamma\left(\Psi(p)\right)\right)=C_{p\eta(p)}\left(\eta(p)\- \gamma(p) \eta(p)\right), \notag\\
%					     &=C_{p\eta(p)}\left( \eta(p)\- \right) C_p\bigg( \gamma(p)\eta(p) \bigg)
%					        =C_p\left( \eta(p) \right)\- C_p\bigg( \gamma(p)\eta(p) \bigg).  \notag\\[1.5mm]
%C(\gamma)^\eta\defeq \Psi^*C(\gamma)&= C(\eta)\- C(\gamma\eta).   \label{GTCgamma}
%\end{align}
%%{\color{red} Give this before, when equivariance of $C(\gamma)$ is given, label it and use it here.}
%
%On a connection form we therefore have,
we have
\begin{align*}
\left( \omega^\gamma \right)^\eta \defeq \Psi^*\left( \Phi^*\omega \right)&= \Psi^* \left( C(\gamma)\-\omega\, C(\gamma) + C(\gamma)\-dC(\gamma) \right), \\
					      &= C(\gamma\eta)\- C(\eta) \left( C(\eta)\-\omega\, C(\eta) + C(\eta)\-dC(\eta)  \right)C(\eta)\- C(\gamma\eta)  +  C(\gamma\eta)\- C(\eta)\, d\left( C(\eta)\- C(\gamma\eta) \right), \\
					      &= C(\gamma\eta)\-\omega\, C(\gamma\eta) + C(\gamma\eta)\-dC(\gamma\eta).
\end{align*}
This shows the  consistency of the notation for active gauge transformations, in terms of which the above result is simply $\left( \omega^\gamma \right)^\eta=\omega^{\gamma\eta}$.
This extends easily to tensorial forms. For $\alpha \in \Omega^\bullet_{\text{tens}}\big(\P, C(H)\big)$:
\begin{align*}
\left( \alpha^\gamma \right)^\eta  \defeq \Psi^*\left( \Phi^*\alpha \right)= \Psi^* \left( \rho\left[ C(\gamma)\-\right]\alpha \right) 
= \rho\left[ C(\gamma\eta)\- C(\eta) \right]   \rho\left[ C(\eta)\-\right]\alpha =   \rho\left[ C(\gamma\eta)\-\right]\alpha.
\end{align*}
Which is simply $\left( \alpha^\gamma \right)^\eta=\alpha^{\gamma\eta}$. Further gauge transformations of $\Omega$, $\vphi$ and $D\vphi$ are special cases of this result.
\end{proof}
We thus have well-defined spaces $\C(\P)^T$, and  $\Omega^\bullet_{\text{tens}}\big(\P, C(H)\big)$ endowed with an exterior covariant derivative $D$,  with a consistent right action of the gauge group $\H \simeq \Aut_v(\P)$ of the underlying principal bundle $\P$.

\section{Local description} %%%%%%%%%%%%%%%%%%%%%%%%%%%%%%%%%%%%%%%%%%%%%
\label{Local description} %%%%%%%%%%%%%%%%%%%%%%%%%%%%%%%%%%%%%%%%%%%%%%%%

We now turn to the local description of the global objects just described. As a principal bundle, $\P$ is locally trivialisable, meaning that for any open subset $\U$ of the base manifold $\M$: $\P_{|\U}\simeq \U \times H$. Given a local trivialising section $\s:\U\rarrow \P_{|\U}$, any form $\beta \in \Omega^\bullet(\P)$ can be pulled-back as a form $b\!\defeq\!\s^*\beta \in \Omega(\U)$. Also, any vector $X \in \Gamma(T\U)$ can be pushed-forward as a vector $\s_*X \in \Gamma(T\P_{|\U})$. 

The pullbacks  of the connection and its curvature are respectively Lie$G$-valued $1$-form and $2$-form on $\U$. We denote $A\!\defeq\!\s^*\omega \in \Omega^1(\U, \text{Lie}G)$, and $F\!\defeq\!\s^*\Omega \in \Omega^2(\U, \text{Lie}G)$. By the naturality of the pullback, Cartan's structure equation still holds: $F=dA+\sfrac{1}{2}[A, A]$. The pullback of a section $\vphi \in \Omega^0_{\text{tens}}\big(\P, C(H)\big)\simeq \Gamma\big(E^C\big)$ is a $V$-valued map on $\U$ that we denote $\phi\!\defeq\!\s^*\vphi \in \Omega^0(\U, V)$. Still by naturality of the pullback we have that: $D\phi\!\defeq\!\s^*D\vphi \in \Omega^1(\U, V)$, with $D\phi=d\phi+\rho_*(A)\phi$. In general, let us denote the pullbacks of a $C$-tensorial form and its exterior covariant derivative $\alpha, D\alpha \in \Omega^\bullet_{\text{tens}}(\P, C(H))$ by $a\!\defeq\!\s^*\alpha, Da\!\defeq\!\s^*D\alpha  \in \Omega^\bullet(\U, V)$, with  $Da=da+\rho_*(A)a$.

If this framework would apply to (particle) physics - which happens on $\M$, describing space-time - $A$ would be a generalised/twisted gauge potential and $F$ would be its field strength, while $\phi$ would be a generalised/twisted matter field and $D\phi$ would describe its minimal coupling to $A$.

As forms on $\U\subset \M$, the local variables $A, F,  a$ and $Da$ do not have equivariance w.r.t the structure group $H$ of $\P$. Nevertheless, the fact that they are shadows of global objects shows both in their gluing properties from one open subset of $\M$ to another, and in the local version of their gauge transformations. These are discussed in the next two sections.

\subsection{Gluing properties: passive gauge transformations} %%%%%%%%%%%%%%%%%%%%%%%%%%%%%%%%%%%%%%%%%%%%%
\label{Gluing properties: passive gauge transformations} %%%%%%%%%%%%%%%%%%%%%%%%%%%%%%%%%%%%%%%%%%%%%%%%

Consider $\U, \U' \subset \M$ such that $\U \cap \U' \neq \emptyset$, endowed with  local sections $\s:\U\rarrow \P_{|\U}$ and $\s':\U'\rarrow \P_{|\U'}$. On the overlap, both sections are related as 
\begin{align*}
\s'=\s g, \quad \text{ with }\quad g:\U \cap \U' &\rarrow \, H, \\				%\s'(x)&=\s(x)g(x)
						          x \quad   &\mapsto \, g(x).
\end{align*}
It is a standard result that for $X_x \in T_x\M$, $x \in \U \cap \U'$, the pushforwards by $\s'$ and $\s$ are related by
\begin{align}
\label{Xpushforward}
\s'_* X_x = R_{g(x)*}\left( \s_* X_x \right)  + [ g\-dg]_{|x}(X_x) \big|^v_{\s'(x)},
\end{align}
with $\s'_* X_x \in T_{\s'(x)}\P_{|\U'}$ and $\s_* X_x \in T_{\s(x)}\P_{|\U}$. 

Furthermore, let us introduce the maps $C_\s(h)\!\defeq\!\s^*C(h):\U \rarrow G$, with $h \in H$, as well as $C_\s(g):\U \rarrow G$. Notice that, just like $C(\gamma)$ has a double dependence on $p \in \P$, $C_\s(g)$ has a double dependence on $x\in  \U \cap \U' \subset \M$. %For $x \in \U\cap \U'$, 
Their counterparts on $\U'$ are $C_{\s'}(h)\!\defeq\!\s'^*C(h):\U' \rarrow G$ and $C_{\s'}(g'):\U' \rarrow G$, and we have 
\begin{align}
C_{\s'(x)}(h)&=C_{\s(x)g(x)}(h)=C_{\s(x)}(g(x))\- C_{\s(x)}\big( g(x)h\big),  \notag\\
C_{\s'(x)}(g'(x))&=C_{\s(x)g(x)}\big(g'(x)\big)=C_{\s(x)}(g(x))\- C_{\s(x)}\left( g(x)g'(x)\right),   \label{Cs'-Cs}
\end{align}
We are ready to state the following,
\begin{prop}%%%%%%%%%%%%%%%%%%%%%%%%%%%%%%%%%%%%%%%%%
\label{passiveGTs}
The gluing properties of the local representatives of a twisted connection and a tensorial form are, 
\begin{align}
\label{passiveGTconnection}
A'&=C_\s(g)\- A\, C_\s(g) + C_\s(g)\-dC_\s(g), \\[1.5mm]
\label{passiveGTtensorial}
a'&=\rho\left[C_\s(g)\-\right] a.
\end{align}
\end{prop}
\begin{proof} %%%%%%%%%%%%%%%%%%%%%%%%%%%%%%%%%%%%%%%%%%%%%%%%%%%%%%%%%%%
 For the connection, using \eqref{Xpushforward} and \eqref{1st axiom}-\eqref{2nd axiom}, we have:
\begin{align*}
A'_x(X_x) &= \s'^*\omega_{\s'(x)}(X_x)= \omega_{\s'(x)}(\s'_*X_x)%,  \\
		=\omega_{\s'(x)} \left( R_{g(x)*}\left( \s_* X_x \right)  + [ g\-dg]_{|x}(X_x) \big|^v_{\s'(x)} \right), \\
		&= R^*_{g(x)} \omega_{\s'(x)} (\s_* X_x ) + \tfrac{d}{d\tau} C_{\s'(x)}\left( e^{\tau \, [g\-dg]_{|x}(X_x)} \right)\big|_{\tau =0},\\
		&=\left(   C_{\s(x)}( g(x) )\- \omega_{\s(x)} C_{\s(x)}( g(x) ) +  C_{\s(x)}( g(x) )\- d C( g(x))_{|\s(x)}   \right)(\s_* X_x) \\
		& \hspace{6cm}  + \tfrac{d}{d\tau} C_{\s(x)}( g(x) )\- C_{\s(x)} \left( g(x)\, e^{\tau\, g(x)\- dg_{|x}(X_x) }  \right)\big|_{\tau =0},  \\
		&= C_{\s(x)}( g(x) )\- \s^*\omega_{\s(x)}(X_x)\,  C_{\s(x)}( g(x) ) +  C_{\s(x)}( g(x) )\- d C_\s( g(x))_{|x} (X_x) \\
		& \hspace{6cm} + C_{\s(x)}( g(x) )\- dC_{\s(x)|\, g(x)} \bigg(\underbrace{g(x) \tfrac{d}{d\tau}  e^{\tau\, g(x)\- dg_{|x}(X_x) }\big|_{\tau =0}}_{dg_{|x}(X_x)} \bigg),  \\
		&= \bigg( C_{\s(x)}( g(x) )\- A_x\,  C_{\s(x)}( g(x) ) +  C_{\s(x)}( g(x) )\- \big(\underbrace{ d C_\s( g(x))_{|x}   +   dC_{\s(x)}(g)_{|x} }_{dC_\s(g)_{|x}}\big) \bigg)(X_x).
\end{align*}
Likewise for a tensorial form,
\begin{align*}
a'_x(X_x, \ldots) &= \s'^*\alpha_{\s'(x)}(X_x, \ldots)= \alpha_{\s'(x)}(\s'_*X_x, \ldots),  \\
		&=\alpha_{\s'(x)} \left( R_{g(x)*}\left( \s_* X_x\right)  + [ g\-dg]_{|x}(X_x) \big|^v_{\s'(x)}, \ldots \right), \\
		&=R^*_{g(x)} \alpha_{\s'(x)} (\s_* X_x, \ldots ) ,
		= \rho\left[ C_{\s(x)}(g(x)) \-\right] \alpha_{\s(x)}  (\s_* X_x, \ldots ), \\
		&=\rho\left[ C_{\s(x)}(g(x)) \-\right] \s^*\alpha_{\s(x)}  (X_x, \ldots ) = \rho\left[ C_{\s(x)}(g(x)) \-\right] a_x  (X_x, \ldots ).
\end{align*}
\end{proof}

The last result  holds true for the exterior covariant derivative: $(Da)'=\rho\left[C_\s(g)\-\right] Da$, which is also found from \eqref{passiveGTconnection} by having $(Da)'=D'a'=da'+\rho_*(A')\,a'$. In the language of physics, this would be an implementation of the \emph{gauge principle}.
As a special case of \eqref{passiveGTtensorial},  we obtain the gluing properties of the local representatives of the curvature, sections and their covariant derivative,
\begin{align}
\label{passiveGTothers}
F'= C_\s(g)\- F\, C_\s(g), \qquad   \phi'= \rho\left[C_\s(g)\-\right] \phi \quad \text{ and } \quad (D\phi)'=D'\phi'= \rho\left[C_\s(g)\-\right] D\phi.
\end{align}
The first result can be obtain from Cartan's structure equation and \eqref{passiveGTconnection} by having, $F'=dA'+\tfrac{1}{2}[A', A']$. 
%In the same way, the last result is obtained by $(D\phi)'=D'\phi'=d\phi'+\rho_*(A')\phi'$.
\medskip

Suppose now that we have a third open subset $\U''$ such that $\U'' \cap \U' \cap \U \neq \emptyset$, and consider a section $\s'':\U''\rarrow G$ such that on $\U'' \cap \U' \cap \U$, 
\begin{align*}
\s''=\s' g'=\s gg', \quad \text{ where }\quad g':\U'' \cap \U'\cap \U &\rarrow \, H, \\			%\s'(x)&=\s(x)g(x)
						          x \qquad   &\mapsto \, g'(x).
\end{align*}
We check that the gluing properties are well-behaved across open subsets. 
Using \eqref{Cs'-Cs} and \eqref{passiveGTconnection} we find that, 
\begin{align}
A''&=C_{\s'}(g')\- A' C_{\s'}(g') + C_{\s'}(g')\-d C_{\s'}(g'),  \notag\\
    &= C_\s(gg')\- C_\s(g) \left(    C_{\s}(g)\- A \, C_{\s}(g) + C_{\s}(g)\-d C_{\s}(g)      \right) C_\s(g)\- C_\s(gg') \notag \\
    &\hspace{4cm} + C_\s(gg')\- C_\s(g) d \left(   C_\s(g)\- C_\s(gg')  \right), \notag\\
%    \intertext{which  gives:}
%A''
&= C_\s(gg')\- A\, C_\s(gg') + C_\s(gg')\-d C_\s(gg').    \label{passivedoubleGTconnection}
\end{align}
In the same way, using  \eqref{Cs'-Cs} and \eqref{passiveGTtensorial}, 
\begin{align}
a''= \rho\left[C_{\s'}(g')\-\right] a' = \rho\left[  C_\s(gg')\- C_\s(g)  \right] \rho\left[C_\s(g)\-\right] a %, \notag\\
    =\rho\left[C_\s(gg')\-\right] a.
\end{align}
So that, $(D\alpha)''\!=\rho\left[C_\s(gg')\-\right] Da$, which is also obtained from $(Da)''\!=D''a''\!=\!da''+\rho_*(A'')\,a''$.
As special cases of this result, we have:
\begin{align}
\label{passivedoubleGTothers}
F''\!= C_\s(gg')\- F\, C_\s(gg'), \qquad   \phi''\!= \rho\left[C_\s(gg')\-\right] \phi \quad \text{ and } \quad (D\phi)''\!=\rho\left[C_\s(gg')\-\right] D\phi.
\end{align}
%with the first result obtained also from Cartan's structure equation and \eqref{passivedoubleGTconnection} by  $F''=dA''+\tfrac{1}{2}[A'', A'']$.
% and the last result also obtained by $(D\phi)''=D''\phi''=d\phi''+\rho_*(A'')\phi''$.
%\medskip

The gluing properties  \eqref{passiveGTconnection}-\eqref{passiveGTtensorial} in proposition \ref{passiveGTs}  resemble the active gauge transformations of propositions \ref{ConnectionGT} and \ref{TensorialGT}. But while the latter describe the  transformation of global objects (i.e. living on $\P$) into new global objects, the former merely describe how the same global objects are seen from  different open subsets of $\M$ - or from the same subset but through different local sections. This justifies the terminology \emph{passive gauge transformations} for the gluing properties, that is  of common use in physics. 

This is in close analogy with changes of coordinate representations of intrinsic geometric objects on $\M$ in (pseudo) Riemannian geometry and General Relativistic physics, which are dubbed \emph{passive diffeomorphisms} due to their formal identity with the action of $\Diff(\M)$ which transforms intrinsic objects into new ones. Elements of $\Diff(\M)$ are therefore sometimes called \emph{active diffeomorphisms}.

Yet, there is obviously also a local representation of the active gauge transformations discussed in section \ref{Action of vertical automorphisms and gauge transformations}. This is the object of the next section.

\subsection{Local active gauge transformations} %%%%%%%%%%%%%%%%%%%%%%%%%%%%%%%%%%%%%%%%%%%%%
\label{Local active gauge transformations} %%%%%%%%%%%%%%%%%%%%%%%%%%%%%%%%%%%%%%%%%%%%%%%%

Let us first denote the local representatives on $\U\subset\M$ of the gauge group elements $\gamma, \eta \in \H$ by  upright greek letters, 
%\begin{align*}
$\upgamma\defeq \s^*\gamma:\U \rarrow H$, %\qquad
and
 $\upeta\defeq \s^*\eta:\U \rarrow H$.
%\end{align*}
The local gauge group on $\U$ is then simply defined as $\H_{\text{loc}}\defeq \left\{ \upgamma : \U \rarrow H\, |\, \upgamma^\upeta=\upeta\- \upgamma \upeta \right\}$, where the defining property is the pullback by $\s$ of the action of $\H$ on itself. 
We then define the smooth map,
\begin{align*}
C_\s(\upgamma)\defeq \s^* C(\gamma) : \U &\rarrow G, \\
			  					x  &\mapsto C_{\s(x)}(\upgamma(x)),
\end{align*}
which  resembles the map $C_\s(g)$ introduced above. Its local active gauge transformation is,
\begin{align}
C_\s(\upgamma)^\upeta\defeq  \s^*\left( C(\gamma)^\eta\right)=%\s^*\left( \Psi^*C(\gamma)\right) =
						\s^*\left(  C(\eta)\- C(\gamma\eta) \right)
					=       C_\s(\upeta)\- C_\s(\upgamma\upeta)  \label{localGTCgamma}
\end{align}
Notice the close formal analogy with \eqref{Cs'-Cs}. The local active gauge transformations of a connection and tensorial forms are then:
\begin{align}
A^\upgamma&= \s^* \omega^\gamma=  \s^* \left( C(\gamma)\-\omega\, C(\gamma) + C(\gamma)\-dC(\gamma) \right), \notag \\
		    &= C_\s(\upgamma)\-A\, C_\s(\upgamma) + C_\s(\upgamma)\-dC_\s(\upgamma).   \label{localGTconnection}
\intertext{and}
a^\upgamma&=\s^* \alpha^\gamma=\s^*\left(  \rho\left[ C(\gamma)\- \right] \alpha \right), \notag \\
		    &=  \rho\left[ C_\s(\upgamma)\- \right] a. 			\label{localGTtensorial}
\end{align}
 This latter result holds true for the exterior covariant derivative, $(Da)^\upgamma= \rho\left[ C_\s(\upgamma)\- \right] Da$, which is also obtained from \eqref{localGTconnection} via $(Da)^\upgamma=D^\upgamma a^\upgamma= da^\upgamma + \rho_*\left( A^\upgamma \right)a^\upgamma$. This is again an  implementation of the \emph{gauge principle}.

As a special case of \eqref{localGTtensorial},  we obtain the local active gauge transformations of the curvature, sections and their covariant derivative,
\begin{align}
\label{localGTothers}
F^\upgamma= C_\s(\upgamma)\- F\, C_\s(\upgamma), \qquad   \phi^\upgamma= \rho\left[C_\s(\upgamma)\-\right] \phi \quad \text{ and } \quad (D\phi)^\upgamma=D^\upgamma\phi^\upgamma= \rho\left[C_\s(\upgamma)\-\right] D\phi.
\end{align}
The first result being also obtained from Cartan's structure equation and \eqref{localGTconnection} via $F^\upgamma=dA^\upgamma+\tfrac{1}{2}[A^\upgamma, A^\upgamma]$. 
\medskip

Finally, from \eqref{localGTCgamma} is is easily seen that, 
\begin{align*}
(A^\upgamma)^\upeta&=A^{\upgamma\upeta}= C_\s(\upgamma\upeta)\- A\, C_\s(\upgamma\upeta) + C_\s(\upgamma\upeta) \-d C_\s(\upgamma\upeta), \\
(a^\upgamma)^\upeta&=a^{\upgamma\upeta}=\rho\left[ C_\s(\upgamma\upeta)\- \right] a.
\end{align*}
Further transformations of $F$, $\phi$ and $D\phi$ ensue. This shows that the action of $\H_{\text{loc}}$ on local objects on $\U\subset \M$ is a well-behaved right action.
\medskip

Let us reiterate a standard yet important gauge theoretic observation: The local active gauge transformations \eqref{localGTconnection}-\eqref{localGTtensorial}, relating  local representatives seen through the same section $\s$ (by the same observer) of different global objects, are formally indistinguishable from the passive gauge transformations \eqref{passiveGTconnection}-\eqref{passiveGTtensorial}, relating  local representatives seen  through different sections $\s$ and $\s'$ (by distinct observers) of the same global objects. This is clear by  observing that $\s^*\Phi(p)=\Phi(\s(x))=\s(x)\gamma(\s(x))=\s(x)\upgamma(x)$. So the pullback by $\s$ of objects actively transformed by $\Phi/\gamma$ is formally equivalent to the pullback of the untransformed objects by a new local section $\s'=\s\upgamma$.
 Nevertheless, in physics, symmetry under active gauge transformations is  of much greater conceptual importance than the mere symmetry under passive ones. In the case of general relativistic physics for example, while symmetry of the theory under  coordinate changes translates as a principle -  the principle of general relativity - of democratic access to intrinsic objects of $\M$ which is thus at first seen as identical to the objective spacetime,  symmetry under $\Diff(\M)$ implies that the manifold $\M$ and its points are non-physical and that only relative field configurations over it - that can be diffeomorphically dragged - have physical meaning (this is the famous ``hole argument"). In Yang-Mills gauge theories - and a fortiori here - it is still not entirely clear how one should interpret the two types of formally equivalent symmetries.\footnote{This matter is distinct from another  important discussion, mainly addressed by philosophers of physics, regarding the demarcation criterion between \emph{substantial} and \emph{artificial} gauge symmetries. A main takeaway is that substantial gauge symmetries (either passive or active) in Yang-Mills theories signal non-local physical properties or phenomena, while artificial symmetries do not. See \cite{Francois2018} and references therein.}

 \section{Mixing with the standard situation} %%%%%%%%%%%%%%%%%%%%%%%%%%%%%%%%%%%%%%%%%%%%%
\label{Mixing with the standard situation} %%%%%%%%%%%%%%%%%%%%%%%%%%%%%%%%%%%%%%%%%%%%%%%%
 
 In this section we consider the minimal conditions under which the geometry described above mixes and coexists with the standard one. 
At the risk of some repetition, we thus slightly generalise the previous construction.
 %{\color{gray} We begin by providing some basic ingredients and results so as to streamline the presentation of the main propositions.}
 %\medskip
 
 We will first suppose that the structure group of the bundle $\P$ is a direct product $H \times K$, whose 
 elements are written $(h, k)=hk$ for $h\in H$ and $k \in K$. With the two subgroups commuting, the composition law is simply $hk \cdot h'k' = hkh'k'= hh'kk' \in H\times K$. The right action on $\P$ is thus $R_{hk}=R_{kh}$, and
  the right actions of the two subgroups commute: $R_h\circ R_k = R_k \circ R_h$. 
  
  We also consider the (inner) semi-direct product group $G \rtimes K$, whose elements are written $gk$ for $g \in G$ and $k\in K$. The two subgroups do not commute,  the composition law is $gk \cdot g'k' = gkg'k'=g\, kg'k\-\!\cdot kk' \in G \rtimes K$, and the group morphism $K \rarrow \Aut(G)$ defining the semi-direct product is  $k \mapsto \text{Conj}(k)$. 
Finally, we require that the representation $(\rho, V)$ of $G$ extends to a representation of $G \rtimes K$, and is therefore a representation for both subgroups. 
%\smallskip 

 The group of vertical automorphisms is also a  direct product $\Aut_v(\P)=\Aut_v(\P, H) \times \Aut_v(\P, K)$, with elements $\Psi=(\Phi, \Xi) $. Correspondingly, the gauge group is $\H \times \K$ with elements $(\gamma, \zeta)=\gamma\zeta$. The association is $\Psi(p)=p\gamma(p)\zeta(p)$. Because of the commutativity of the actions of $H$ and $K$ we have,
 \begin{equation}
 \label{eq-GTmaps}
 \begin{split}
 R^*_k\gamma=\gamma, \qquad \Xi^*\gamma=\gamma, \\
 R^*_h\zeta=\zeta , \qquad \Phi^*\zeta=\zeta. 
 \end{split}
 \end{equation}
 From this we have indeed that $\Psi=\Phi \circ \Xi= \Xi \circ\Phi$.
 
 Consider $X^v$ and $Y^v \in \Gamma(V\P)$ generated respectively by $X\in$ Lie$H$ and $Y\in$ Lie$K$. 
 We have the infinitesimal versions of the above equivariance laws,
  \begin{gather}
  \label{LieDer-GTmaps}
L_{Y^v} \gamma =Y^v(\gamma) = 0 \qquad \text{and} \qquad L_{X^v} \zeta =X^v (\zeta) =0 .
 \end{gather}
 Still by commutativity of the action of $H$ and $K$ we get, 
 \begin{equation} 
 \label{pushforward-vect}
 \begin{split}
 R_{k*}X_p^v=X_{pk}^v, \qquad  R_{k*}Y_p^v=\Ad_{k\-}Y \big|_{pk}^v, \\
  R_{h*}Y_p^v = Y_{ph}^v, \qquad	R_{h*}X_p^v = \Ad_{k\-}X \big|_{ph}^v ,
  \end{split}
 \end{equation}
 Also, for $Z^v=\{ X^v, Y^v \}$ it is easily shown that 
% \begin{gather}
$\Phi_*Z^v_p =Z^v_{\Phi(p)}$ and $\, \Xi_*Z^v_p=Z_{\Xi(p)}$.
% \end{gather}
\bigskip

The definition of the cocycle map $C$ prescribes its $H$-equivariance. We need to specify also its $K$-equivariance. It is easily found that the simplest choice compatible with its $H$-equivariance is, for $h' \in H$ and $k\in K$,
\begin{align}
\label{KeqC}
R^*_k C(h') = k\- C(h') k, \quad \text{ whose infinitesimal version is} \quad L_{Y^v}C(h')= [C(h'), Y].
\end{align}
Indeed, from \eqref{HeqC} and \eqref{KeqC}, one has on the one hand $C_{pkh}(h')=C_{pk}(h)\-C_{pk}(hh')=k\- C_p(h)\-k\cdot k\- C_p(hh')k$. And on the other hand,, $C_{phk}(h')=k\- C_{ph}(h')k=k\- C_p(h)\-C_p(hh')k$. The infinitesimal equivariance if obtained from $\left( L_{Y^v}C(h') \right)(p)=\tfrac{d}{d\tau}C_{pe^{\tau Y}}(h')\big|_{\tau = 0}$. 
It follows that the $\K$-gauge transformation of this map is,
\begin{align}
\label{K-GT-C}
C(h')^\zeta\defeq \Xi^*C(h') = \zeta\- C(h') \zeta.
\end{align}

In the same way, the $H$-equivariance of the map $C(\gamma)$ is known from \eqref{HequivCgamma}. From above, we get its $K$-equivariance
\begin{align}
\label{KequivCgamma}
R^*_k C(\gamma) = k\- C(\gamma) k, \quad \text{ with infinitesimal version } \quad L_{Y^v}C(\gamma)= [C(\gamma), Y].
\end{align}
From which follows that,
\begin{align}
\label{K-GT-Cgamma}
C(\gamma)^\zeta\defeq \Xi^*C(\gamma) = \zeta\- C(\gamma) \zeta.
\end{align}
 Now, let us see what we can do with these ingredients.
  
\subsection{Mixed vector bundles and tensorial forms} %%%%%%%%%%%%%%%%%%%%%%%%%%%%%%%%%%%%%%%%%%%%%
\label{Mixed vector bundles and tensorial forms} %%%%%%%%%%%%%%%%%%%%%%%%%%%%%%%%%%%%%%%%%%%%%%%%
 
 Given a representation $(\rho, V)$ of $G \rtimes K$, we define the mixed vector bundle   $\E^C=\P \times_{C(H)\rtimes K} V \defeq \P \times V / \sim$, with equivalence relation $(p, v) \sim \left(phk=pkh,\, \rho\left( k\-C_p(h)\- \right) v \right)$. It is well defined because on the one hand we have, 
 \begin{align*}
 (p, v) \sim \left(  ph, \rho\left[C_p(h)\-\right] v \right) \sim \left(  phk,  \rho\left( k\- \right)\rho\left[C_p(h)\-\right] v   \right)= \left(  phk,  \rho\left[k\-C_p(h)\-\right] v   \right).
 \end{align*}
 On the other hand, using \eqref{KeqC}, 
 \begin{align*}
 (p, v) \sim \left( pk, \rho\left( k\- \right) v\right) \sim \left( pkh, \rho\left[ C_{pk}(h)\- \right] \rho\left( k\- \right) v \right)=\left( pkh, \rho\left[ k\-C_{p}(h)\- k \right] \rho\left( k\- \right) v \right)=\left( pkh, \rho\left[ k\-C_{p}(h)\- \right]  v \right).
 \end{align*}
 It is clear that the twisted bundle $E^C=\P \times_{C(H)} V$ is a subbundle of $\E^C$, and  so is the standard vector bundle $E=\P \times_{K} V \defeq \P \times V / \sim$ for the equivalence relation $(p, v)\sim \big(pk, \rho( k\-) v \big)$. Hence the name for $\E^C$. 
 \medskip
 
  Define the space of $C(H)\rtimes K$-tensorial differential forms,
 \begin{align*}
 \Omega^\bullet_{\text{tens}}\big(\P, C(H)\!\rtimes\! K \big)=\bigg\{ \alpha \in \Omega^\bullet(\P, V)\, \big|\, \alpha_p(Z_p^v, \ldots)=~0 \allowbreak \text{ for } Z^v=\{X^v, Y^v\},\allowbreak \text{ and } \allowbreak R^*_{hk}\alpha=R^*_{kh}\alpha= \rho\big(C_p(h)k\big)\-\alpha \bigg\}.
 \end{align*}
 Clearly, these are in particular  both $C(H)$-tensorial and $K$-tensorial, and we indeed verify  the compatibility relations:
\begin{align*}
 R^*_k\left(R^*_h \alpha_{pkh} \right)&=R^*_k \left( \rho\left( C_p(h)\- \right) \alpha_{pk}\right) = \rho\left( C_{pk}(h)\-\right) R^*_k\alpha_{pk} = \rho\left( k\- C_p(h)\- k \right) \rho\left( k\- \right) \alpha_p= \rho\left( k\- C_p(h)\- \right)\alpha_p,\\
 R^*_h\left(R^*_k \alpha_{phk}\right) &= R^*_h \left( \rho\left( k\-\right) \alpha_{ph}\right)=\rho\left( k\-\right) \rho\left( C_p(h)\- \right)\alpha_p.
 \end{align*}
  As per the usual argument, there is an isomorphism $\Gamma\big(\E^C\big)\simeq \Omega^0_{\text{tens}}\big(\P, C(H)\!\rtimes\! K\big)$. 
 
  Again, the question arises as to the adequate notion of connection on $\P$ that provides a good covariant derivative on $\Omega^\bullet_{\text{tens}}\big(\P, C(H)\!\rtimes\! K\big)$, and on sections of $\E^C$ in particular.

\subsection{Mixed twisted connections} %%%%%%%%%%%%%%%%%%%%%%%%%%%%%%%%%%%%%%%%%%%%%
\label{Mixed twisted connections} %%%%%%%%%%%%%%%%%%%%%%%%%%%%%%%%%%%%%%%%%%%%%%%%
 
 We endow the bundle $\P(\M, H\!\times\! K)$ with a connection $\omega \in \Omega^1_\text{eq}\big(\P, \text{Lie}\big(G \rtimes K)\big)$ satisfying,
 \begin{align} 
 \label{Mixed 1st cond}
 &\omega_p\left( X_p^v+ Y_p^v \right)=\tfrac{d}{d\tau} C_p\left(e^{\tau \, X} \right)\big|_{\tau=0}+ Y= dC_{p|e}(X) + Y  \,\in \text{Lie}G \oplus \text{Lie}K,    \tag{\Rmnum{1}$^\star$}  \\[1mm]
 \label{Mixed 2nd cond}
  &R^*_{hk}\omega_{phk}=R^*_{kh} \omega_{pkh} = [ C_p(h)k]\- \omega_p\, [C_p(h)k] + [C_p(h)k]\- d [ C(h)k  ]_{|p}. \tag{\Rmnum{2}$^\star$}
 \end{align}
 It is clear that on the $H$-subbundle, $\omega$ satisfies the properties \eqref{1st axiom}-\eqref{2nd axiom} of a twisted connection, while on the $K$-subbundle it satisfies the definition of a standard $K$-principal connection:
% \begin{align}
 %\label{3rd axioms}
$ \omega_p(Y_p^v)=Y$ %\qquad \text{
 and %} \qquad  
 $R^*_k \omega_{pk} = \Ad_{k\-}\omega_p$. 
 %\end{align}
 We therefore call a 1-form $\omega$ defined by \eqref{Mixed 1st cond} and \eqref{Mixed 2nd cond}, a \emph{mixed} twisted connection. 

\subsubsection{Covariant derivative} %%%%%%%%%%%%%%%%%%%%%%%%%%%%
\label{Covariant derivative}%%%%%%%%%%%%%%%%%%%%%%%%%%%%

The infinitesimal version of the equivariance property of tensorial forms is, 
\begin{align*}
L_{X^v+ Y^v}\alpha = \tfrac{d}{d\tau} R^*_{e^{\tau (X+Y)}} \alpha\big|_{\tau=0} = \tfrac{d}{d\tau} \,\rho\left[ e^{-\tau Y}C\left(e^{\tau X}\right)\-\right] \alpha\big|_{\tau=0} = -\rho_* \left[Y+ \tfrac{d}{d\tau} C\left( e^{\tau X}\right) \big|_{\tau=0} \right] \alpha=-\rho_*\big( Y + dC_e(X) \big) \,\alpha.
\end{align*}
With this in mind, we obtain the following.

\begin{prop} %%%%%%%%%%%%%%%%%%%%%%%%%%%%%%%%%%%%%%%%%%%%%%%
\label{MixedCovDiff}
The exterior covariant derivative  defined as $D\defeq d\, +\!\rho_*(\omega)$ preserves both $ \Omega^\bullet_{\text{eq}}\big(\P, C(H)\!\rtimes\! K\big)$ and $ \Omega^\bullet_{\text{tens}}\big(\P, C(H)\!\rtimes\! K\big)$.

\end{prop}
\begin{proof} %%%%%%%%%%%%%%%%%%%%%%%%%%%%%%%%%%%%%%%%%%%%%%
First, using \eqref{Mixed 2nd cond},we show that  $D: \Omega^\bullet_{\text{eq}}\big(\P, C(H)\!\rtimes\! K\big) \rarrow \Omega^\bullet_{\text{eq}}\big(\P, C(H)\!\rtimes\! K\big)$. For $\alpha \in \Omega^\bullet_{\text{eq}}\big(\P, C(H)\!\rtimes\! K\big)$:
\begin{align*}
R^*_{hk} D\alpha &= d R^*_{hk} \alpha + \rho_*\left(R^*_{hk} \omega \right) R^*_{hk} \alpha, \\
		       &= d \rho\big(C(h)k\big)\-  \cdot \alpha + \rho\big(C(h)k\big)\-d \alpha %\\
		       %& \hspace{4cm}  
		        + \rho_*\left( [ C(h)k]\- \omega\, [C(h)k] + [C(h)k]\- d [ C(h)k  ] \right) \rho\big(C(h)k\big)\-\alpha,\\
		       &= \rho\big(C(h)k\big)\- \big(d\alpha + \rho_*(\omega)\alpha \big)=\rho\big(C(h)k\big)\- D\alpha.
\end{align*}
Then, using \eqref{Mixed 1st cond}, we show that $D: \Omega^\bullet_{\text{tens}}\big(\P, C(H)\!\rtimes\! K\big) \rarrow \Omega^\bullet_{\text{tens}}\big(\P, C(H)\!\rtimes\! K\big)$. It is enough to prove it for $\alpha \in  \Omega^1_{\text{tens}}\big(\P, C(H)\!\rtimes\! K\big)$:
\begin{align*}
D\alpha(X^v+Y^v, Z) &= \big(d\alpha+\rho_*(\omega)\alpha \big)(X^v+Y^v, Z), \\
		     % &= X^v\cdot \omega(Y) - Y\cdot \omega(X^v) - \omega([X^v, Y]) + \rho_*(\omega(X^v))\alpha(Y) - \rho_*(\omega(Y))\underbrace{\alpha(X^v)}_{=0}.
		         &= \underbrace{d\alpha(X^v+Y^v, Z)}_{\left( L_{X^v+Y^v}\alpha \right)(Z)}+ \rho_*(\omega(X^v+Y^v))\,\alpha(Z) - \rho_*(\omega(Z))\underbrace{\alpha(X^v+Y^v)}_{=0}, \\
		         &=-\rho_*\big( Y + dC_e(X) \big) \,\alpha(Z) + \rho_*\big(  dC_e(X) + Y\big) \,\alpha(Z)=0.
\end{align*}
\end{proof}
\noindent In particular, $D$ provides a good notion of covariant differentiation of sections of the mixed vector bundle $\E^C$.

\subsubsection{Curvature} %%%%%%%%%%%%%%%%%%%%%%%%%%%%
\label{Curvature}%%%%%%%%%%%%%%%%%%%%%%%%%%%%

The curvature of the mixed connection is defined in the usual way, so that we have the result:

\begin{prop}%%%%%%%%%%%%%%%%%%%%%%%%%%%%%%%%%%%%%%%%%%%%%%%
\label{Mixed Curvature}
The curvature $\Omega=d\omega+\tfrac{1}{2}[\omega, \omega]$ is a mixed tensorial $2$-form, $\Omega \in \Omega_{\text{tens}}^2\big(\P, C(H)\!\rtimes\! K\big)$. It satisfies a Bianchi identity $D\Omega=d\Omega+[\omega, \Omega]=0$.
\end{prop}
\begin{proof}%%%%%%%%%%%%%%%%%%%%%%%%%%%%%%%%%%%%%%%%%%%%%%%
The equivariance is proven the usual way via \eqref{Mixed 2nd cond},
\begin{align*}
R^*_{hk} \Omega &= dR^*_{hk} \omega + \tfrac{1}{2}[R^*_{hk} \omega, R^*_{hk}\omega], \\
		        &= d\left( [ C(h)k]\- \omega\, [C(h)k] + [C(h)k]\- d [ C(h)k  ]  \right) \\
		        & \hspace{1cm} + \tfrac{1}{2}\left[[ C(h)k]\- \omega\, [C(h)k] + [C(h)k]\- d [ C(h)k  ] , \, [ C(h)k]\- \omega\, [C(h)k] + [C(h)k]\- d [ C(h)k  ]  \right], \\
		        &= \ldots \\
		        &=  [C(h)k]\-\left( d\omega + \tfrac{1}{2} [\omega, \omega]\right) [C(h)k]
		        = [C(h)k]\- \Omega \, [C(h)k].
\end{align*}
Now, taking $Y^v=0$ in \eqref{Mixed 1st cond}, the horizontality of $\Omega$ w.r.t. $H$-vertical vector fields is proven as in Proposition \ref{curvature}. Taking $X^v=0$ in \eqref{Mixed 1st cond}, $\omega$ is a standard Ehresmann $K$-connection, so the horizontality of $\Omega$ w.r.t. $K$-vertical vector fields is proven the usual way. By linearity, $\Omega$ is $H\times K$-horizontal. It is therefore $C(H)\!\rtimes\! K$-tensorial. Since here $\rho=\Ad$, the covariant derivative is $D\Omega=d\Omega+[\omega, \Omega]$ and vanishes by definition of $\Omega$. 

\end{proof}
\noindent It is easily seen that another standard result that extends to the mixed case is  that for $\alpha \in \Omega^\bullet_\text{tens}\big(\P, C(H)\!\rtimes\! K \big)$, $D^2 \alpha=DD\alpha=\rho_*(\Omega)\alpha$.

\subsection{Mixed gauge transformations} %%%%%%%%%%%%%%%%%%%%%%%%%%%%%%%%%%%%%%%%%%%%
\label{Mixed gauge transformations} %%%%%%%%%%%%%%%%%%%%%%%%%%%%%%%%%%%%%%%%%%%%%%%

The  gauge transformations of the mixed connection and tensorial forms assume a simple form because the actions of $\Aut_v(\P, H)\simeq\H$ and $\Aut_v(\P, K)\simeq\K$ commute. Indeed we have the following. 

 \begin{prop} %%%%%%%%%%%%%%%%%%%%%%%%%%%%%%%%%%%
 \label{mixedGTconnection}
 The gauge transformations of $\omega$ and $\alpha \in \Omega^\bullet_\text{tens}\big(\P, C(H)\!\rtimes\! K \big)$ are, 
 \begin{align}
 \label{H-K_GTconnection}
 \omega^{\gamma\zeta}&%=(\omega^\zeta)^\gamma= (\omega^\gamma)^\zeta 
 					= [C(\gamma)\zeta]\- \omega\, [C(\gamma)\zeta] + [C(\gamma)\zeta]\-d [C(\gamma)\zeta],\\[1mm]
 \label{H-K_GTtensorial}
   \alpha^{\gamma\zeta}&%=(\alpha^\zeta)^\gamma= (\alpha^\gamma)^\zeta 
   					= \rho\big[ C(\gamma)\zeta \big]\- \alpha.
 \end{align}
 \end{prop}
 \begin{proof} %%%%%%%%%%%%%%%%%%
 First, notice that the push forward of a vector $X_p \in T_p\P$ by a vertical automorphism $\Psi \in \Aut_v(\P, H \times K)$ is, 
 \begin{align*}
 \Psi_*X_p&= R_{(\gamma\zeta)(p)*} X_p + [(\gamma\zeta)(p)]\-d(\gamma\zeta)_p(X_p)\big|^v_{\Psi(p)},\\
 		&=R_{\gamma(p)\zeta(p)*} X_p + \gamma(p)\-d\gamma_{|p}(X_p)\big|^v_{\Psi(p)} + \zeta(p)\-d\zeta_{|p}(X_p)\big|^v_{\Psi(p)}
 \end{align*}
 Therefore, the full gauge transformation of the mixed connection is by definition, 
 \begin{align*}
 \omega^{\gamma\zeta}_p(X_p) &\defeq \left(\Psi^*\omega\right)_p (X_p)=\omega_{\Psi(p)}\left( \Psi_* X_p\right)
					  = R_{\gamma(p)\zeta(p)}^*\omega_{\Psi(p)}(X_p) 
								+ \tfrac{d}{d\tau}C_{\Psi(p)}\left( e^{\gamma(p)\-d\gamma_p(X_p)}  \right)\big|_{\tau=0} +  \zeta(p)\-d\zeta_{|p}(X_p), \\
					&=\left( [C_p(\gamma(p))\zeta(p)]\- \omega_p [C_p(\gamma(p))\zeta(p)] + [C_p(\gamma(p))\zeta(p)]\-d[C(\gamma(p))\zeta(p)]_{|p}\right)(X_p)	\\	
 					& \hspace{2cm} +  \zeta(p)\- \tfrac{d}{d\tau}C_{p\gamma(p)}\left( e^{\gamma(p)\-d\gamma_p(X_p)}  \right)\big|_{\tau=0}\,  \zeta(p) 
					  +  \zeta(p)\-d\zeta_{|p}(X_p),\\
					&= [C_p(\gamma(p))\zeta(p)]\- \omega_p(X_p) [C_p(\gamma(p))\zeta(p)] + \zeta(p)\-\, C_p(\gamma(p))\-d C(\gamma(p))_{|p} (X_p)\, \zeta(p)	\\	
 					& \hspace{2cm} +  \zeta(p)\- \left(C_p(\gamma(p))\- dC_p(\gamma)_{|p} \right)(X_p)\,  \zeta(p) 
					  +  \zeta(p)\-d\zeta_{|p}(X_p),\\
					&=  [C_p(\gamma(p))\zeta(p)]\- \omega_p(X_p) [C_p(\gamma(p))\zeta(p)] + \zeta(p)\-\, C_p(\gamma(p))\-d C(\gamma)_{|p} (X_p)\, \zeta(p)	\\	
 					& \hspace{2cm}	 +  \zeta(p)\-d\zeta_{|p}(X_p),\\
					&=  [C_p(\gamma(p))\zeta(p)]\- \omega_p(X_p) [C_p(\gamma(p))\zeta(p)] + \zeta(p)\-\, C_p(\gamma(p))\-d [C(\gamma)\zeta]_{|p} (X_p),\\
				%	&= \left( [C_p(\gamma(p))\zeta(p)]\- \omega_p [C_p(\gamma(p))\zeta(p)] +  [C_p(\gamma(p))\zeta(p)]\-d [C(\gamma)\zeta]_{|p} \right)(X_p),\\
 					&= \left( [C(\gamma)\zeta]\- \omega\, [C(\gamma)\zeta] +  [C(\gamma)\zeta]\-d [C(\gamma)\zeta] \right)_{|p}(X_p).
\end{align*} 
 In the same way, for a mixed tensorial form, 
 \begin{align*}
 \alpha^{\gamma\zeta}_p(X_p, \ldots) &\defeq \left(\Psi^*\alpha\right)_p (X_p, \ldots)=\alpha_{\Psi(p)}\left( \Psi_* X_p, \ldots \right)
 					  = \alpha_{\Psi(p)}\left( R_{\gamma(p)\zeta(p)*} X_p, \ldots \right)
					  = R_{\gamma(p)\zeta(p)}^*\alpha_{\Psi(p)}(X_p, \ldots) , \\
					&= \rho\big[ C_p(\gamma(p)) \zeta(p) \big]\- \alpha_p (X_p, \ldots)= \left( \rho\big[ C(\gamma)\zeta \big]\-\alpha \right)_{|p} (X_p, \ldots).
\end{align*}
\end{proof} %%%%%%%%%%%%%%%%%

As special cases of \eqref{H-K_GTtensorial}, we obtain the gauge transformations of the curvature, of mixed sections and their covariant derivative, 
 \begin{align*}
 \Omega^{\gamma\zeta}=  [C(\gamma)\zeta]\- \Omega\, [C(\gamma)\zeta], \qquad \vphi^{\gamma\zeta} =\rho\big[ C(\gamma)\zeta \big]\-\vphi, \quad \text{and} \quad 
 (D\vphi)^{\gamma\zeta} =\rho\big[ C(\gamma)\zeta \big]\-D\vphi.
 \end{align*}

 It is clear from \eqref{H-K_GTconnection} that $\omega$ transforms in particular as a twisted connection under $\H$ and a standard connection under $\K$. So, on the one hand, by  \eqref{eq-GTmaps},
 \begin{align*}
 (\omega^\zeta)^\gamma&= \Phi^*\left( \Xi^*\omega \right) = \Phi^*\left( \zeta\-\omega\, \zeta+\zeta\-d\zeta \right)
 										      = \zeta\- \Phi^*\omega\, \zeta + \zeta\-d\zeta,\\
				       &= \zeta\- \left(  C(\gamma)\- \omega\, C(\gamma) + C(\gamma)\-dC(\gamma)  \right) \zeta +  \zeta\-d\zeta,\\
				       &= \zeta\-C(\gamma)\- \omega\, C(\gamma)\zeta + \zeta\-C(\gamma)\-d \left( C(\gamma)\zeta \right).
 \end{align*}
 On the other hand, by   \eqref{K-GT-Cgamma} ,
 \begin{align*}
 (\omega^\gamma)^\zeta&=\Xi^*\left( \Phi^*\omega \right) = \Xi^* \left( C(\gamma)\- \omega\, C(\gamma) + C(\gamma)\-dC(\gamma) \right), \\
				   &= \zeta\-C(\gamma)\-\zeta \left( \zeta\-\omega\, \zeta + \zeta\-d\zeta \right) \zeta\-C(\gamma)\zeta + \zeta\-C(\gamma)\-\zeta d\left( \zeta\-C(\gamma)\zeta \right), \\
				   &= \zeta\-C(\gamma)\- \omega\, C(\gamma)\zeta + \zeta\-C(\gamma)\-d \left( C(\gamma)\zeta \right).
\end{align*}
 In the same way, from \eqref{H-K_GTtensorial} it is clear that $\alpha$ transforms in particular as a twisted tensorial form under $\H$ and as a standard tensorial form under $\K$. So that,   
 \begin{align}
 \label{active-mixedGT-tensorial}
 \begin{split}
 (\alpha^\zeta)^\gamma=\Phi^*\left( \Xi^*\alpha\right)&=\Phi^*\left( \rho\left( \zeta\- \right)\alpha\right) = \rho\left( \zeta\- \right) \rho\left[ C(\gamma)\- \right]\alpha=\rho\left[C(\gamma)\zeta \right]\-\!\alpha,\\
 (\alpha^\gamma)^\zeta=\Xi^*\left( \Phi^*\alpha\right)&=\Xi^*\left( \rho\left[ C(\gamma)\- \right]\alpha \right) = \rho\left[ \zeta\- C(\gamma)\-\zeta \right] \rho\left( \zeta\- \right)\alpha= \rho\left[  C(\gamma)\zeta \right]\-\!\alpha.
 \end{split}
 \end{align}
The results \eqref{H-K_GTconnection}-\eqref{H-K_GTtensorial}  indeed express the commutativity of the action of $\H$ and $\K$. 
 \medskip

We can use a notational game that allows to perform symbolically the computation of gauge transformations we have seen so far, and may be used as a memory trick. %It is sometimes used in the physics literature. 
 If $\omega^\zeta$ denotes the result of a gauge transformation by $\K$, a further such transformation is noted $(\omega^\zeta)^\xi=(\omega^\xi)^{\zeta^\xi}=(\omega^\xi)^{\xi\- \zeta \xi}=\omega^{\zeta\xi}$. Now, denote $\omega^\gamma=\omega^{C(\gamma)}$ the result of a gauge transformation by $\H$. A further such transformation if then $(\omega^\gamma)^\eta=(\omega^{C(\gamma)})^\eta=(\omega^\eta)^{C(\gamma)^\eta}=(\omega^{C(\eta)})^{C(\eta)\-C(\gamma\eta)}= \omega^{C(\gamma\eta)}$. Then, $(\omega^\gamma)^\zeta=(\omega^{C(\gamma)})^\zeta=(\omega^\zeta)^{C(\gamma)^\zeta}=(\omega^\zeta)^{\zeta\- C(\gamma)\zeta}=\omega^{C(\gamma)\zeta}$. Also, $(\omega^\zeta)^\gamma=(\omega^\gamma)^{\zeta^\gamma}=(\omega^{C(\gamma)})^\zeta=\omega^{C(\gamma)\zeta}$. So we could note $\omega^{\gamma\zeta}$ as $\omega^{C(\gamma)\zeta}$ instead. This would have the advantage of making the mixed structure clear. Idem for  twisted or mixed tensorial forms.
 
 \subsection{Local version} %%%%%%%%%%%%%%%%%%%%%%%%%%%%%%%%%%%%%%%%%%%%%
\label{Local version} %%%%%%%%%%%%%%%%%%%%%%%%%%%%%%%%%%%%%%%%%%%%%%%%

To be complete, and at the risk of some redundancy, we provide in the next two subsections the local description of the above mixed global geometry. 

\subsubsection{Passive mixed gauge transformations}  %%%%%%%%%%%%%%%%%%%%%%%%%%%%%%%%%%%%%%%%%%%%%%%%%%
\label{Passive mixed gauge transformations}  %%%%%%%%%%%%%%%%%%%%%%%%%%%%%%%%%%%%%%%%%%%%%%%%%%

Consider again $\U, \U' \subset \M$ such that $\U \cap \U' \neq \emptyset$, with the  local sections $\s:\U\rarrow \P_{|\U}$ and $\s':\U'\rarrow \P_{|\U'}$ related by, 
%$\s'=R_g\circ R_\ell \circ \s = R_\ell\circ R_g \circ \s $ :
\begin{align*}
\s'&=\s g \ell, \hspace{7mm} \text{ where }\hspace{-30mm}& g:\U \cap \U' &\rarrow \, H,&\hspace{-30mm} \text{ and } \quad \ell : \U \cap \U' &\rarrow\, K. \\
    &=\s \ell g				 &	                    x \quad   &\mapsto \, g(x) &                                                                                    x \quad   & \mapsto \, \ell(x) 
\end{align*}
Then, for $X_x \in T_x\M$, $x \in \U \cap \U'$, the pushforwards by $\s'$ and $\s$ are related by
\begin{align*}
\s'_* X_x &= R_{g(x)\ell(x)*}\left( \s_* X_x \right)  + [ (g\ell)\-d(g\ell)]_{|x}(X_x) \big|^v_{\s'(x)}, \\
               &= R_{g(x)\ell(x)*}\left( \s_* X_x \right) + [ g\-dg]_{|x}(X_x) \big|^v_{\s'(x)} + [ \ell\-d\ell]_{|x}(X_x) \big|^v_{\s'(x)}. 
\end{align*}
From this we find the following gluing properties.

\begin{prop}%%%%%%%%%%%%%%%%%%%%%%%%%%%%%%%%%%%%%%%%%
\label{passive-mixed-GTs}
The gluing properties of the local representatives on $\U$ and $\U'$ of a mixed connection and a mixed tensorial forms are, 
\begin{align}
\label{passiveGT-mixed-connection}
A'&=[C_\s(g)\ell]\- A\, [C_\s(g)\ell] + [C_\s(g)\ell]\-d[C_\s(g)\ell], \\[1.5mm]
\label{passiveGT-mixed-tensorial}
a'&=\rho\left[C_\s(g)\ell\right]\-\! a.
\end{align}
\end{prop}
\begin{proof} %%%%%%%%%%%%%%%%%%%%%%%%%%%%%%%%%%%%%%%
The proof is  as for active gauge transformations in Proposition \ref{mixedGTconnection}.
Using  \eqref{Mixed 1st cond}-\eqref{Mixed 2nd cond} as well as the proof of Proposition \ref{passiveGTs}, we find:
\begin{align*}
A'_x(X_x)&= (\s'^{*}\omega)_x (X_x) = \omega_{\s'(x)}\left( \s'_* X_x\right),  \\
               &= \omega_{\s'(x)}\left( R_{g(x)\ell(x)*}\left( \s_* X_x \right) + [ g\-dg]_{|x}(X_x) \big|^v_{\s'(x)} + [ \ell\-d\ell]_{|x}(X_x) \big|^v_{\s'(x)} \right),  \\
               &= R^*_{g(x)\ell(x)} \omega_{\s'(x)}\left(  \s_* X_x   \right) + \tfrac{d}{d\tau} C_{\s'(x)}\left( e^{\tau \, [g\-dg]_{|x}(X_x) } \right) \big|_{\tau=0}
               							+    [\ell\-d\ell]_{|x}(X_x) ,\\
	       &= \left( [C_{\s(x)}\left( g(x) \right) \ell(x)]\- \omega_{\s(x)}\, [C_{\s(x)}\left( g(x) \right) \ell(x)]  + [C_{\s(x)}\left( g(x) \right) \ell(x)] \- d[C(g(x))\ell(x)]_{|\s(x)}\right) (\s_*X_x)\\
	       &	\hspace{3cm}			+   \ell(x)\-  \left( \tfrac{d}{d\tau} C_{\s(x)g(x)}\left( e^{\tau \, [g\-dg]_{|x}(X_x) } \right) \big|_{\tau=0} \right) \ell(x) +       [\ell\-d\ell]_{|x}(X_x) ,\\
	       &=\left( [C_{\s(x)}\left( g(x) \right) \ell(x)]\- A_x\, [C_{\s(x)}\left( g(x) \right) \ell(x)]  + [C_{\s(x)}\left( g(x) \right) \ell(x)] \- d[C_\s(g(x))\ell(x)]_{|x}\right) (X_x)\\
	       &	\hspace{3cm}			+   \ell(x)\-  C_{\s(x)}(g(x))\- dC_{\s(x)}(g)_{|x}(X_x)\, \ell(x) +       [\ell\-d\ell]_{|x}(X_x) ,\\
               &=\left( [C_{\s(x)}\left( g(x) \right) \ell(x)]\- A_x\, [C_{\s(x)}\left( g(x) \right) \ell(x)]  + [C_{\s(x)}\left( g(x) \right) \ell(x)] \- d[C_\s(g)\ell(x)]_{|x}\right) (X_x)
	      	+       [\ell\-d\ell]_{|x}(X_x) ,\\
               &= \left( [C_{\s(x)}\left( g(x) \right) \ell(x)]\- A_x\, [C_{\s(x)}\left( g(x) \right) \ell(x)]  + [C_{\s(x)}\left( g(x) \right) \ell(x)] \- d[C_{\s}\left( g \right) \ell]_{|x}  \right)(X_x).
\end{align*}

For the local representatives of a tensorial form $\alpha \in \Omega^\bullet_{\text{tens}}\big(\P, C(H)\!\rtimes\!K\big)$ we get:
\begin{align*}
a'_x(X_x, \ldots) &= \left( \s'^{*}\alpha \right)_x(X_x, \ldots)= \alpha_{\s'(x)}\left(  \s'_* X_x, \ldots \right)
                = \alpha_{\s'(x)}\left(  R_{g(x)\ell(x)*}\left( \s_* X_x \right), \ldots  \right)%, \\
                =R^*_{g(x)\ell(x)} \alpha_{\s(x)} \left( \s_* X_x, \ldots \right)  ,\\
                &= \rho\left[C_{\s(x)}(g(x))\ell(x)\right]\-\! \alpha_{\s(x)} \left(  \s_* X_x, \ldots \right) =  \rho\left[C_{\s(x)}(g(x))\ell(x) \right]\-\!  a_x \left(  X_x, \ldots \right). 
\end{align*}
%Or, restarting from the third equality,
%\begin{align*}
%a'_x(X_x, \ldots) &= \alpha_{\s'(x)}\left(  R_{g(x)\ell(x)*}\left( \s_* X_x \right), \ldots  \right), \\
%                &=R^*_{g(x)} \alpha_{\s'(x)} \left( R_{\ell(x)*}\left( \s_* X_x \right), \ldots \right)  ,\\
%                &= \rho\left[ C_{s(x)\ell(x)} \left(g(x) \right) \right] \alpha_{s(x)\ell(x)} \left(  R_{\ell(x)*}\left( \s_* X_x\right), \ldots \right) 
%                   = \rho\left[ \ell(x)\- C_{s(x)} \left(g(x) \right) \ell(x) \right]    R^*_{\ell(x)}\alpha_{s(x)\ell(x)}\left( \s_* X_x, \ldots \right) ,\\
%                &= \rho\left[ \ell(x)\- C_{s(x)} \left(g(x) \right) \ell(x) \right] \rho\left[ \ell(x)\- \right]    \alpha_{\s(x)}\left( \s_* X_x \right)= \rho\left[ \ell(x)\-  C_{s(x)} \left( g(x) \right)\- \right] a_x (X_x, \ldots).
%\end{align*}                
\end{proof}
\noindent This last result is also valid for the exterior covariant derivative, which is also found by having $(Da)'=D'a' = da' + \rho_*\left( A' \right)a'$. 
 As special cases of \eqref{passiveGT-mixed-tensorial}, we have the gluings of the local representatives of the curvature, sections and their covariant derivative: 
 \begin{align}
 \label{passiveGT-mixed-others}
F'= [C_\s(g)\ell]\- F\, [C_\s(g)\ell], \qquad   \phi'= \rho\left[C_\s(g)\ell\right]\-\! \phi \quad \text{ and } \quad (D\phi)'=D'\phi'= \rho\left[C_\s(g)\ell\right]\-\! D\phi.
 \end{align}
 And as usual the first result is also obtained  from Cartan structure equation and \eqref{passiveGT-mixed-connection}.

\subsubsection{Local active mixed gauge transformations}  %%%%%%%%%%%%%%%%%%%%%%%%%%%%%%%%%%%%%%%%%%%%%%%%%%
\label{Local active mixed gauge transformations}  %%%%%%%%%%%%%%%%%%%%%%%%%%%%%%%%%%%%%%%%%%%%%%%%%%
 
On a single open set $\U \subset \M$ with local section $\s : \U \rarrow \P_{|\U}$,   the local gauge group is $\H_{\text{loc}} \times \K_{\text{loc}}$, with $\H_{\text{loc}}\defeq \left\{ \upgamma : \U \rarrow H\, |\, \upgamma^\upeta=\upeta\- \upgamma \upeta \right\}$ and $\K_{\text{loc}}\defeq \left\{ \upzeta  : \U \rarrow K\, |\, \upzeta^\upxi=\upxi\- \upzeta \upxi \right\}$, but also - as local version of \eqref{eq-GTmaps} - $\upgamma^\upzeta=\upgamma$ and $\upzeta^\upgamma=\upzeta$. 
The  $\H_{\text{loc}}$-transformation of the map $C_\s(\upgamma) : \U \rarrow G$ is given by \eqref{localGTCgamma}. Its $\K_{\text{loc}}$-transformation, the counterpart of \eqref{K-GT-Cgamma}, is
 \begin{align}
 \label{local-K-GT-Cgamma}
 C_\s(\upgamma)^\upzeta= \upzeta\- C_\s(\upgamma) \upzeta.
 \end{align}

The local mixed active gauge transformations of $A$ and $a$ are simply,
  \begin{align}
  A^{\upgamma\upzeta} &= \s^* \omega^{\gamma\zeta}= [C_\s(\upgamma)\upzeta]\- A \, [C_\s(\upgamma)\upzeta] + [C_\s(\upgamma)\upzeta]\- d[C_\s(\upgamma)\upzeta],      															    \label{local-active-mixedGT-connection}. \\[1.5mm]
  a^{\upgamma\upzeta} &= \s^* \alpha^{\gamma\zeta}= \rho[ C_\s(\upgamma)\upzeta]\- a.     \label{local-active-mixedGT-tensorial}
  \end{align}
 We prove as in section \ref{Mixed gauge transformations} that $\H_{\text{loc}}$ and $\K_{\text{loc}}$ commute.
  Using  \eqref{local-K-GT-Cgamma}, one easily shows that $\big(A^\upgamma \big)^\upzeta= \big(A^\upzeta \big)^\upgamma$ and
  $\big(a^\upgamma \big)^\upzeta= \big(a^\upzeta \big)^\upgamma$.
  Here also we could write \eqref{local-active-mixedGT-connection} as $A^{C(\upgamma)\upzeta}$ and \eqref{local-active-mixedGT-tensorial} as $a^{C(\upgamma)\upzeta}$.

 Again, \eqref{local-active-mixedGT-tensorial} hold for $(Da)$ and is also found via $(Da)^{\upgamma\upzeta}=D^{\upgamma\upzeta} a^{\upgamma\upzeta}$. The mixed local gauge transformations for $F$, $\phi$ and $D\phi$ are found as special cases, 
  \begin{align}
 \label{activeGT-mixed-others}
F^{\upgamma\upzeta}= [C_\s(\upgamma)\upzeta]\- F\, [C_\s(\upgamma)\upzeta] , \qquad   \phi^{\upgamma\upzeta}= \rho\left[C_\s(\upgamma)\upzeta\right]\- \phi \quad \text{ and } \quad (D\phi)^{\upgamma\upzeta}= \rho\left[C_\s(\upgamma)\upzeta \right]\- D\phi.
 \end{align}
 As always, the transformation of $F$ is also obtained  from Cartan's structure equation and \eqref{local-active-mixedGT-connection}.
 
 Finally, we notice the exact formal resemblance between \eqref{passiveGT-mixed-connection}-\eqref{passiveGT-mixed-tensorial} and \eqref{local-active-mixedGT-connection}-\eqref{local-active-mixedGT-tensorial}, and refer back to our comments at the end of section \ref{Local active gauge transformations}.

% \clearpage
 
\section{Action of the Lie algebra of vertical automorphisms} %%%%%%%%%%%%%%%%%%%%%%%%%%%%%%%%%%%%%%%%%
\label{Action of the Lie algebra of vertical automorphisms} %%%%%%%%%%%%%%%%%%%%%%%%%%%%%%%%%%%%%%%%%%%

Let us  denote the Lie algebras of $\Aut_v(\P, H)$ and $\H$ respectively by $\aut_v(\P, H)$ and Lie$\H$. Given $\Phi_\tau \in \Aut_v(\P,H)$, we have $\Phi_\tau(p)=p\gamma_\tau(p)=pe^{\tau\, \chi(p)}$ with $\chi : \P \rarrow \text{Lie}H$. Since $R^*_h\gamma_\tau=h\- \gamma_\tau h$, it comes that $R^*_h\,\chi =\Ad_{h\-}\chi$. We have then Lie$\H \defeq \left\{\, \chi : \P \rarrow \text{Lie}H\, | \, R^*_h\,\chi= \Ad_{h\-}\chi \right\}$. 
In other terms Lie$\H=C^\infty_\Ad(\P, \text{Lie}H) \simeq \Gamma\left( \P\times_\Ad \text{Lie}H \right)$.

Now, $\Phi_\tau$ is the flow of the vector field 
\begin{align}
\label{iso-aut_v-LieH}
\chi^v_p\defeq \tfrac{d}{d\tau} \Phi_\tau(p)\big|_{\tau=0}=\tfrac{d}{d\tau} pe^{\tau\, \chi(p)} \big|_{\tau=0}. 
\end{align}
By the way,
$\chi^v_{ph}\defeq \tfrac{d}{d\tau} phe^{\tau\, \chi(ph)}=\tfrac{d}{d\tau} phe^{\tau\, h\-\chi(p)h}=\tfrac{d}{d\tau} ph\, h\-e^{\tau\, \chi(p)}h=\tfrac{d}{d\tau} R_h\left(pe^{\tau\, \chi(p)}\right) \rdefeq R_{h*}\,\chi^v_p$.
That is $\chi^v$ is a right-invariant vector field, $\chi^v \in \Gamma^H(T\P)$. Conversely it is easily shown that the flow $\phi_\tau:\P \rarrow \P$ of a right-invariant vertical vector field is such that $\phi_\tau (ph)=\phi_\tau(p)h$ and $\pi \circ \phi_\tau = \pi$, i.e. it is a vertical automorphism. So 
$\aut_v(\P, H) = \Gamma^H(T\P) \cap \Gamma(V\P)$. The  definition \eqref{iso-aut_v-LieH} describe explicitly the isomorphism $\aut_v(\P, H) \simeq \text{Lie}\H$.

\subsection{Infinitesimal gauge transformations} %%%%%%%%%%%%%%%%%%%%%%%%%%%%%%%%%%%%%%%%%%%%%
\label{Infinitesimal gauge transformations} %%%%%%%%%%%%%%%%%%%%%%%%%%%%%%%%%%%%%%%%%%%%%%%%

To begin with, we want the infinitesimal counterpart of propositions \eqref{ConnectionGT} and \eqref{TensorialGT}. For this, let us first give the following result,
\begin{align}
\label{result1}
\tfrac{d}{d\tau} C_p\left( e^{\tau \, \chi(p)}\right)\big|_{\tau=0} = dC_{p|e}\left(  \tfrac{d}{d\tau}  e^{\tau \, \chi(p)} \big|_{\tau=0} \right) = dC_{p|e}(\chi(p)).
\end{align}
Remember that in $dC_{p|e}$, $d$ is not de Rham derivative on $\P$ but  signifies the pushforward of the map $C_p:H \rarrow G$.
%{\color{red} Here should be the H-equiv of this map (perhaps also its Lie derivative w.r.t vert vect field) that is of use for reference in the BRST section.} 
Given this, it is easy to prove the following,

\begin{prop} %%%%%%%%%%%%%%%%%%%%%%%%%%%%%%%%%%
\label{infinitesimal-activeGT}
The infinitesimal active gauge transformations of the connection and of $C$-tensorial forms are,
\begin{align}
\label{inf-activeGT}
\begin{split}
L_{\chi^v} \omega &= d\left( dC_{|e}(\chi)) \right) + \left[ \omega, dC_{|e}(\chi) \right] , \\[1.5mm]
L_{\chi^v} \alpha   &= -\rho_* \left( dC_{|e}(\chi) \right) \alpha.
\end{split}
\end{align}
\end{prop}
\begin{proof} %%%%%%%%%%%%%%%%%%%%%%%%%%%%%%%%%%
The Lie derivative of $\omega$ w.r.t $\chi^v$ is,  
\begin{align*} 
L_{\chi^v}\omega \defeq& \tfrac{d}{d\tau}\,  \Phi_\tau^*\omega \, \big|_{\tau=0}=  \tfrac{d}{d\tau} \, C\left( e^{\tau\, \chi }\right)\- \omega\, C\left( e^{\tau\, \chi }\right) + C\left( e^{\tau\, \chi }\right)\-dC\left( e^{\tau\, \chi }\right)  \big|_{\tau=0} ,\\
			   =& -\tfrac{d}{d\tau} C\left( e^{\tau\, \chi }\right) \big|_{\tau=0} \, \omega + \omega \,\tfrac{d}{d\tau} C\left( e^{\tau\, \chi }\right) \big|_{\tau=0} + d\, \tfrac{d}{d\tau} C\left( e^{\tau\, \chi }\right) \big|_{\tau=0}.
\end{align*} 
Using \eqref{result1} the result follows. As for $\alpha \in \Omega^\bullet_{\text{tens}}\big(\P, C(H)\big)$,
\begin{align*} 
L_{\chi^v}\alpha \defeq& \tfrac{d}{d\tau}\,  \Phi_\tau^*\alpha \, \big|_{\tau=0}=  \tfrac{d}{d\tau} \, \rho\left[ C\left( e^{\tau\, \chi }\right)\-\right] \alpha\,   \big|_{\tau=0} 
			   = -\rho_*\left[\tfrac{d}{d\tau}  C\left( e^{\tau\, \chi }\right) \big|_{\tau=0}  \right] \alpha.
\end{align*} 
\end{proof}
\noindent From this is immediately deduced that, 
\begin{align*} 
L_{\chi^v}\Omega = \left[ \Omega, dC_{|e}(\chi) \right], \quad L_{\chi^v}\vphi = -\rho_* \left( dC_{|e}(\chi) \right) \vphi, \quad \text{ and } \quad L_{\chi^v}D\vphi = -\rho_* \left( dC_{|e}(\chi) \right) D\vphi.
\end{align*} 
\medskip

 It suffices to pullback \eqref{inf-activeGT} on $\U \subset \M$ to obtain  the infinitesimal active gauge transformations of the local representatives, counterparts of \eqref{localGTconnection} and \eqref{localGTtensorial}. Given $\upchi\defeq \s^*\chi : \U \rarrow \text{Lie}H \in \text{Lie}\H_{\text{loc}}$, this gives
\begin{align}
\label{local-inf-activeGT}
 \begin{split}
\delta_\upchi A &= \s^*\left( L_{\chi^v }\omega \right) = d\left( dC_{\s|e}(\upchi) \right) + \left[ A, dC_{\s|e}(\upchi)\right],  \\[0.5mm]
\delta_\upchi a &= \s^*\left( L_{\chi^v} \alpha \right) = -\rho_*\left( dC_{\s|e}(\upchi)\right)  a. 
 \end{split}
\end{align}
Alternatively, these are obtained in a way analogous to the above proof, by defining $\delta_\upchi A\defeq \tfrac{d}{d\tau}\, A^{\upgamma_\tau} \big|_{\tau=0}$ and $\delta_\upchi a\defeq \tfrac{d}{d\tau}\, a^{\upgamma_\tau} \big|_{\tau=0}$.
\bigskip

The infinitesimal versions of the gluing properties, proposition \ref{passiveGTs}, are not related to the action of $\aut_v(\P, H)\simeq \text{Lie}\H$ but are obtained by a completely analogous computation. On poses $g_\tau(x)=e^{\tau\lambda(x)}$, with $\lambda :\U \rarrow \text{Lie}H$, and defines:
\begin{align}
 \begin{split}
\delta_{\!\lambda} A &\defeq \tfrac{d}{d\tau}\, A'  \big|_{\tau=0} = d\left( dC_{\s|e}(\lambda) \right) + \left[ A, dC_{\s|e}(\lambda)\right],  \\[0.5mm]
\delta_{\!\lambda}  a  &\defeq \tfrac{d}{d\tau}\, a' \big|_{\tau=0} = -\rho_*\left( dC_{\s|e}(\lambda)\right)  a. 
 \end{split}
\end{align}
And as usual, these passive infinitesimal gauge transformations are formally identical to the active ones, Eq.\eqref{local-inf-activeGT}.

\subsection{Mixed case} %%%%%%%%%%%%%%%%%%%%%%%%%%%%%%%%%%%%%%%%%%%%%
\label{} %%%%%%%%%%%%%%%%%%%%%%%%%%%%%%%%%%%%%%%%%%%%%%%%

In the same spirit as in the above considerations, under the assumptions of section \ref{Mixing with the standard situation} we are interested in the infinitesimal counterpart of the mixed gauge transformations of the connection \eqref{H-K_GTconnection} and of $C$-tensorial forms \eqref{H-K_GTtensorial}. Let us denote  Lie$\K\defeq\left\{ \upsilon: \P \rarrow K\, | \, R^*_k\upsilon =\Ad_{k\-}\upsilon  \right\}$.
% = C^\infty_\Ad(\P, \text{Lie}K)\simeq \Gamma\left( \P \times_\Ad \text{Lie}K \right)$. 
 It is isomorphic to $\aut_v(\P, K)$ via the definition $\upsilon^v_p\defeq \tfrac{d}{d\tau} pe^{\tau\, \upsilon(p)} \big|_{\tau=0}$. 
 Elements of Lie$\H$ and Lie$\K$ also satisfy: $R^*_k \,\chi =\chi$ and $R^*_h \upsilon =\upsilon$.
 We have then $\aut_v(P)=\aut_v(\P, H) \oplus \aut_v(\P, K) \simeq \text{Lie}\H \oplus \text{Lie}\K$,
 and $\xi^v \in \aut_v(\P)$ s.t $\xi^v=\chi^v+\upsilon^v$ is generated by  $\xi=\chi+\upsilon$.
 
 In preparation for the next proposition, consider that:
 \begin{align}
 \label{result2}
 \tfrac{d}{d\tau} \, C\left(  e^{\tau \chi }\right)e^{\tau \upsilon} \,\big|_{\tau=0} 
                       = \tfrac{d}{d\tau} \, C\left(  e^{\tau \chi }\right)e\,\big|_{\tau=0} \,e_K + e_G \,\tfrac{d}{d\tau} \, e^{\tau \upsilon} \,\big|_{\tau=0}
                       = dC_{|e}(\chi) + \upsilon.
 \end{align}
 
 \begin{prop} %%%%%%%%%%%%%%%%%%%%%%%%%%%%%%%%%%%%%%%%%%%%%%%%
 The infinitesimal mixed gauge transformations of the connection and $C$-tensorial forms are, 
 \begin{align}
 \label{inf-mixed-activeGT}
 \begin{split}
 L_{\xi^v} \omega & = d\left( dC_{|e}(\chi) + \upsilon \right)+ \left[ \omega, dC_{|e}(\chi) + \upsilon \right]= L_{\chi^v} \omega+ L_{\upsilon^v} \omega, \\[1.5mm]
  L_{\xi^v} \alpha  &= -\rho_*\left(  dC_{|e}(\chi) + \upsilon \right) \alpha=  L_{\chi^v} \alpha+ L_{\upsilon^v} \alpha.
 \end{split}
 \end{align}
 \end{prop}
\begin{proof}%%%%%%%%%%%%%%%%%%%
It goes exactly as in the proof of proposition \ref{infinitesimal-activeGT}, using proposition \ref{mixedGTconnection}: 
\begin{align*}
L_{\xi^v} \omega \defeq& \tfrac{d}{d\tau}  \left( \Phi \circ \Xi \right)_\tau^* \omega \,\big|_{\tau=0}
                                    = \tfrac{d}{d\tau}  \left[ C\left(  e^{\tau \chi }\right)e^{\tau \upsilon} \right]\- \omega \left[ C\left(  e^{\tau \chi }\right)e^{\tau \upsilon} \right] + 
                                         \left[ C\left(  e^{\tau \chi }\right)e^{\tau \upsilon} \right]\- d\left[ C\left(  e^{\tau \chi }\right)e^{\tau \upsilon} \right] \big|_{\tau=0}, \\
                           =&  -\tfrac{d}{d\tau} C\left( e^{\tau\, \chi }\right)e^{\tau \upsilon}  \big|_{\tau=0} \, \omega + \omega \,\tfrac{d}{d\tau} C\left( e^{\tau\, \chi }\right)e^{\tau \upsilon}  \big|_{\tau=0} + d\, \tfrac{d}{d\tau} C\left( e^{\tau\, \chi }\right)e^{\tau \upsilon}  \big|_{\tau=0}.            
\end{align*}
The result follows from \eqref{result2}. Idem for $\alpha \in \Omega^\bullet_\text{tens}\big(\P, C(H)\!\rtimes \!K\big)$:
\begin{align*}
L_{\xi^v} \alpha \defeq& \tfrac{d}{d\tau}  \left( \Phi \circ \Xi \right)_\tau^* \alpha \,\big|_{\tau=0}=  \tfrac{d}{d\tau} \, \rho\left[ C\left( e^{\tau\, \chi }\right)e^{\tau \upsilon}\right]\- \alpha \,  \big|_{\tau=0} 
			   = -\rho_*\left[\tfrac{d}{d\tau}  C\left( e^{\tau\, \chi }\right) e^{\tau \upsilon}  \big|_{\tau=0}  \right] \alpha.
\end{align*}
\end{proof}
The local version on $\U \subset \M$ of this, the linear couterpart of \eqref{local-active-mixedGT-connection}-\eqref{local-active-mixedGT-tensorial}, is obtained either by pullback or by $\delta A\defeq \tfrac{d}{d\tau}\, A^{\upgamma_\tau\upzeta_\tau} \big|_{\tau=0}$ and $\delta a\defeq \tfrac{d}{d\tau}\, a^{\upgamma_\tau\upzeta_\tau} \big|_{\tau=0}$. Either way, given  $\upupsilon\! \defeq \s^*\upsilon : \U \rarrow \text{Lie}K \in \text{Lie}\K_{\text{loc}}$, we get
 \begin{align}
 \label{local-inf-mixedGT}
 \begin{split}
 \delta A &= d\left( dC_{\s|e}(\upchi) + \upupsilon \right)+ \left[ A, dC_{\s|e}(\upchi) + \upupsilon \right]= \delta_\upchi A + \delta_\upupsilon A \\[1.5mm]
\delta a  &= -\rho_*\left(  dC_{\s|e}(\upchi) + \upupsilon \right) a =  \delta_\upchi a + \delta_\upupsilon a.
 \end{split}
 \end{align}
 
The linearised gluing properties, proposition \ref{passive-mixed-GTs}, are obtained analogously. One poses $\ell_\tau(x)=e^{\tau \nu(x)}$, with $\nu: \U \rarrow LieK$, and defines: 
\begin{align}
 \begin{split}
\delta A &\defeq \tfrac{d}{d\tau}\, A'  \big|_{\tau=0} = d\left( dC_{\s|e}(\lambda) + \nu \right) + \left[ A, dC_{\s|e}(\lambda) + \nu \right]=\delta_{\!\lambda} A + \delta_\nu A,  \\[0.5mm]
\delta a  &\defeq \tfrac{d}{d\tau}\, a' \big|_{\tau=0} = -\rho_*\left( dC_{\s|e}(\lambda) +\nu\right)  a= \delta_{\!\lambda} a + \delta_\nu a. 
 \end{split}
\end{align}

\subsection{A word on BRST} %%%%%%%%%%%%%%%%%%%%%%%%%%%%%%%%%%%%%%%%%%%%%
\label{A word on  BRST} %%%%%%%%%%%%%%%%%%%%%%%%%%%%%%%%%%%%%%%%%%%%%%%%

The BRST formalism is widely used in physics as an efficient algebraic way to handle infinitesimal gauge transformations. The BRST cohomology is rich, and allows notably  to classify viable Lagrangian and gauge anomalies (stemming from gauge symmetry breaking),  \cite{Dubois-Violette1987, Bertlmann} . The geometric origin of its heuristic rules has been explored by several authors, and the most satisfying view holds that it is grounded in the infinite dimensional analogue of the  Chevalley-Eilenberg construction \cite{Chevalley-Eilenberg} associated to the gauge group. The so-called ghost being essentially the Maurer-Cartan form on $\H$, while the BRST operator is the de Rham differential on the infinite dimensional  group manifold.  See \cite{Bonora-Cotta-Ramusino, DeAzc-Izq}. 
Here we  give the minimal BRST presentation reproducing the above  linearised local active gauge transformations.

One considers that all objects have a new grading in addition to the de Rham form degree, the so-called ghost degree. The forms $A, F, \phi$ are attributed ghost degree $0$, while $\upchi$ and $\upupsilon$ now stands for ghost fields\footnote{Which are generic place holders for specific elements of Lie$\H$ and Lie$\K$, i.e. they are the Maurer-Cartan forms on $\H$ and $\K$ respectively.} and have ghost degree $1$. Let us denote for simplicity $dC_{\s|e}(\upchi)\rdefeq c(\upchi)$. The BRST differential, noted $s$, increases by one unit the ghost number and is such that $sd+ds=0$ and $s^2=0$. So that $(d+s)^2=0$. It acts as
\begin{align*}
sA= - d  c(\upchi)   - \left[ A, c(\upchi)  \right], \qquad   sF &= [F, c(\upchi)] \qquad   s\phi=-\rho_*(c(\upchi))\phi, \\[1.5mm]
\text{and } \quad sc(\upchi)&=-\tfrac{1}{2}[c(\upchi), c(\upchi)].
\end{align*}
The last equality stems from requiring $s^2=0$ on $A$, $F$ or $\phi$. In the mixed case, just replace $c(\upchi)$ above by the total ghost $c(\upchi)+\upupsilon$. It is clear by linearity that  the total BRST operator splits accordingly as $s=s_\upchi+s_\upupsilon$. The last equality above in particular nicely encapsulates the relations  $s_\upupsilon \upupsilon =-\tfrac{1}{2}[\upsilon, \upupsilon]$ and $s_\upchi c(\upchi) =-\tfrac{1}{2} [c(\upchi) , c(\upchi) ]$ - ensuring that we have BRST subalgebras for both sectors - as well as $s_\upchi \upupsilon =0$ and $s_\upupsilon c(\upchi) = -[c(\upchi) , \upupsilon]$ which are the BRST versions of the fourth equality \eqref{eq-GTmaps} and of  \eqref{K-GT-Cgamma}.

\section{Twisted Cartan connection} %%%%%%%%%%%%%%%%%%%%%%%%%%%%%%%%%%%%%%%%%%%%%
\label{Twisted Cartan connection} %%%%%%%%%%%%%%%%%%%%%%%%%%%%%%%%%%%%%%%%%%%%%%%%

Cartan connections are ancestors to Ehresmann connections. Ehresmann gave them their first recognisably modern definition while proposing his own generalisation (see \cite{Ehresmann50} also \cite{Kobayashi1957} and reference therein). We refer to \cite{Sharpe} and \cite{Cap-Slovak09} for modern in depth treatments of the subject. 

Given a principal bundle $\P(\M, H)$, and Lie$G'\!\supset$ Lie$H$, a Cartan connection is a differential $1$-form $\varpi \in \Omega^1(\P, \text{Lie}G')$ enjoying the same two defining properties of an Ehresmann connection, $\varpi_p(X_p^v)=X \in $ Lie$H$  and $R^*_h\varpi_{ph} = \Ad_{h^{-1}} \varpi_p$, but it is also required that  $\varpi_p:T_p\P \rarrow$ Lie$G'$ be a linear isomorphism $\forall p\in \P$. That is, the Cartan connection defines an absolute parallelism on $\P$. This means that dim$\M=$ dim Lie$G'\!/$Lie$H$, and more precisely that given the projection $\tau: \text{Lie}G' \rarrow \text{Lie}G'\!/ \text{Lie}H$, $\tau(\varpi)$ is a soldering form (whose local version  induces a metric structure on $\M$ if a bilinear form $\eta: \text{Lie}G'\!/ \text{Lie}H \times \text{Lie}G'\!/ \text{Lie}H \rarrow \RR$ is given). 
This is as well expressed by the fact that the Cartan connection induces a soldering, i.e. an isomorphism of vector bundles : $T\M \simeq \P\times_{\Ad(H)} \text{Lie}G'\!/\text{Lie}H$. These facts make Cartan connections especially well suited to  gauge theories of gravity. 
 
 The curvature is still given by Cartan's structure equation $\Omega=d\varpi+ \tfrac{1}{2}[\varpi, \varpi]\in \Omega^2_{\text{tens}}(\P, \text{Lie}G')$, and $\Theta\defeq\tau(\Omega)$ is the torsion. Gauge transformations are given by the action of $\Aut_v(\P, H) \simeq \H$, as it would for an Ehresmann connection: $\varpi^\gamma\defeq \Phi^*\varpi=\gamma\- \varpi \gamma + \gamma\-d\gamma$ and  $\Omega^\gamma\defeq \Phi^*\Omega=\gamma\- \Omega \gamma$.
\medskip

Since Cartan connections can be seen as special cases of Ehresmann connections, and since twisted connections are a particular generalisation of Ehresmann connections, we would like to define a special class of twisted connections that would give an acceptable generalisation of Cartan connections. We might call these twisted Cartan connections. There might be more than one clever way to define such a class. To emulate the desirable properties of Cartan connections, especially regarding the gauge treatment of gravity, the following desiderata seem sufficient. 

We consider again a principal bundle $\P(\M, H)$ and a cocycle for the group action of $H$ on $\P$, $C: \P \times H \rarrow G$. We demand that the twisted Cartan connection  be $\varpi \in \Omega^1(\P, \text{Lie}G')$, where the Lie algebra Lie$G' \supset$ Lie$G$  is such that  dim~Lie$G'\!/$Lie$G$ = dim$\M$, and satisfies:
\begin{align}
\varpi_p(X^v_p)&=\tfrac{d}{d\tau} C_p\left( e^{\tau X} \right)\big|_{\tau=0}=dC_{p|e}(X) \in \text{Lie}G, \quad \text{ for } X \in \text{Lie}H,      \tag{\Rmnum{1}} \label{1st Cond} \\[1mm]
R^*_h \varpi_{ph}&= C_p(h)\- \varpi_p\, C_p(h) + C_p(h)\-d C(h)_{|p},             \tag{\Rmnum{2}} \label{2nd Cond} \\[1mm]
\ker \varpi_p &= \emptyset, \quad \forall p\in \P, \text{and} \quad \text{Lie}G'\!/\text{Lie}G \subset \im \varpi_p.          \tag{\Rmnum{3}}  \label{Cartan Cond}
\end{align}
Thus, the first two axioms are those of a twisted connection as originally defined, but the third is a weakened version of the one satisfied by a Cartan connection: We only require  that $\varpi$ be an injection, not a linear isomorphism. Still, it is sufficient to define a soldering. 

\begin{prop} %%%%%%%%%%%%%%%%%%%%%%%%%%%%%%%%%
\label{C-soldering}
A twisted Cartan connection on $\P(\M, H)$ induces a soldering, i.e. a vector bundle isomorphism $T\M \simeq \P \times_{\Ad_{C(H)}} \text{Lie}G'\!/\text{Lie}G$.
\end{prop}
 \begin{proof} %%%%%%%%%%%%%%
Consider the following diagram 

\begin{center}
\begin{tikzcd}[column sep=large, row sep=large]

& V_p\P
 \arrow[r, "\varpi_p", "\simeq"' ] \arrow[d, ""] 
& \im\varpi_p \cap \text{Lie}G
 \arrow[d, ""] \\

&T_p\P
 \arrow[r,  "\varpi_p", "\simeq"'] \arrow[d,  "\pi_*"] 
& \im\varpi_p \cap \text{Lie}G' \arrow[d, "\tau"] \\

&T_x\M \arrow[r, "", dashrightarrow] 
& \parbox{2cm}{$\im\varpi_p~\cap~\text{Lie}G'/\text{Lie}G  \\[1mm] =~\text{Lie}G'/\text{Lie}G$}

\end{tikzcd}  
\end{center}
The first columns is clearly a short exact sequences, and the  second column is one also because of \eqref{Cartan Cond}. The maps in the two upper rows are isomorphisms due to \eqref{1st Cond} and \eqref{Cartan Cond}, so there must be an isomorphism in in the bottom row that makes the diagram commute. Such an isomorphism depends on $p \in \P$ and is given only up to the $C$-twisted adjoint action of $H$. Indeed, denote  $\beta_p : T_x\M \rarrow \text{Lie}G'/\text{Lie}G$, and consider $X_p \in T_p\P$ projecting as $\pi_*X_p=X_x \in T_x\M$. The commutativity of the bottom square means $\beta_p (X_x)=\beta_p(\pi_*X_p)=\tau \circ \varpi_p(X_p)$. But also, $\beta_{ph}(X_x)=\beta_{ph}(\pi_*R_{h*}X_p)=\tau\circ \varpi_{ph}(R_{h*}X_p)=\tau \circ \Ad_{C_p(h)\-}\varpi_p(X_p)=\Ad_{C_p(h)\-} \tau \circ \varpi_p(X_p)= \Ad_{C_p(h)\-}\beta_p(X_x)$, where \eqref{2nd Cond} is used.
Now, if one defines the map,  
\begin{align*}
\iota :\, \,  \P \times \text{Lie}G'\!/\text{Lie}G \quad &\rarrow \quad T\M  \\
            \left( p, \mathpzc{t}\right)  \qquad   &\mapsto \, \, \left(\pi(p),\, \beta_p\-(\mathpzc{t})\right),
\end{align*}
it appears that: $\iota\left(ph, \Ad_{C_p(h)\-} \mathpzc{t}\right)=\left( \pi(ph), \, \beta_{ph}\-\left( \Ad_{C_p(h)\-}\mathpzc{t} \right) \right)=\left( \pi(p), \, \beta_{p}\-(\mathpzc{t}) \right)=\iota(p, \mathpzc{t})$. Thus,  we have actually a bundle map $\iota : \P \times_{\Ad_{C(H)}} \text{Lie}G'\!/\text{Lie}G \rarrow T\M$, which is a bundle equivalence because it  covers the identity  on $\M$ and is an isomorphism of fibers. 
 \end{proof}
 
A corollary is that there is a bijective correspondance between vector fields $X\in \Gamma(T\M)$  and $\Ad_{C(H)}$-equivariant maps $\vphi: \P \rarrow \text{Lie}G'\!/\text{Lie}G$.
   
  Relatedly, it is clear that a twisted Cartan connection defines a soldering form $\theta\defeq \tau \circ \varpi \in \Omega^1(\P, \text{Lie}G'\!/\text{Lie}G)$ which is  both horizontal and
  $\Ad_{C(H)}$-equivariant: $\theta(X^v)=0$ and $R^*_h \theta= \Ad_{C(h)\-} \theta$. That is, the soldering form is tensorial, $\theta \in \Omega^1_\text{tens}\big(\P, C(H)\big)$. 
 Locally, on $\U \subset \M$ endowed with a section $\s$, its pullback is $e\defeq \s^*\theta \in \Omega^1(\U, \text{Lie}G'\!/\text{Lie}G)$. Given a coordinate system $\{x^\mu \}$ on $\U$ and a basis $\{ u_a\}$ of $\text{Lie}G'\!/\text{Lie}G$,  $e={e^a}_\mu\, dx^\mu \otimes u_a$ where %, in the parlance of physicists,
  ${e^a}_\mu$ is a vielbein field. If there is a symmetric non-degenerate bilinear form $\eta: \text{Lie}G'\!/ \text{Lie}G \times \text{Lie}G'\!/ \text{Lie}G \rarrow \RR$, a metric $g: \Gamma(T\U) \times \Gamma(T\U) \rarrow \RR$ is induced via $g\defeq \eta \circ e$. If $\eta$ is $\Ad_G$-invariant (possibly up to a non-vanishing multiplicative factor), the local metric $g$ extends to a (conformal) metric on $\M$. 
  \medskip
  
  Let $V$ be a $(\text{Lie}G', G)$-module, i.e. it supports an action of Lie$G'$ via $\rho'_*$, and an action of $G$ via $\rho$ which is such that $\rho_*=\rho'_{*|\text{Lie}G}$. The space $\Omega^\bullet_\text{tens}\big(\P, C(H)\big)$  of $V$-valued twisted tensorial forms is defined as in section \ref{Covariant differentiation}, in particular their equivariance is $R^*_h \alpha=\rho\left[C(h)\right]\-\alpha$. 
 The twisted Cartan connection, on account of \eqref{1st Cond} and \eqref{2nd Cond}, induces an exterior covariant derivative on this space defined as usual by $D\defeq d + \rho'_*(\varpi)$.

The curvature of the twisted Cartan connection is again given by $\Omega=d\varpi+ \tfrac{1}{2}[\varpi, \varpi]\in \Omega^2(\P, \text{Lie}G')$,  and the torsion is $\Theta\defeq\tau(\Omega)$.  The curvature is a $\Ad_{C(H)}$-tensorial form so $D$ acts on it, but trivially so, which gives a Bianchi identity: $D\Omega =d\Omega + [\varpi, \Omega]=0$.
\medskip
 
Gauge transformations are of course given by the action of $\Aut_v(\P, H) \simeq \H$. On account of \eqref{1st Cond} and \eqref{2nd Cond}, the twisted Cartan connection transforms as: $\varpi^\gamma\defeq \Phi^*\varpi=C(\gamma)\- \varpi\, C(\gamma) + C(\gamma)\- dC(\gamma)$. 
We already know that for  twisted tensorial forms $\alpha^\gamma\defeq \Phi^*\alpha=\rho\left[ C(\gamma)\right]\-\alpha$.  In particular, this gives the gauge transformations of the curvature, the torsion and the soldering form:
  $\Omega^\gamma=C(\gamma)\- \Omega\, C(\gamma)$, 
  $\Theta^\gamma=C(\gamma)\- \Theta\, C(\gamma)$ and 
  $\theta^\gamma=C(\gamma)\- \theta\, C(\gamma)$.

\subsection{Mixed Cartan connection}%%%%%%%%%%%%%%%%%%%%%%%%%%%%%%%%%%%%%%%%%%%%%
\label{Mixed Cartan connection} %%%%%%%%%%%%%%%%%%%%%%%%%%%%%%%%%%%%%%%%%%%%%%%%

 The previous construction easily extends to the mixed case.
Consider the principal bundle $\P(\M, H\times K)$ with a cocycle  $C: \P \times H \rarrow G$. A mixed Cartan connection  is $\varpi \in \Omega^1(\P, \text{Lie}G')$, with Lie$G' \supset$ Lie$(G\! \rtimes\! K)$   such that  dim~Lie$G'\!/$Lie$(G\! \rtimes\! K)$ = dim$\M$,  satisfying:
\begin{align}
&\varpi_p(X^v_p+ Y^v_p)=\tfrac{d}{d\tau} C_p\left( e^{\tau X} \right)\big|_{\tau=0}+ Y=dC_{p|e}(X) + Y \in \text{Lie}(G\! \rtimes\! K),      \tag{\Rmnum{1}$^\star$} \label{1st Cond Mixed}  \\[1mm]
&R^*_{hk}\varpi_p{phk}=R^*_{kh} \varpi_p{pkh} = [ C_p(h)k]\- \varpi_p\, [C_p(h)k] + [C_p(h)k]\- d [ C(h)k  ]_{|p},             \tag{\Rmnum{2}$^\star$} \label{2nd Cond Mixed}  \\[1mm]
&\ker \varpi_p = \emptyset, \quad \forall p\in \P, \text{and} \quad \text{Lie}G'\!/\text{Lie}(G\! \rtimes\! K) \subset \im \varpi_p.          \tag{\Rmnum{3}$^\star$}  \label{Cartan Cond Mixed}
\end{align}
The first two axioms are those of a mixed connection,  the third is a slightly extended version of the one satisfied by a twisted Cartan connection and ensures that, mutadis mutandis, the result of Proposition \ref{C-soldering} generalises as $T\M \simeq \P \times_{\Ad_{C(H)\rtimes K}} \text{Lie}G'\!/\text{Lie}(G\!\rtimes\! K)$. As above, the projection $\tau : \text{Lie}G' \rarrow \text{Lie}G'\!/\text{Lie}(G\!\rtimes\! K)$, allows to define the soldering form $\theta \in \Omega^1\big(\P, C(H)\!\rtimes\! K \big)$.
Given a $(\text{Lie}G', G)$-module $V$, the space $\Omega^\bullet_\text{tens}\big(\P, C(H)\!\rtimes\! K\big)$  of $V$-valued mixed tensorial forms is defined as in section \ref{Mixed vector bundles and tensorial forms}. A mixed Cartan connection, on account of \eqref{1st Cond Mixed} and \eqref{2nd Cond Mixed}, induces an exterior covariant derivative $D\defeq d + \rho'_*(\varpi)$ on this space. 
The curvature  $\Omega \in \Omega^2_\text{tens}\big(\P, C(H)\!\rtimes\! K\big)$ satisfies Bianchi, $D\Omega =0$,  and the torsion is still $\Theta\defeq\tau(\Omega)$. 
  On account of \eqref{1st Cond Mixed} and \eqref{2nd Cond Mixed}, the gauge transformation of the mixed Cartan connection under $\Aut_v(\P) \simeq \H\times \K$ is, $\varpi^{\gamma\zeta}\defeq \Psi^*\varpi=[C(\gamma)\zeta]\- \varpi\, [C(\gamma)\zeta] + [C(\gamma)\zeta]\- d[C(\gamma)\zeta]$. 
For mixed tensorial forms it is given by \eqref{H-K_GTtensorial} of Proposition \ref{mixedGTconnection},  which gives the gauge transformations of the curvature, the torsion and the soldering form as special cases.

\subsection{Reductive and parabolic mixed Cartan geometries} %%%%%%%%%%%%%%%%%%%%%%%%%%%%%%%%%%%%%%%%%%%%%
\label{Reductive and parabolic mixed Cartan geometries} %%%%%%%%%%%%%%%%%%%%%%%%%%%%%%%%%%%%%%%%%%%%%%%%

Following a well established nomenclature for Cartan geometries \cite{Sharpe, Cap-Slovak09}, we signal and briefly characterise two notable classes of mixed Cartan geometries. 

If there is an $\Ad_G$- invariant decomposition Lie$G'\!=$ Lie$G +V^n$, we  call the mixed Cartan geometry \emph{reductive}. In this case, there is clean splitting of the mixed Cartan connection $\varpi=\omega+\theta$ into a mixed connection $\omega$ and a mixed soldering form $\theta$. The curvature splits in the same way as the sum of the curvature of $\omega$ and the torsion, given by $\Theta=d\theta + [\omega, \theta]$.

If $G'$ is semi-simple and $G$ is a parabolic subgroup corresponding to a $|k|$-grading of Lie$G'\!=\!\bigoplus\limits_{-k \leq i\leq k} \LieG'_i$, s.t. $[\LieG'_i, \LieG'_j]\subset \LieG'_{i+\!j}$, the associated mixed Cartan geometry will be called \emph{parabolic}. Both the mixed Cartan connection and its curvature split along the $|k|$-grading.

\section{Twisted gauge theories} %%%%%%%%%%%%%%%%%%%%%%%%%%%%%%%%%%%%%%%%%%%%%
\label{Twisted gauge theories} %%%%%%%%%%%%%%%%%%%%%%%%%%%%%%%%%%%%%%%%%%%%%%%%

Concerning  applications of this twisted/mixed geometry to physics, the local representatives are to be regarded as generalised gauge fields: They satisfy  the gauge principle while being of a different geometric nature than standard gauge fields. Given the usual ingredients, it is easy to come up with Lagrangian functionals that specify twisted gauge theories in a way that exactly parallels the construction of Yang-Mills or gravity gauge theories. 

To give prototypical examples, consider the Killing form on Lie$G$ or Lie$G'$ - which, for special linear groups, reduces to the trace operator, $B(XY)=\Tr(XY)$ for $X, Y \in \text{Lie}G^{(\prime)}$. Suppose also that the representation maps  $\rho'_*$, $\rho$ of the $(\text{Lie}G', G)$-module $V$ are unitary for some bilinear form $\langle\, , \, \rangle$ on $V$. Finally, suppose $\M$ endowed with a metric $g$, giving a Hodge operator $* : \Omega^\bullet (\M) \rarrow \Omega^{m-\bullet}(\M)$. Then, one can write the gauge-invariant Lagrangian,
\begin{align*}
L(A, \phi)= \tfrac{1}{2} \Tr\left( F \w *F\right) + \langle D\phi, *D\phi\rangle + U\left( \langle\phi,\phi\rangle\right)*\!\1,
\end{align*} 
 where $U$ is some potential. Or, if  $\phi$ is a spinor (fermions), so that we note it $\psi$ instead, we can write: 
 \begin{align*}
L(A, \psi)= \tfrac{1}{2} \Tr\left( F \w *F\right) + \langle \psi, \slashed{D}\psi\rangle - m \langle \psi, * \psi \rangle. 
\end{align*} 
 with $\slashed D= \gamma \w * D\psi$ where $\gamma=\gamma_a\, {e^{a}}_\mu\, dx^{\,\mu} = \gamma_{\mu\,} dx^{\,\mu}$ is the Dirac gamma matrices-valued $1$-form (with ${e^a}_\mu$ the vielbein associated to the metric $g$).
 It  describes the dynamics of the twisted gauge potential $A$ coupled to a twisted Dirac field of mass $m$. Quite obviously, the field equations for $A$ and $\psi$ obtained from the action $S(A, \phi)=\int_\M L(A, \phi)$, are respectively  Yang-Mills' equation sourced by $\psi$ - $D*F=J(\psi)$ - and Dirac's equation - $(\slashed D- m)\psi=0$.

\section{Applications: Conformal tractors, twistors, and anomalies in QFT} %%%%%%%%%%%%%%%%%%%%%%%%%%%%%%%%%%%%%%%%%%%%%
\label{Applications: Conformal tractors, twistors, and anomalies in QFT} %%%%%%%%%%%%%%%%%%%%%%%%%%%%%%%%%%%%%%%%%%%%%%%%

Tractors are sections of a $n\!+\!2$-dim real vector bundle over a conformal $n$-manifold $\left(\M, [g] \right)$, the tractor bundle $\T$. These sections have a special covariance under Weyl rescaling of the metric, that is under change of metric within the conformal class.  This bundle is endowed with a covariant derivative often called the tractor connection $\nabla^\T$. As a matter of fact, $\big( \T, \nabla^\T \big)$ is the basis of a tractor calculus which is the analogue for conformal manifolds of the tensorial calculus on (pseudo) Riemannian manifolds \cite{Bailey-et-al94}. Recently, it has found some applications in physics, especially General Relativity \cite{curry_gover_2018}.

Twistors are sections of a $4$-dim complex vector bundle over a conformal $4$-manifold, the (local) twistor bundle $\mathbb{T}$. They have special covariance under Weyl rescaling, so $\mathbb{T}$ is endowed with a covariant derivative $\nabla^\mathbb{T}$ - the twistor connection - and $\big(\mathbb{T}, \nabla^\mathbb{T} \big)$ is the basis of a twistorial calculus which is the analogue of spinorial calculus on (pseudo) Riemannian manifolds \cite{Penrose-Rindler-vol1, Penrose-Rindler-vol2}. Penrose originally devised  twistors with physics purpose in mind, and nowadays they found renewed relevance in string theory and loop quantum gravity \cite{Atiyah_et_al2017}.  

Twistors are the spin version of $n=4$ tractors. Initially, and still in many standard presentations, both tractors and twistors are constructed in a similar way: bottom-up, by prolongation of  differential equations. Starting on a conformal manifold, one defines the so-called Almost Einstein (AE) equation and twistor equation, which are prolonged into closed linear systems. These are rewritten as linear operators acting on multiplets of variables (the parallel tractors and global twistors respectively). The special covariance of these multiplets under Weyl rescaling is given by definition and commute with the action of their respective  associated linear operators, which are thus called tractor and twistor connections. The multiplets are interpreted as parallel sections of vector bundles, the tractor and (local) twistor bundles. One can consult \cite{Bailey-et-al94, curry_gover_2018} for the  details of this procedure in the tractor case, and the reference text \cite{Penrose-Rindler-vol2} for the twistor case (or \cite{Bailey-Eastwood91} where the case of paraconformal (PCF) manifolds is treated in a very similar fashion).

The fact that $\T$ and $\mathbb{T}$ had a close relationship with the principal bundle of the conformal Cartan geometry \cite{Sharpe} did not go unnoticed \cite{Friedrich77, Bailey-et-al94}. In the modern treatment of Cartan geometries, the term tractor assume a  broader meaning: Bundles associated to a Cartan geometry via a restriction of a representation of $G'\!\supset\!H$ are termed  tractor bundles, and for parabolic Cartan geometries the covariant derivatives on these bundles induced by a Cartan connection are called tractor connections (see \cite{Cap-Slovak09} sections 1.5.7 and 3.1.22). But then the original meaning is only recovered via special sections of the underlying Cartan principal bundle known as Weyl structures (\cite{Cap-Slovak09}, section~5.1). 
In the cases under consideration, the bottom-up procedure via prolongation described above has been occasionally deemed more explicit \cite{Bailey-Eastwood91}, or more intuitive and direct \cite{Bailey-et-al94} than the viewpoint in terms of vector bundles associated to the conformal Cartan geometry. But this constructive approach requires a rather significant amount of computation. 

An alternative top-down approach is proposed in \cite{Attard-Francois2016_I, Attard-Francois2016_II} where $\T$ and $\mathbb{T}$  are explicitly constructed via  gauge symmetry reduction of the real and complex vector bundles naturally associated to the conformal Cartan geometry, while $\nabla^\T$ and $\nabla^\mathbb{T}$ are shown to be induced by the reduced Cartan connection. 
This alternative constructive procedure relies on a general scheme of gauge symmetry reduction known as the \emph{dressing field method} \cite{Attard_et_al2017}. It is  computationally more economical than the bottom-up approach via prolongation, and allows to generalise tractors and twistors to manifolds with torsion. It preserves the insight of the more abstract modern treatment - by articulating how $\T$ and $\mathbb{T}$ are associated to the conformal Cartan geometry - yet it is more user friendly.\footnote{For an explanation of the relation between the dressing fiel method and the notion of Weyl structure, one can consult appendix A in \cite{Francois2019}. } 
 Finally, most relevant for the present work, this method hints at the fact  that tractors and twistors can be seen as simple, and somewhat degenerated, examples of the twisted/mixed geometry elaborated above. This last point is elaborated in the next two subsections.
\medskip

After a third section where we comment briefly on conformal gravity, 
in a fourth and final subsection we show how the twisted geometry arises naturally in the very definition of anomalies in quantum gauge field theory, and  therefore ought be relevant to their study.
%\bigskip

\subsection{The case of conformal tractors}%%%%%%%%%%%%%%%%%%%%%%%%%%%%%%%%%
\label{The case of conformal tractors}%%%%%%%%%%%%%%%%%%%%%%%%%%%%%%%%%

%The claim is proven simply by display, so let us proceed. 
Consider a conformal $4$-manifold $(\M, [g])$. A tractor  is a map $\vphi: \U \subset \M \rarrow \RR^6$, $x \mapsto \vphi(x) =\begin{psmallmatrix} \rho \\ \ell \\ \s \end{psmallmatrix}$ with $\ell = \ell^a \in \RR^4$, $\rho \in \RR$ and $\s \in \RR_*$. The (generalised) tractor connection is 
$\varpi=\begin{psmallmatrix}  0 & P & 0  \\
				              e & A & P^t \\
				              0 & e^t & 0
 \end{psmallmatrix} $,
 where $A \in \Omega^1\big(\U, \so(1,3)\big)$, $e={e^a}_\mu dx^{\,\mu} \in \Omega^1\big(\U, \RR^4 \big)$ is the soldering form, and $P  \in \Omega^1\big(\U, \RR^{4*} \big)$. The operation $|^t$ is the transposition via the Minkowski metric $\eta$.  
 The tractor connection enters the definition of the tractor derivative $D\vphi = d\vphi + \varpi \vphi$. 
 The tractor curvature is 
 $\b\Omega=\begin{psmallmatrix}  f & \sC & 0  \\
				              \sT & \sW & \sC^t \\
				              0 & \sT^t & -f
 \end{psmallmatrix} $,
 where $\sW \in \Omega^2\big(\U, \so(1,3)\big)$, $\sT \in \Omega^2\big(\U, \RR^4 \big)$ is the torsion,  $\sC  \in \Omega^2\big(\U, \RR^{4*} \big)$, and $f \in \Omega^2\big(\U, \RR \big)$ is such that $f_{ab}=P_{[ab]}$.
 
 If one imposes the conditions $\text{Ricc}(\sW)\defeq{\sW^a}_{bac}=0$ and $\sT=0$, several consequences follow. First $A$ is the Lorentz/spin connection (expressed as a function of the components of $e$). Then $f=0$ (by the Bianchi identify) and $P$ is the Schouten 1-form (with components the symmetric Schouten tensor). Finally, $\sW$ and $\sC$ are the Weyl and Cotton 2-forms. In this case, $\varpi$ is the standard tractor connection. 
 \medskip
 
 The (local) gauge group of Weyl rescalings is $\W\defeq \left\{ z: \U \in \M \rarrow \RR^+_*\, \big|\, z^{z'}={z'}\-zz'=z\right\}$. The Weyl gauge transformations of the above variables are obtained via elements of the form 
 \begin{align}
 \label{tractor-cocycle}
 C(z)= \begin{pmatrix} z & \Upsilon(z) & \sfrac{z\-}{2}\,\Upsilon(z)\Upsilon(z)^t \\  0 & \1 & z\-\Upsilon(z)^t \\ 0 & 0 & z\- \end{pmatrix}, \quad \text{where} \quad \Upsilon(z)=\Upsilon(z)_a\defeq z\-\d_\mu z \,\,{e^{\,\mu}}_a  \in \RR^{4*}.
 \end{align}
% \intertext{}
 We have indeed,
 \begin{equation}
 \label{Tractor-Weyl-GT}
 \begin{split}
& \vphi^z= C(z)\- \vphi, \qquad \varpi^z=C(z)\- \varpi \, C(z) + C(z)\-dC(z),  \\[1mm]
& (D\vphi)^z=D^z\vphi^z= C(z)\- D\vphi, \qquad \b\Omega^z=C(z)\- \b\Omega\, C(z).
 \end{split}
 \end{equation}
 It is easily verified that for $z,z' \in \W$ one has $C(z)^{z'}= C(z')\- C(zz')$ (because $e^z =ze$), which is an instance of \eqref{localGTCgamma}, local version  of the gauge transformation law of a cocycle \eqref{GTCgamma}. So that  \eqref{Tractor-Weyl-GT} is indeed a special case of \eqref{localGTconnection}%and \eqref{localGTtensorial} (or 
 -\eqref{localGTothers}, local version of the gauge transformations of twisted fields as described in Proposition \ref{ConnectionGT} and \ref{TensorialGT}. We conclude that tractor variables are indeed twisted gauge fields w.r.t. the Weyl gauge group $\W$. 
 %\medskip
 
 The (local) Lorentz gauge group is $\SO\defeq \left\{ \sS = \begin{psmallmatrix} 1 & 0 & 0 \\ 0 & S & 0 \\ 0 & 0 & 1 \end{psmallmatrix},\, S:\U \rarrow S\!O(1,3)\, \big|\, \sS^{\sS'}={\sS'}\-\sS \sS' \right\}$. The Lorentz gauge transformations of the tractor variables are, 
 \begin{align*}
& \vphi^\sS= \sS\- \vphi, \qquad \varpi^\sS=\sS\- \varpi\, \sS + \sS\-d\sS,  \\[1mm]
& (D\vphi)^\sS=D^\sS\vphi^\sS= \sS\- D\vphi, \qquad \b\Omega^\sS=\sS\- \b\Omega\, \sS.
 \end{align*}
 Which means that they are standard gauge fields w.r.t. $\SO$. By the way, it is easily verified that $C(z)^\sS = \sS\- C(z)\, \sS$ (because $e^S=S\-e$), which is a special case of  \eqref{local-K-GT-Cgamma}, local version of \eqref{K-GT-Cgamma}. This ensures that the actions of $\W$ and $\SO$ on the tractor variables commute and that we have, 
  \begin{align*}
& \vphi^{\,z\sS}= [C(z)\sS]\- \vphi, \qquad \varpi^{ \,z\sS}=[C(z)\sS]\- \varpi \, [C(z)\sS] + [C(z)\sS]\-d[C(z)\sS],  \\[1mm]
& (D\vphi)^{\,z\sS}=D^{\,z\sS}\vphi^{\,z\sS}= [C(z)\sS]\- D\vphi, \qquad \b\Omega^{\,z\sS}=[C(z)\sS]\- \b\Omega\, [C(z)\sS],
 \end{align*}
 that is a special case of  \eqref{local-active-mixedGT-connection}-\eqref{activeGT-mixed-others}, local version of the gauge transformations for mixed gauge fields as given in Proposition  \ref{mixedGTconnection}. We conclude that tractor variables are mixed gauge fields w.r.t. the gauge group $\W \times \SO$.
 
 \bigskip

From the above local construct we can attempt to recover the global twisted geometry that it stems from.
Consider the bundle $\P\big(\M, \,W \times S\!O(1,3)\big)$, with $W\defeq %\bigg\{ \begin{psmallmatrix} z & 0 & 0 \\  0 & \1 & 0 \\ 0 & 0 & z\- \end{psmallmatrix}\big| \,z \in 
\RR^+_*$, %and $S\!O\defeq \left\{ \begin{psmallmatrix} 1 & 0 & 0 \\  0 & S & 0 \\ 0 & 0 & 1 \end{psmallmatrix}\big| \,S \!\in S\!O(1,3)\right\}$. 
as well as a cocycle map  $C: \P \times W \rarrow G$ where the target group is defined as,
$G= \left\{ \begin{psmallmatrix} 1 & r & \sfrac{1}{2}\,rr^t \\  0 & \1 & r^t \\ 0 & 0 & 1 \end{psmallmatrix} \rtimes \begin{psmallmatrix} z & 0 & 0 \\  0 & \1 & 0 \\ 0 & 0 & z\- \end{psmallmatrix} = \begin{psmallmatrix} z & r & \sfrac{z\-}{2}\,rr^t \\  0 & \1 & z\-r^t \\ 0 & 0 & z\- \end{psmallmatrix}\, \big| \, r\in \RR^{4*}, z\in W  \right\}^{\vphantom{|}}_{\vphantom{|}}$.
For $z \in \W$, the local cocycle  \eqref{tractor-cocycle} is a map $C(z) : \U \rarrow G$.
 Consider also the (inner) semidirect product group $G \rtimes S\!O \defeq \left\{ \begin{psmallmatrix} z & r & \sfrac{z\-}{2}\,rr^t \\  0 & \1 & z\-r^t \\ 0 & 0 & z\- \end{psmallmatrix} \rtimes \begin{psmallmatrix} 1 & 0 & 0 \\  0 & S & 0 \\ 0 & 0 & 1 \end{psmallmatrix}= \begin{psmallmatrix} z & rS & \sfrac{z\-}{2}\,rr^t \\  0 & S & z\-r^t \\ 0 & 0 & z\- \end{psmallmatrix} \,\big|\, S\in S\!O(1,3) \right\}^{\vphantom{|}}$, where the group morphism $S\!O \rarrow \Aut(G)$ defining the semidirect structure is $S \mapsto \text{Conj}(S)$. 
 Its Lie algebra is 
 Lie$(G \rtimes S\!O)= \left\{ \begin{psmallmatrix} \epsilon & \iota & 0 \\ 0 & s & \iota^t \\ 0 & 0 & -\epsilon \end{psmallmatrix} \, \big|\, \epsilon \in \RR^+_*,  s\in  \so(1,3), \iota \in \RR^{4*} \right\}$. Clearly, the tractor connection takes value in the bigger Lie algebra 
 Lie$G'=\text{Lie}(G\rtimes SO)\oplus\RR^4=\left\{ \begin{psmallmatrix} \epsilon & \iota & 0 \\ \tau & s & \iota^t \\ 0 & \tau^t & -\epsilon \end{psmallmatrix} \, \big|\, \tau \in \RR^4 \right\}^{\vphantom{|}}_{\vphantom{|}}$, and dim Lie$G'\!/\text{Lie}(G\!\rtimes\! S\!O) =\text{dim}\M$. 
 Therefore, it appears that the tractor connection $\varpi$  is a local mixed Cartan connection as defined in section \ref{Mixed Cartan connection}.
 By the way, there is a decomposition of Lie$G'$ that shows it is graded: 
 Lie$G'= \LieG'_{-1}+\LieG'_0+\LieG'_1=\left\{  \begin{psmallmatrix} 0 & 0 & 0 \\ \tau & 0 & 0 \\ 0 & \tau^t & 0 \end{psmallmatrix}  + \begin{psmallmatrix} \epsilon & 0 & 0 \\ 0 & s & 0 \\ 0 & 0 & -\epsilon \end{psmallmatrix}  +  \begin{psmallmatrix} 0 & \iota & 0 \\ 0 & 0 & \iota^t \\ 0 & 0 & 0 \end{psmallmatrix} \, \big|\, \ldots \right\}^{\vphantom{|}}_{\vphantom{|}}$, $[\LieG'_i, \LieG'_j] \in \LieG'_{i+j}$, and Lie$(G\rtimes S\!O)=\LieG'_0+\LieG'_1$ is a parabolic subalgebra.
 So tractor geometry is an instance of parabolic mixed Cartan geometry. In particular the tractor bundle is a mixed  vector bundle associated to $\P$, $\T=\P \times_{C(W)\rtimes S\!O} \RR^6$, as defined in section \ref{Mixed vector bundles and tensorial forms}.
 
 Notice however that, in view of \eqref{tractor-cocycle}, the $\P$-dependance of the cocycle map $C: \P \times W \rarrow G$  is localised in the coefficients of the soldering form (${e^a}_\mu$) in the definition of $\Upsilon(z)$, which also contains a derivative of $z$. So, for constant $z \in W$ the $\P$-dependance vanishes, and the cocycle reduces to a group morphism $C: W \rarrow G$, $z\mapsto C(z)=\begin{psmallmatrix} z & 0& 0 \\ 0 & \1&0 \\ 0 & 0 & z\- \end{psmallmatrix}$. Then, at the level of $W\times S\!O$-equivariance, the $W$-twisted side of tractor geometry  seems to degenerate into a standard non-twisted case. Only at the level of $\W \times \SO$-action/gauge transformations does the cocycle structure, and the $\W$-twisted geometry, is manifest. Because of this, the fact that tractors are mixed gauge fields  and provide an instance of a new type of geometry might go unnoticed.\footnote{From a physicist's  point of view  it is clear that something unusual is going on, because gauge elements $C(z)$ of type \eqref{tractor-cocycle} clearly cannot come from the mere gauging of a Lie group.  } The same goes for twistors, as we now show.
 
 \medskip
%{\color{blue} What about the $\SO$-reduced case ? }

 \subsection{The case of local twistors}%%%%%%%%%%%%%%%%%%%%%%%%%%%%%%%%%
\label{The case of local twistors}%%%%%%%%%%%%%%%%%%%%%%%%%%%%%%%%%
 
 We follow closely the treatment given in the previous section.  
 But first, let us start by reminding some basic results and fix our notations. 
 
Denote Minkowski space by $\sM:=(\RR^4, \eta)$, and
consider $\left\{\s_a^{AA'}\right\}_{a=0, \ldots, 3}$, a basis of $2\times2$ hermitian matrices Herm$(2, \CC)=\left\{ M \in M_2(\CC)\ | \ M^*=M  \right\}$, where $*$ denote Hermitian transposition. There is a vector space isomorphism
$
\sM \rarrow \text{Herm}(2, \CC),$  
$x=x^a \mapsto \b x=\b x^{AA'}:=x^a \s_a^{AA'} = \tfrac{1}{2}\begin{psmallmatrix} x^0+x^3 & x^1 - ix^2 \\[0.5mm] x^1 + ix^2 & x^0 - x^3 \end{psmallmatrix}_{\vphantom{|}}^{\vphantom{|}} 
$.
Upper case Latin letters are Weyl spinor indices, with values $0$ and $1$. The spacetime interval is then given by $x^T \eta x=4 \det(\b x)$. 
The isomorphism for the dual is
$
\RR^{4*} \rarrow \text{Herm}(2, \CC)$, 
$r=x^t=x^T \eta  \mapsto \b r:=x^0\s_0- x^i \s_i = \tfrac{1}{2}\begin{psmallmatrix} x^0-x^3 & -x^1 + ix^2 \\[0.5mm] -x^1 - ix^2 & x^0 + x^3 \end{psmallmatrix}^{\vphantom{|}}_{\vphantom{|}}
$.

Correspondingly, we have the double cover group morphism $S\!L(2, \CC) \rarrow S\!O(1, 3)$, $\pm \b S \mapsto S$. It is a spin representation of the Lorentz group.  So the  action $S\!O(1, 3) \times \sM  \rarrow \sM$, $(S, x) \mapsto Sx$  preserving $\eta$,  is represented by the action $S\!L(2, \CC) \times \text{Herm}(2, \CC) \rarrow \text{Herm}(2, \CC)$, $(\b S, \b x) \mapsto \b S \b x \b S^*$ preserving~$\det$.
The associated Lie algebra isomorphism $\sl(2, \CC) \rarrow \so(1,3)$, $\b s \mapsto s$, implies that the action $\so(1, 3) \times \sM  \rarrow \sM$,  $(s, x) \mapsto sx$, is represented by the action $\sl(2, \CC) \times \text{Herm}(2, \CC) \rarrow \text{Herm}(2, \CC)$,  $(\b s, \b x) \mapsto \b s \b x + \b x \b s^*$.

\medskip

On a conformal $4$-manifold $\M$, a (local) twistor is a map 
$\psi : \U \subset \M \rarrow \CC^4$, $x \mapsto \psi(x)=\begin{psmallmatrix} \pi \\ \omega \end{psmallmatrix}$ with $\pi, \omega \in \CC^2$ dual Weyl spinors. The (generalised) twistor connection is 
$\b\varpi = \begin{psmallmatrix}  -\b A^* & -i \,\b P \\[0.5mm]
						  i \,\b e & \b A   
		\end{psmallmatrix}^{\vphantom{|}}_{\vphantom{|}}$, 
with $\b A \in \Omega^1\big(\U, \sl(2, \CC)\big)$ and $\b e, \b P \in \Omega^1\big(\U, \text{Herm}(2, \CC)\big)$.
 It enters the definition of the twistor derivative (or twistor transport), $\b D\defeq d+ \b\varpi$.
 The twistor curvature is 
 $\b\Omega= \begin{psmallmatrix}  -\b \sW^*+ \sfrac{f}{2}\1_2 & -i\, \b \sC \\[0.5mm]
						      i\, \b \sT & \b \sW   -\sfrac{f}{2}\1_2
	          	\end{psmallmatrix}^{\vphantom{|}}_{\vphantom{|}} $, 
with $\b \sW \in \Omega^2\big(\U, \sl(2, \CC)\big)$, $\b \sT, \b \sC \in \Omega^2\big(\U, \text{Herm}(2, \CC)\big)$, and $f \in \Omega^2\big(\U, \RR \big)$.
Here again, imposing $\text{Ricc}(\b\sW)=0$ and $\b\sT=0$ implies that $\b A$ is the spin connection, $f=0$,  and $\b P, \b\sC, \b\sW$ are the Schouten, Cotton and Weyl tensors. So that $\b\varpi$ is the standard twistor connection. 

\medskip

The (local) Weyl gauge transformations of these twistor variables are obtained via elements of type
\begin{align}
\label{twistor-cocycle}
\b C(z) = \begin{pmatrix} z^{\sfrac{1}{2}}    &  -i\, z^{-\sfrac{1}{2}}\, \b\Upsilon(z) \\[0.5mm]
 					   0  &     z^{-\sfrac{1}{2}}
	     \end{pmatrix}, \quad \text{where} \quad \b\Upsilon(z)=\Upsilon(z)_{AA'}  \in \text{Herm}(2, \CC),
\end{align}
  so that, 
 \begin{equation}
 \label{Twistor-Weyl-GT}
 \begin{split}
& \psi^z= \b C(z)\- \psi, \qquad \b\varpi^z=\b C(z)\- \b\varpi \, \b C(z) + \b C(z)\-d\b C(z),  \\[1mm]
& (\b D\psi)^z=\b D^z\psi^z= \b C(z)\- \b D\psi, \qquad \b\Omega^z=\b C(z)\- \b\Omega\, \b C(z).
 \end{split}
 \end{equation}
One verifies again that given $z,z' \in \W$, one has $\b C(z)^{z'}= \b C(z')\- \b C(zz')$,  as instance of  \eqref{localGTCgamma} and a local version of \eqref{GTCgamma}. It follows that \eqref{Twistor-Weyl-GT} is a special case of \eqref{localGTconnection} %and \eqref{localGTtensorial} (or 
 -\eqref{localGTothers}, local version %of the  transformations of twisted gauge fields described in
  Proposition \ref{ConnectionGT} and \ref{TensorialGT}. 
 Twistor variables are therefore twisted gauge fields w.r.t. the Weyl gauge group $\W$. 
 \medskip

 The (local) spin gauge group is $\SL\defeq \left\{ \b\sS = \begin{psmallmatrix} {{\b S}}^{-1*} & 0  \\ 0 & \b S  \end{psmallmatrix},\, \b S:\U \rarrow S\!L(2,\CC)\, \big|\, \b\sS^{\b\sS'}={\b\sS}^{'-1}\b\sS \b\sS' \right\}$. The spin gauge transformations of the twistor variables are, 
 \begin{align*}
&\psi^{\b\sS}= {\b\sS}\- \psi, \qquad \b\varpi^{\b\sS}={\b\sS}\- \b\varpi\, \b\sS + {\b\sS}\-d\b\sS,  \\[1mm]
& (\b D\psi)^{\b\sS}=D^{\b\sS}\psi^{\b\sS}= {\b\sS}\- \b D\psi, \qquad \b\Omega^{\b\sS}=\b\sS\- \b\Omega\, \b\sS.
 \end{align*}
 so they are standard gauge fields w.r.t. $\SL$. Also,  one verifies that $\b C(z)^{\b\sS} = {\b\sS}\- \b C(z)\, \b\sS$,  a special case of  \eqref{local-K-GT-Cgamma} and a  local version of \eqref{K-GT-Cgamma}. So the actions of $\W$ and $\SO$ on the twistor variables commute and  we have, 
  \begin{align*}
& \psi^{\,z\b\sS}= [\b C(z)\b\sS]\- \psi, \qquad {\b\varpi}^{\,z\b\sS}=[\b C(z)\b\sS]\- \b\varpi \, [\b C(z)\b\sS] + [\b C(z)\b\sS]\-d[\b C(z)\b\sS],  \\[1mm]
& (\b D\psi)^{\,z\b\sS}={\b D}^{\,z\b\sS}\psi^{\,z\b\sS}= [\b C(z)\b\sS]\- \b D\psi, \qquad {\b\Omega}^{\,z\b\sS}=[\b C(z)\b\sS]\- \b\Omega\, [\b C(z)\b\sS],
 \end{align*}
 as special case of  \eqref{local-active-mixedGT-connection}-\eqref{activeGT-mixed-others} and local version of Proposition  \ref{mixedGTconnection}. Thus, like tractors, twistor variables are mixed gauge fields w.r.t. the gauge group $\W \times \SL$.
  
 \bigskip
 Again, we can guess the global twisted geometry from the local data.
Consider the bundle $\P\big(\M, \,W \times S\!L(2,\CC)\big)$,
and a cocycle  $C: \P \times W \rarrow \b G$ where 
$\b G= \left\{ \begin{psmallmatrix} \1 & -i\,\b r \\  0 & \1  \end{psmallmatrix} \rtimes \begin{psmallmatrix} z^{\sfrac{1}{2}} & 0 \\  0 & z^{-\sfrac{1}{2}}  \end{psmallmatrix} = \begin{psmallmatrix} z^{\sfrac{1}{2}} & -i\, z^{-\sfrac{1}{2}}\,\b r \\  0 & z^{-\sfrac{1}{2}} \end{psmallmatrix}\, \big| \, \b r\in \text{Herm}(2, \CC), z\in W  \right\}^{\vphantom{|}}_{\vphantom{|}}$.
For $z \in \W$, the local cocycle  \eqref{twistor-cocycle} is a map $\b C(z) : \U \rarrow \b G$.
 Consider also the semidirect product group
  $\b G \rtimes S\!L \defeq \left\{  \begin{psmallmatrix} z^{\sfrac{1}{2}} & -i\, z^{-\sfrac{1}{2}}\,\b r \\  0 & z^{-\sfrac{1}{2}} \end{psmallmatrix}\rtimes \begin{psmallmatrix} {\b S}^{-1*} & 0  \\  0 & \b S  \end{psmallmatrix}
  = \begin{psmallmatrix} z^{\sfrac{1}{2}}{\b S}^{-1*}  & -i\, z^{-\sfrac{1}{2}}\,\b r \b S \\  0 & z^{-\sfrac{1}{2}} \b S \end{psmallmatrix} \,\big|\, \b S\in S\!L(2,\CC) \right\}^{\vphantom{|}}_{\vphantom{|}}$, 
  where the group morphism $S\!L \rarrow \Aut(\b G)$ of the semidirect structure is $\b S \mapsto \text{Conj}(\b S)$. 
 Its Lie algebra is 
 Lie$(\b G \rtimes S\!L)= \left\{ \begin{psmallmatrix} -(\b s - \sfrac{\epsilon}{2}\1)^* & -i\, \b \iota  \\ 0 & \b s -\sfrac{\epsilon}{2}\1  \end{psmallmatrix} \, \big|\, \epsilon \in \RR^+_*,  s\in  \sl(2,\CC), \b\iota \in \text{Herm}(2, \CC)\right\}^{\vphantom{|}}_{\vphantom{|}}$. 
 The twistor connection manifestly takes value in a bigger Lie algebra, 
 Lie$\b G'=\text{Lie}(\b G\rtimes SL)\oplus \text{Herm}(2, \CC)=\left\{ \begin{psmallmatrix} -(\b s - \sfrac{\epsilon}{2}\1)^* & -i\, \b \iota  \\  i\, \b \tau & \b s -\sfrac{\epsilon}{2}\1  \end{psmallmatrix}\, \big|\, \b\tau \in \text{Herm}(2, \CC) \right\}^{\vphantom{|}}_{\vphantom{|}}$, 
 and dim Lie$\b G'\!/\text{Lie}(\b G\!\rtimes\! S\!L) =\text{dim}\M$. 
 Therefore, the twistor connection $\b\varpi$  is a local mixed Cartan connection.
Obviously,  like its real counterpart, Lie$\b G'$ is graded: 
 Lie$\b G'= \b \LieG'_{-1}+\b\LieG'_0+\b\LieG'_1=\left\{  \begin{psmallmatrix} 0 & 0  \\ -i\, \b\tau & 0 \end{psmallmatrix}  + \begin{psmallmatrix} -(\b s - \sfrac{\epsilon}{2}\1)^* & 0  \\ 0 & \b s - \sfrac{\epsilon}{2}\1 \end{psmallmatrix}  +  \begin{psmallmatrix} 0 & -i\, \b\iota \\ 0 & 0 \end{psmallmatrix} \, \big|\, \ldots \right\}^{\vphantom{|}}_{\vphantom{|}}$, $[\b \LieG'_i, \b \LieG'_j] \in \b\LieG'_{i+j}$, with Lie$(\b G\rtimes S\!L)=\b\LieG'_0+\b\LieG'_1$ a parabolic subalgebra.
So, twistor geometry is an instance of parabolic mixed Cartan geometry, and  the twistor bundle is a mixed  vector bundle $\mathbb{T}=\P \times_{\b C(W)\rtimes S\!L} \CC^4$.
 
 Notice again that, as in the tractor case,  given \eqref{tractor-cocycle}, the $\P$-dependance of the cocycle map $\b C: \P \times W \rarrow \b G$  comes from the coefficients of the soldering form (${e^a}_\mu$) entering the definition of $\b\Upsilon(z)$, containing  a derivative of $z$. For constant $z \in W$ the cocycle then reduces to a group morphism $\b C: W \rarrow \b G$, $z\mapsto \b C(z)=\begin{psmallmatrix}z^{\sfrac{1}{2}}\1 & 0 \\ 0 & z^{-\sfrac{1}{2}}\1 \end{psmallmatrix}^{\vphantom{|}}_{\vphantom{|}}$. Again, at the level of $W\times S\!O$-equivariance the $W$-twisted side of twistor geometry  is degenerated, and only at the level of $\W \times \SO$-gauge transformations does the cocycle structure and the $\W$-twisted geometry are manifest.

\subsection{Conformal gravity as a twisted gauge theory}%%%%%%%%%%%%%%%%%%%%%%%%%%%%%%%%%%
\label{Conformal gravity as a twisted gauge theory}%%%%%%%%%%%%%%%%%%%%%%%%%%%%%%%%%%%%%

An attempt to interpret local twistors and the twistor covariant derivative as gauge fields in the spirit of Yang-Mills theory, compatible with the gauge principle of field theory, was first proposed in \cite{Merkulov1984_I}. This work is cited in the  reference text of Penrose and Rindler \cite{Penrose-Rindler-vol2}.\footnote{The footnote p.133 reads, ``\emph{Local twistors [...] can under certain circonstances be thought of as defining a kind of Yang-Mills theory, cf Merkulov (1984)(Bach tensor current)}.".} 
From the above considerations, it appears clearly that actually twistors better fit in the generalised geometry developed in this paper, and are therefore not Yang-Mills gauge fields, but rather twisted/mixed gauge fields - that indeed provide a new satisfying instantiation of the gauge principle.

It was further shown in \cite{Merkulov1984_I} that the Yang-Mills equation for the twistor connection $\b\varpi$ reproduces the Bach equation of conformal (or Weyl) gravity, and that upon examination, the Yang-Mills type Lagrangian for $\b\varpi$ is indeed the Weyl tensor-squared Lagrangian of conformal gravity. %The equivalence between the Yang-Mills equation for the tractor connection $\varpi$ and  Bach equation was first noticed in \cite{Korz-Lewand-2003}. 

This is most clearly understood in the context of Section \ref{Twisted gauge theories}.
First, define the  Killing forms $\b B$ and $B$ on Lie$\b G'$ and Lie$G'$. Given $\b M, \b N \in$ Lie$\b G'$,  $\b B(\b M, \b N):=\tfrac{1}{2}\left(  \Tr(\b M\b N) + \Tr(\b N^*\b M^*)\right)$. Given $M, N \in$ Lie$G'$, $B(M, N):=\Tr(MN)$. The same formulae hold for $\sl(2, \CC)$ and $\so(1, 3)$ and define $\b B_{\sl(2, \CC)}$ and $B_{\so(1, 3)}$, which must coincide since $\sl(2,\CC) \simeq \so(1, 3)$. As a matter of fact, for $m, n \mapsto \b m, \b n$ one has  $\b B_{\sl(2, \CC)}(\b m, \b n)=B_{\so(1, 3)}(m, n)$.
The Yang-Mills Lagrangians associated to the standard tractor  and twistor connections both reproduce conformal gravity,
\begin{align*}
\left.
\begin{array}{l}
L_\text{YM}(\varpi)= \tfrac{1}{2}B(\Omega,  *\Omega)=  \tfrac{1}{2}B_{\so(1, 3)}(\sW, *\sW) \\[1mm]  
L_\text{YM}(\b\varpi)=\tfrac{1}{4}\b B(\b\Omega,  * \b\Omega)=\tfrac{1}{2} \b B_{\sl(2, \CC)}(\b \sW, *\b \sW)
\end{array}
\right\}
= \tfrac{1}{2}\Tr(\sW \w *\sW)=L_\text{Weyl}(e).
\end{align*}
It is then no surprise that the field equations obtained  by varying the action $S_\text{\!YM}(\b\varpi)$ w.r.t. $\b\varpi$, or $S_\text{\!YM}(\varpi)$ w.r.t. $\varpi$, on the one hand, and by varying the action $S_\text{\!Weyl}(e)$ w.r.t. $e$ on the other hand, should coincide. In the first case we obtain the Yang-Mills equations for the standard tractor and twistor connections, and in the second case we obtain the Bach equation:

\begin{center}
\begin{tikzcd}[column sep=tiny, row sep=normal]

\frac{\delta S_\text{\!YM}(\b\varpi)}{\delta\b\varpi}=0 \ \rarrow \ \b D*\b\Omega=0 
\ar[rr,  start anchor={[xshift=3ex]},end anchor={[xshift=-3ex]}, leftrightarrow] \arrow[dr, leftrightarrow, bend right, start anchor={[xshift= -2ex, yshift=-1ex]}, end anchor={[xshift=-3ex, yshift=+2.5ex]}]	 &	  	 &
 \frac{\delta S_\text{\!YM}(\varpi)}{\delta\varpi}=0 \ \rarrow \ D*\Omega=0
 \arrow[dl, leftrightarrow, bend left,  start anchor={[xshift= +2ex, yshift=-1ex]}, end anchor={[xshift=+3ex, yshift=+2.5ex]}] \\[1mm]
                    	 &       \frac{\delta S_\text{\!Weyl}(e)}{\delta e}=0 \ \rarrow\  B_{ab}=0     &

\end{tikzcd}  
\end{center}
where $B_{ab}$ is the Bach tensor. The equivalence on the right side of the diagram was first noticed in \cite{Korz-Lewand-2003}. From this we conclude that conformal gravity is  a mixed $\W \times \SO$ -gauge theory hiding in plain sight.

% \subsection{The case of projective tractors (?)}%%%%%%%%%%%%%%%%%%%%%%%%%%%%%%%%%
%\label{The case of projective tractors}%%%%%%%%%%%%%%%%%%%%%%%%%%%%%%%%%

\clearpage

\subsection{Application to anomalies in QFT} %%%%%%%%%%%%%%%%%%%%%%%%%%%%%%%%%%%%%%%%%%%%%
\label{Application to anomalies in QFT} %%%%%%%%%%%%%%%%%%%%%%%%%%%%%%%%%%%%%%%%%%%%%%%%

Anomalies arise in QFT when the quantization of a classical gauge theory fails to uphold the gauge invariance. First discovered through perturbative methods, (consistent) anomalies were found to be characterized by BRST cohomological methods and obtainable via Stora-Zumino descent equations. They finally came to be understood as degree $1$ elements in the cohomology of Lie$\H$ \cite{Bonora-Cotta-Ramusino}, and soon after as $\H$-1-cocycles \cite{Faddeev-Shatashvili1984, Reiman-et-al1984}, see in particular \cite{ Falqui-Reina1985}.
Interestingly, in \cite{Falqui-Reina1985} and \cite{Catenacci-et-al1986, Catenacci-Pirola1990},  an infinite dimensional  \emph{twisted line bundle} appears as a relevant object in the study of anomalies (see  also \cite{Mickelsson1986, Mickelsson1987}, or \cite{Ferreiro-Perez2018} more recently).
It comes out as follows. 
\medskip

Consider a Yang-Mills  gauge theory such  that the relevant fields space is the space $\A$  of Ehresmann connections of a $H$-principal bundle $\P$ with gauge group $\H$.
Under proper restrictions, $\A$ is itself a $\H$-principal bundle over the moduli space $\A/\H$, where one is here in the realm of  infinite dimensional Hilbert manifolds \cite{Singer, Cotta-Ramusino-Reina1984}.

A quantum (vacuum) functional is smooth map $W\!:\A\! \rarrow\CC^*$. For a gauge invariant quantized theory, $W$ is s.t.  $W(A^\gamma)\!=\!W(A)$, $\gamma\in\H$, it is therefore projectable and descends to a functional on the base $\A/\H$. But an anomalous functional is s.t. $W(A^\gamma)=C(A, \gamma)\- W(A)$, where $C\!:\A \times \H \rarrow U(1)$ is $C(A, \gamma)=\exp\{-i2\pi f(A, \gamma)\}$. The functional in the phase is the Wess-Zumino term, also called \emph{integrated anomaly} since indeed $\smash{\lim\limits_{\tau \rarrow 0} f(A, \gamma_\tau)/\tau =a(\chi ,A)}$ is the anomaly (linear in $\chi \in$ Lie$\H$).

 Now, consistency of the right action of the gauge group $W(A^{\gamma\gamma'})\!=\!W((A^\gamma)^{\gamma'})$ implies the cocycle relation $C(A, \gamma\gamma')=C(A, \gamma)\ C(A^\gamma\!, \gamma')$, which is a form of the Wess-Zumino consistency condition for the gauge anomaly. An anomalous quantum functional $W$ is then a \emph{twisted} $C$-equivariant function on $\A$, i.e. a section of the twisted associated line bundle $\L^C:=\A \times_{C(\H)} \CC^*$. 

\medskip

As far as I can tell,  the peculiar geometrical nature of  twisted line bundles like  $\L^C$  was first stressed in \cite{Blau1988, Blau1989}.\footnote{In the introduction of the latter reference we read, ``\emph{[...] recently objects (called generalized associated bundles hereafter) have appeared in the physics literature, about whose general structure little seems to be known}". And after defining the twisted line bundle with fiber $\CC$ we find, just below Eq.(2.3), the comment ``\emph{bundles of this kind have recently appeared in the physics literature (mainly in relation with anomalies). Their geometrical structure, however, was not further investigated.}" The present paper happens to  contribute to this investigation. }
These references also introduce a connection for a twisted line bundle $L^C:=\P\times_{C(H)} \CC^*$ associated to a $H$-principal bundle $\P$, with $C:\P\times H \rarrow \CC^*$, that turns out to be a special case of twisted connection defined in section \ref{Connection}. 
%a locally flat twisted connection, noted  $\Gamma$, built out of the cocycle $C$ and  a so-called (local) dressing field $u$.

Working  on a patch $\U\subset \M$, 
 %with a local section $\s:\U\rarrow \P$, one defines $C_\s:=\s^*C: \U\times H \rarrow \CC$
 consider %$C: \pi\-(\U)\times H \rarrow \CC$
  %and 
  a map $u:\pi\-(\U)\subset \P \rarrow H$ defined by the equivariance property $R^*_hu=h\-u$.\footnote{Such a map always exists on a trivial bundle, here the bundle $\P_{|\U}=\U \times H$.} This we call a - locally defined - dressing field \cite{Attard_et_al2017}.
The map $C(u):\pi\-(\U) \rarrow H$, is such that $C_{ph}(u(ph))=C_{ph}(h\-u(p))=C_{ph}(h\-)\,C_p(u(p))=C_p(h)\-C_p(u(p))$. So its equivariance is  $R^*_h C(u)=C(h)\- C(u)$ (call it a twisted dressing field, see \cite{Attard-Francois2016_I, Attard-Francois2016_II, Francois2019}). 
Built then the $1$-form $\Gamma:= C(u)dC(u)\-$ satisfying, for $X_p^v \in V_p\P$ generated by $X\in$ Lie$H$, 
\begin{align*}
\Gamma_p(X^v_p)&=C_p(u(p)) dC(u)\-_p(X_p^v)= C_p(u(p)) [X^v(C(u))]\-(p)=C_p(u(p)) \tfrac{d}{d\tau} C_{pe^{\tau X}} \left(u(pe^{\tau X})\right)\-\big|_{\tau=0}, \\
                              &= \tfrac{d}{d\tau} C_p \left(e^{\tau X}\right)\-\big|_{\tau=0} = dC_{p|e}(X),
\end{align*}
and, for $h\in H$ and $X_p \in T_p\P$,
\begin{align*}
R^*_h \Gamma_{ph}&=C_{ph}(u(ph)) dR^*_hC(u)_{|p}= C_p(h)\-C_p(u(p)) \, d \left(C(u)\-C(h) \right)_{|p}, \\
				&= C_p(h)\- \Gamma_p C_p(h) + C_p(h)\- dC(h)_{|p}.
\end{align*}
Thus, $\Gamma \in \C(\P_{|\U})^T$. Given a partition of unity $\{\delta_i\}$ subordinate to a covering $\{\U_i\}$ of $\M$,  the $\Gamma_i$'s can be glued into a non flat twisted connection $\Gamma=\sum (\pi^*\delta_i) \Gamma_i \in \C(\P)^T$, with curvature $\Omega=d\Gamma \in \Omega^2_\text{tens}(\P, C)$. It defines a covariant derivative $D=d\,+\Gamma$ on $\Omega^\bullet_\text{tens}(\P, C)$,  and on sections of $L^C$ in particular. 

Applied to the case $\P \Rightarrow \A$ and $L^C \Rightarrow \L^C$,we see that the twisted connection $\Gamma$ is given essentially by the Wess-Zumino term $f(A, \gamma)$.
\bigskip

\medskip

% \clearpage
 
\section{Conclusion} %%%%%%\%%%%%%%%%%%%%%%%%%%%%%%%%%

In this paper, we have constructed a geometry that generalises associated vector bundles $E$ built via representations $(\rho, V)$ of the structure group $H$ of a principal bundle $\P(\M, H)$. These generalised associated bundles $E^C$ are built from cocycles for the action of the structure group on the principal bundle, $C:\P\times H \rarrow G $, and representations $(\rho, V)$ of the target group $G$ instead of $H$. We therefore call these  - unimaginatively - \emph{twisted} \mbox{associated} bundles. We have also characterised the space of $V$-valued twisted tensorial forms $\Omega^\bullet_\text{tens}\big(\P, C(H) \big)$, whose subspace of degree $0$ is isomorphic with the space $\Gamma(E^C)$ of  sections of  twisted bundles. We have then defined a notion of twisted \mbox{connection} form on $\P$ that generalises Ehresmann connection $1$-forms in  providing a good exterior covariant derivative on $\Omega^\bullet_\text{tens}\big(\P, C(H) \big)$ - thus allowing to define  a parallel transport on $\Gamma(E^C)$ - and whose curvature belongs to $\Omega^2_\text{tens}\big(\P, C(H) \big)$. As usual, the gauge transformations of the twisted connections and tensorial forms are obtained from the action of  vertical automorphisms, $\Aut_v\big(\P,H\big)$, of the principal bundle. These  geometrical objects provide a new way to implement the gauge principle of physics, so that the local representatives on $\M$ can be seen as twisted gauge fields generalising Yang-Mills type gauge fields. The possibility of building twisted gauge theories straightforwardly ensue. 

As the most immediate extension of this new framework, we have considered the case of  bundles associated to, and tensorial forms on,   a principal bundle $\P(\M, H\times K)$, which behave as twisted objects w.r.t. the action of $H$, but as standard objects w.r.t. the action of $K$. For this reason we call them respectively \emph{mixed} associated bundles $\E^C$ and \emph{mixed} tensorial forms $\Omega^\bullet_\text{tens}\big(\P, C(H)\!\rtimes\! K \big)$. We then defined a corresponding notion of mixed connection, which is both a $H$-twisted connection and a $K$-standard (Ehresmann) connection. It induces a good covariant derivative on $\Omega^\bullet_\text{tens}\big(\P, C(H)\!\rtimes\! K \big)$ and $\Gamma(\E^C)$, and its curvature is mixed tensorial. 

In the same way that Cartan connection $1$-forms are a distinguished subclass of Ehresmann connection $1$-form, we have proposed a definition for a subclass of our twisted/mixed connections that may be a sensible generalisation of Cartan connections. We then call these twisted/mixed Cartan connections. 
 \medskip
 
 To convince ourselves that all this is not idle exploration, we have shown that conformal tractors and local twistors can be seen as simple and slightly degenerate instances of the general framework presented here. The tractor and twistor bundles are mixed vector bundles, while the tractor and twistor connections are mixed Cartan connections. This  clarifies and puts on firmer mathematical ground the attempt \cite{Merkulov1984_I} to interpret local twistors as gauge fields of a kind: they are indeed mixed gauge fields as defined here. We are then led to the surprising conclusion that conformal gravity is an unsuspected example of twisted/mixed  gauge theory.

At least one other example could have been added to the list: projective tractors. Like conformal tractors, these can be constructed bottom-up via prolongation \cite{Bailey-et-al94}, but they can also be more economically obtained via the dressing field method \cite{Attard_et_al2017}. In the latter case, the twisted structure is more readily and explicitly seen. We refrained from presenting this case because it would have added little to the discussion and we wanted to spare the reader an elaboration that might have felt repetitive. %{\color{gray} We may nevertheless give the full treatment of this case in a note at some futur time.}
\medskip

%On the  mathematical side, one might be interested in further developing and understanding the twisted/mixed geometry, so as to e.g. see how much of the standard notions such as holonomy, characteristic classes, Chern-Weyl theory... export to this new context. %and what is their meaning. 

%{\color{blue} Not sure about this: drop it ?}
%{\color{gray}
%%As a strictly mathematical possibility still, g
%Given that our proposal is conservative in starting with a principal bundle, one might want to examine if a well defined super-twisted/mixed geometry starting with a principal super-bundle could be defined consistently, and 
%generalise the framework developed here in the same way that differential super-geometry extends standard differential geometry (as we alluded to in the introduction).
%%(whose relevance to modern speculations in theoretical physics would need to be evaluated).
%
%An Atiyah Lie algebroid is associated to any principal bundle. The class of connections on algebroids defined in \cite{Masson-Lazz} contains Ehresmann connections as a subspace. On may wonder if  there is a natural definition of twisted connections on Atiyah Lie algebroids that contains twisted connections as a subspace. }

Keen readers will perhaps have noted that  twisted objects can be related to standard constructions by ``hiding" the cocycle structure. One can indeed build the twisted associated bundle $\Q=\P\times_{C(H)} G$, which is a $G$-principal bundle under the right action of $G$ on itself. One can show that the standard associated bundle $\Q\times_{G} V$ is isomorphic to the twisted bundle $\P\times_{C(H)} V$, and that a twisted connection on $\P$ induces an Ehresmann connection on $\Q$ (and vice-versa). 
In our view this interesting fact does not imply that the twisted geometry described here isn't worthy of further study. No more than the well-known fact that Lie$G$-valued Cartan connections on some $H$-principal bundle $\P$ induce Ehresmann connections on  the $G$-principal associated bundle $\Q=\P\times_H G$ (and vice-versa\footnote{Provided that the Ehresmann connection $\omega$ satisfies $\ker \omega \cap \phi_*T\P = \emptyset$, with the bundle map $\phi : \P \rarrow \Q$.}, see e.g.  \cite{Sharpe} Appendix A, $\S$3.) means that Cartan geometry isn't a worthy subject in its own right. 

One might therefore be interested in further developing and understanding the twisted/mixed geometry, so as to e.g. see how standard notions such as holonomy, characteristic classes, Chern-Weyl theory, etc... export to this new context. And further still, one may want to examine how this framework extends to the super-differential geometric setup, eyeing possible applications to physics and supergravity.

%{\color{gray} 
%Finally, the cocycle map that is the central object here is pervasive in \emph{ergodic theory}, which has to do with representations on Hilbert spaces. 
%Do we extend part of this to "curved" situation with our framework? Can any link can be done with quantum mechanics ?
%}
\medskip

The relevance of the twisted geometry is perhaps especially easy to argue for in view of the application to physics we have already hinted at :  Cocycles in the form of  Wess-Zumino terms appear most naturally in the study of anomalous quantum functionals on the $\H$-bundle $\A$ of Ehresmann connections of a $H$-principal bundle $\P$. These are C-equivariant functionals, i.e. section of a line twisted bundle associated to $\A$. The natural covariant differentiation of such functionals would require to endow $\A$ with a twisted connection, perhaps more general than the one  defined by the cocycle/Wess-Zumino term introduced in \cite{Blau1988, Blau1989}. 
The twisted geometry may also prove useful in relation to works on boundaries in gauge theories. For example, in \cite{Gomes-et-al2018} so-called \emph{field dependent gauge transformations} are introduced that could perhaps be better understood as $\H$-valued cocycles ($C\! : \A \times \H \rarrow \H$),  and~$\A$ is endowed with a connection that might be interpreted as a twisted connection: if their ``field dependent gauge transformation" $g$ is  indeed seen as a cocycle, compare equations (3.9a)-(3.10)/(3.9b) with (\ref{1st axiom})-(\ref{2nd axiom})/(\ref{inf-activeGT}).
 
\medskip

Finally, let us notice that a priori twisted/mixed gauge theories can be quantized following the same strategies used for standard gauge theories. There seems to be no objections to using path integration methods, and we have seen that the BRST framework is general enough to accommodate our new geometric objects. It is also conceivable to attempt canonical quantization, after all this is what Penrose proposed for twistors (see e.g. \emph{twistor quantization} on p.142 of \cite{Penrose-Rindler-vol2} and references therein). 

However a new layer of complexity may appear here. Indeed, the gauge transformations of our new gauge fields are twisted by  cocycles. Up until now, we did not articulate what is the equivalence relation on these cocycles and how is defined the associated cohomology. What, if anything, does this  cohomology add to the gauge structure of mixed gauge fields? How does it interact with the usual gauge-BRST cohomology familiar in gauge theory? How does it relate to the quantisation problem, and are there e.g. new kind of anomalies associated to it? These intriguing questions seem worthy of further investigation. 
\medskip

\section*{Acknowledgment}  %%%%%%%%%%%%%%%%%%%%%%%%%%%%%%%%%%%%%%%%%%

This work was supported by the Fonds de la Recherche Scientifique - FNRS under the grant PDR n0 T.0022.19.  The author thanks Matthias Blau (ITP, Bern University, CH) and Thierry Masson (CPT, Aix-Marseille University, FR) for pointing out the relation between twisted structures on $\P$ and standard constructions on $\Q$. He also thanks Matthias Blau for pointing out relevant contacts between elements of the literature on anomalies in QFT and the twisted geometry developed here.
\bigskip

{
%\Huge
%\huge
%\LARGE
%\Large
%\large
\normalsize %(default)
%\small
%\footnotesize
%\scriptsize
%\tiny
 \bibliography{Biblio1}

\begin{thebibliography}{10}

\bibitem{Rogers2007}
A.~Rogers.
\newblock {\em Supermanifolds}.
\newblock World Scientific Publishing Company, 2007.

\bibitem{Dubois-Violette_kerner_Madore1990a}
M.~Dubois‐Violette, R.~Kerner, and J.~Madore.
\newblock Noncommutative differential geometry of matrix algebras.
\newblock {\em Journal of Mathematical Physics}, 31(2):316--322, 1990.

\bibitem{Dubois-Violette-Kerner-Madore1990b}
M.~Dubois‐Violette, R.~Kerner, and J.~Madore.
\newblock Noncommutative differential geometry and new models of gauge theory.
\newblock {\em Journal of Mathematical Physics}, 31(2):323--330, 1990.

\bibitem{Masson-Lazz}
S.~Lazzarini and T.~Masson.
\newblock {Connections on Lie algebroids and on derivation-based noncommutative
  geometry}.
\newblock {\em J. Geom. Phys}, 62:387, 2012.

\bibitem{Masson-Lazz2013}
C.~Fournel, S.~Lazzarini, and T.~Masson.
\newblock {Formulation of gauge theories on transitive Lie algebroids}.
\newblock {\em J. Geom. Phys.}, 64:174--191, 2013.

\bibitem{Connes-Lott1991}
A.~Connes and J.~Lott.
\newblock Particle models and noncommutative geometry.
\newblock {\em Nuclear Physics B - Proceedings Supplements}, 18(2):29 -- 47,
  1991.

\bibitem{Connes-Marcolli}
A.~Connes and M.~Marcolli.
\newblock In M.~Khalkhali and M.~Marcolli, editors, {\em An Invitation to
  Noncommutative Geometry}, chapter A walk in the noncommutative garden, pages
  1--128. World Scientific Publishing Company, 2008.

\bibitem{Chamseddine-et-al2007}
A.~H. Chamseddine, A.~Connes, and M.~Marcoli.
\newblock Gravity and the standard model with neutrino.
\newblock {\em Adv. Theor. Math. Phys.}, 11(6):991--1089, 2007.

\bibitem{FLMasson}
J.~Fran\c{c}ois, S.~Lazzarini, and T.~Masson.
\newblock {\em {Mathematical Structures of the Universe}}, chapter Gauge field
  theories: various mathematical approaches.
\newblock {Copernicus Center Press, Krak\'ow, Poland}, 2014.

\bibitem{vanSuijlekom}
W.~D. van Suijlekom.
\newblock {\em Noncommutative geometry and particle physics}.
\newblock Springer, 2015.

\bibitem{Cartan-Eilenberg1956}
H.~Cartan and S.~Eilenberg.
\newblock {\em {Homological Algebra}}, volume~19 of {\em Princeton Landmarks in
  Mathematics}.
\newblock Princeton University Press, 1956.

\bibitem{Zimmer1982}
R.~J. Zimmer.
\newblock Ergodic theory, group representations, and rigidity.
\newblock {\em Bull. Amer. Math. Soc. (N.S.)}, 6(3):383--416, 05 1982.

\bibitem{Zimmer1984}
R.~J. Zimmer.
\newblock {\em Ergodic theory and Semisimple groups}, volume~81 of {\em
  Monographs in Mathematics}.
\newblock Birkh{\"a}user Basel, 1st edition, 1984.

\bibitem{Francois2018}
J.~Fran\c{c}ois.
\newblock {Artificial vs Substantial Gauge Symmetries: a Criterion and an
  Application to the Electroweak Model}.
\newblock {\em {Philosophy of Science}}, July 2018.

\bibitem{Dubois-Violette1987}
M.~Dubois-Violette.
\newblock The {W}eil-{BRS} algebra of a {L}ie algebra and the anomalous terms
  in gauge theory.
\newblock {\em J. Geom. Phys}, 3:525--565, 1987.

\bibitem{Bertlmann}
R.~A. Bertlmann.
\newblock {\em Anomalies In Quantum Field Theory}, volume~91 of {\em
  International Series of Monographs on Physics}.
\newblock Oxford University Press, 1996.

\bibitem{Chevalley-Eilenberg}
C.~Chevalley and S.~Eilenberg.
\newblock Cohomology theory of the lie groups and lie algebras.
\newblock {\em Transactions of the American Mathematical Society}, 63:85--124,
  1948.

\bibitem{Bonora-Cotta-Ramusino}
L.~Bonora and P.~Cotta-Ramusino.
\newblock Some remark on {BRS} transformations, anomalies and the cohomology of
  the {L}ie algebra of the group of gauge transformations.
\newblock {\em Commun. Math. Phys.}, 87:589--603, 1983.

\bibitem{DeAzc-Izq}
J.A.~De Azcarraga and J.~M. Izquierdo.
\newblock {\em Lie Groups, Lie Algebras, Cohomology and some Applications in
  Physics.}
\newblock Cambridge Monographs on Mathematical Physics. Cambridge University
  Press, 1995.

\bibitem{Ehresmann50}
C.~Ehresmann.
\newblock Les connexions infinit\'esimales dans un espace fibr\'e
  diff\'erentiable.
\newblock {\em Colloque de topologie de Bruxelles}, pages 29--55, 1950.

\bibitem{Kobayashi1957}
S.~Kobayashi.
\newblock Theory of connections.
\newblock {\em {Annali di Matematica Pura ed Applicata}}, 43(1):119--194,
  December 1957.

\bibitem{Sharpe}
R.~W. Sharpe.
\newblock {\em Differential Geometry: Cartan's Generalization of Klein's
  Erlangen Program}, volume 166 of {\em Graduate text in Mathematics}.
\newblock Springer, 1996.

\bibitem{Cap-Slovak09}
A.~Cap and J.~Slovak.
\newblock {\em Parabolic Geometries I: Background and General Theory}, volume~1
  of {\em Mathematical Surveys and Monographs}.
\newblock American Mathematical Society, 2009.

\bibitem{Bailey-et-al94}
T.N. Bailey, M.G. Eastwood, and A.R. Gover.
\newblock Thomas's structure bundle for conformal, projective and related
  structures.
\newblock {\em Rocky Mountain J. Math.}, 24(4):1191--1217, 12 1994.

\bibitem{curry_gover_2018}
S.~N. Curry and A.~R. Gover.
\newblock {\em {An Introduction to Conformal Geometry and Tractor Calculus,
  with a view to Applications in General Relativity}}, pages 86--170.
\newblock London Mathematical Society Lecture Note Series. Cambridge University
  Press, 2018.

\bibitem{Penrose-Rindler-vol1}
R.~Penrose and W.~Rindler.
\newblock {\em {Spinors and Space-Time}}, volume~1.
\newblock Cambridge University Press, 1984.

\bibitem{Penrose-Rindler-vol2}
R.~Penrose and W.~Rindler.
\newblock {\em {Spinors and Space-Time}}, volume~2.
\newblock Cambridge University Press, 1986.

\bibitem{Atiyah_et_al2017}
M.~Atiyah, M.~Dunajski, and L.~Mason.
\newblock Twistor theory at fifty: from contour integrals to twistor strings.
\newblock {\em Proceedings of the Royal Society of London A: Mathematical,
  Physical and Engineering Sciences}, 473(2206), 2017.

\bibitem{Bailey-Eastwood91}
M.~Eastwood and T.~Bailey.
\newblock Complex paraconformal manifolds - their differential geometry and
  twistor theory.
\newblock {\em Forum mathematicum}, 3(1):61--103, 1991.

\bibitem{Friedrich77}
Helmut Friedrich.
\newblock Twistor connection and normal conformal cartan connection.
\newblock {\em General Relativity and Gravitation}, 8(5):303--312, 1977.

\bibitem{Attard-Francois2016_I}
J.~Attard and J.~Fran\c{c}ois.
\newblock {Tractors and Twistors from conformal Cartan geometry: a gauge
  theoretic approach I. Tractors}.
\newblock {\em ADV THEOR MATH PHYS}, 22(8):1831--1883, 2018.

\bibitem{Attard-Francois2016_II}
J.~Attard and J.~Fran\c{c}ois.
\newblock {Tractors and Twistors from conformal Cartan geometry: a gauge
  theoretic approach II. Twistors}.
\newblock {\em Class. Quantum Grav.}, 34(8), March 2017.

\bibitem{Attard_et_al2017}
J.~Attard, J.~Fran\c{c}ois, S.~Lazzarini, and T.~Masson.
\newblock {\em {Foundations of Mathematics and Physics one Century After
  Hilbert : New Perspectives}}, chapter {The dressing field method of gauge
  symmetry reduction, a review with examples}.
\newblock Springer, 2018.

\bibitem{Francois2019}
J.~Fran\c{c}ois.
\newblock Dilaton from tractor and matter field from twistor.
\newblock {\em Journal of High Energy Physics}, 2019(6):18, June 2019.

\bibitem{Merkulov1984_I}
S.~A. Merkulov.
\newblock The twistor connection and gauge invariance principle.
\newblock {\em Comm. Math. Phys.}, 93(3):325--331, 1984.

\bibitem{Korz-Lewand-2003}
M.~Korzy{\'n}ski and J.~Lewandowski.
\newblock {The normal conformal Cartan connection and the Bach tensor}.
\newblock {\em Class. Quantum Grav.}, 20(16):3745, 2003.

\bibitem{Faddeev-Shatashvili1984}
L.~D. Faddeev and S.~L. Shatashvili.
\newblock Algebraic and hamiltonian methods in the theory of non-abelian
  anomalies.
\newblock {\em Theoretical and Mathematical Physics}, 60(2):770--778, 1984.

\bibitem{Reiman-et-al1984}
A.~G. Reiman, M.~A. Semenov-Tyan-Shanskii, and L.~D. Faddeev.
\newblock Quantum anomalies and cocycles on gauge groups.
\newblock {\em Functional Analysis and Its Applications}, 18(4):319--326, 1984.

\bibitem{Falqui-Reina1985}
G.~Falqui and C.~Reina.
\newblock {BRS cohomology and topological anomalies}.
\newblock {\em Communications in Mathematical Physics}, 102(3):503--515, 1985.

\bibitem{Catenacci-et-al1986}
R.~Catenacci, G.P. Pirola, M.~Martellini, and C.~Reina.
\newblock Group actions and anomalies in gauge theories.
\newblock {\em {Physics Letters B}}, 172(2):223 -- 226, 1986.

\bibitem{Catenacci-Pirola1990}
R.~Catenacci and G.~P. Pirola.
\newblock {A Geometrical description of local and global anomalies}.
\newblock {\em Lett. Math. Phys.}, 19:45--51, 1990.

\bibitem{Mickelsson1986}
J.~Mickelsson.
\newblock Strings on a group manifold, kac-moody groups, and anomaly
  cancellation.
\newblock {\em Phys. Rev. Lett.}, 57:2493--2495, Nov 1986.

\bibitem{Mickelsson1987}
J.~Mickelsson.
\newblock Kac-moody groups, topology of the dirac determinant bundle, and
  fermionization.
\newblock {\em {Communications in Mathematical Physics}}, 110(2):173--183,
  1987.

\bibitem{Ferreiro-Perez2018}
Roberto~Ferreiro P{\'e}rez.
\newblock On the geometrical interpretation of locality in anomaly
  cancellation.
\newblock {\em {Journal of Geometry and Physics}}, 133:102 -- 112, 2018.

\bibitem{Singer}
I.~M. Singer.
\newblock Some remark on the gribov ambiguity.
\newblock {\em Commun. Math. Phys.}, 60:7--12, 1978.

\bibitem{Cotta-Ramusino-Reina1984}
P.~Cotta~Ramusino and C.~Reina.
\newblock The action of the group of bundle-automorphisms on the space of
  connections and the geometry of gauge theories.
\newblock {\em Journal of Geometry and Physics}, 1(3):121 -- 155, 1984.

\bibitem{Blau1988}
Matthias Blau.
\newblock {Wess-Zumino terms and the geometry of the determinant line bundle}.
\newblock {\em {Physics Letters B}}, 209(4):503 -- 506, 1988.

\bibitem{Blau1989}
Matthias Blau.
\newblock Group cocycles, line bundles, and anomalies.
\newblock {\em {Journal of Mathematical Physics}}, 30(10):2226--2232, 1989.

\bibitem{Gomes-et-al2018}
H.~Gomes, F.~Hopfm{\"u}ller, and A.~Riello.
\newblock A unified geometric framework for boundary charges and dressings:
  Non-abelian theory and matter.
\newblock {\em Nuclear Physics B}, 941:249 -- 315, 2019.

\end{thebibliography}
}

\end{document}